\documentclass[1p]{elsarticle}
\usepackage{mathptmx}
\usepackage{amsmath}
\usepackage{amssymb}
\usepackage{color,subfig}
\usepackage{epic}
\usepackage{fancybox,graphpap}
\usepackage{bbm}
\usepackage{array}
\usepackage{algorithmic}
\usepackage{wrapfig}
\usepackage{ifpdf}

\ifpdf
\usepackage[all,arc,curve,color,frame]{xy}
\usepackage{pb-diagram,pb-xy}
\else
\usepackage[all,dvips,arc,curve,color,frame]{xy}
\usepackage{pb-diagram,pb-xy}
\fi

\begin{document}

\begin{frontmatter}

\title{   Modelling Concurrency with Comtraces \\
and Generalized Comtraces}

\author{Ryszard Janicki\corref{c1}}
\ead{janicki@mcmaster.ca}
\address{Department of Computing and Software,
McMaster University, Hamilton, ON, L8S 4K1 Canada}
\cortext[c1]{Corresponding author}

\author{Dai Tri Man L\^{e}}
\ead{ledt@cs.toronto.edu}
\address{Department of Computer Science,
University of Toronto,
Toronto, ON, M5S 3G4 Canada}

\begin{abstract}
{\em Comtraces} (\emph{com}bined \emph{traces}) are extensions of Mazurkiewicz traces that can model the ``not later than'' relationship. In this paper, we first introduce the novel notion  of {\em generalized comtraces}, extensions of comtraces that can additionally model the ``non-simultaneously'' relationship. Then we study some basic algebraic properties and canonical reprentations of comtraces and generalized comtraces. Finally we analyze the relationship between generalized comtraces and generalized stratified order structures. The major technical contribution of this paper is a proof showing that generalized comtraces can be represented by generalized stratified order structures.
\end{abstract}


\end{frontmatter}

\newcommand{\END}{\qed}
\newcommand{\LT}{\mathcal{L}}
\newcommand{\eqa}{\thickapprox}
\newcommand{\eqb}{\equiv}
\newcommand{\PS}[1]{\wp(#1)}
\newcommand{\PSB}[1]{\wp^{\setminus \set{\emptyset}}(#1)}

\newcommand{\h}[1]{\overline{#1}}
\newcommand{\set}[1]{\{#1\}}
\newcommand{\bset}[1]{\bigl\{#1\bigr\}}
\newcommand{\Bset}[1]{\Bigl\{#1\Bigr\}}
\newcommand{\Al}{\biguplus}
\newcommand{\RC}{\div_R}
\newcommand{\LC}{\div_L}
\newcommand{\wei}{\mathit{weight}}
\newcommand{\len}{\mathit{length}}
\newcommand{\LG}[1]{\langle #1 \rangle}
\newcommand{\E}{\mathbb{S}}
\newcommand{\EC}{\widehat{\mathcal{S}}}
\newcommand{\PO}{\prec}
\newcommand{\st}{\mathsf{st}}
\newcommand{\com}{<\!\!>}
\newcommand{\iffdf}{\stackrel{\textit{\scriptsize{df}}}{\iff}\ }
\newcommand{\df}{\triangleq}
\newcommand{\calf}[1]{\mathcal{#1}}
\newcommand{\sq}{\sqsubset}
\newcommand{\ccl}{\;\bowtie}
\newcommand{\todo}[1]{ \textcolor{red}{TODO: #1}}
\newcommand{\tcomment}[1]{\text{\hspace*{2mm}$\langle$~\parbox[t]{\textwidth}{ #1 $\rangle$}}}
\newcommand{\ttcomment}[1]{\text{$\langle$~#1~$\rangle$}}
\newcommand{\sym}[1]{{#1}^{\mathsf{sym}\;}}
\newcommand{\mins}{\mathsf{mins}}
\newcommand{\quotient}[2]{{#1}/\!{#2}}

\newcommand{\lhdf}{\lhd^\frown}
\newcommand{\flhd}{\frown_\lhd}
\newcommand{\lex}[1]{\,{#1}^{\textit{lex}}\,}
\newcommand{\stor}[1]{\,{#1}^{\textit{st}}\,}
\newcommand{\reco}[1]{ {#1}^{\,{\textsf{C}}} }
\newcommand{\si}[1]{(#1)^\Cap}
\newcommand{\CT}[1]{\mathfrak{C}(#1)}
\newcommand{\GCT}[1]{\mathfrak{gC}(#1)}

\newcommand{\It}[1]{\mathit{#1}}
\newcommand{\defref}[1]{Definition~\ref{def:#1}}
\newcommand{\theoref}[1]{Theorem~\ref{theo:#1}}
\newcommand{\propref}[1]{Proposition \ref{prop:#1}}
\newcommand{\lemref}[1]{Lemma~\ref{lem:#1}}
\newcommand{\exref}[1]{Example~\ref{ex:#1}}
\newcommand{\reref}[1]{Remark~\ref{re:#1}}
\newcommand{\eref}[1]{\eqref{eq:#1}}
\newcommand{\figref}[1]{Figure~\ref{fig:#1}}

\newcommand{\secref}[1]{Section~\ref{sec:#1}}
\newcommand{\eqnref}[1]{Eq.~(\ref{eq:#1})}
\newcommand{\TCT}{\stackrel{\mathsf{t}\leftrightsquigarrow \mathsf{c}}{\equiv}}

\newcommand{\kwadrat}{\hfill $\Box$}
\newcommand{\mEND}{\vspace{-5mm}\kwadrat}

\newcommand{\EOD}{\hfill {\scriptsize $\blacksquare$}}
\newcommand{\mEOD}{\vspace{-5mm}\EOD}


\newtheorem{theorem}{Theorem}
\newtheorem{lemma}{Lemma}
\newtheorem{corollary}{Corollary}
\newtheorem{conjecture}{Conjecture}
\newtheorem{proposition}{Proposition}

\newdefinition{definition}{Definition }
\newdefinition{example}{Example }
\newdefinition{remark}{Remark }
\newdefinition{algo}{Algorithm }
\newproof{proof}{Proof}

\numberwithin{equation}{section}

\tableofcontents

\section{Introduction}
Mazurkiewicz traces, or just traces\footnote{The word ``trace'' has many different meanings in Computer Science and Software Engineering. In this paper, we reserve the word ``trace'' for \emph{Mazurkiewicz trace}, which is different from ``traces'' used in Hoare's CSP \cite{Ho}.}, are quotient monoids over sequences (or words) \cite{Foa,Ma1,Di}. The theory of traces has been utilized to tackle problems from    diverse areas including combinatorics, graph theory, algebra, logic and especially concurrency theory \cite{Di}.

As a language representation of finite partial orders, traces can sufficiently model ``true concurrency" in various aspects of concurrency theory. However, some aspects of concurrency cannot be adequately modelled by partial orders (cf. \cite{J0,J4}), and thus cannot be modelled by traces. For example, neither traces nor partial orders can model the ``not later than" relationship \cite{J4}. If an event $a$ is performed ``not later than" an event $b$, then this ``not later than" relationship can be modelled by the following set of two step sequences ${\bf x}=\{\{a\}\{b\},\{a,b\}\}$; where \textit{step} $\{a,b\}$ denotes the simultaneous execution of $a$ and $b$ and the step sequence $\{a\}\{b\}$ denotes the execution of $a$ followed by $b$. But the set ${\bf x}$ cannot be represented by any trace (or equivalently any partial order), even if the generators, i.e. elements of the trace alphabet, are sets and the underlying monoid is the monoid of step sequences (as in \cite{Vog}).

To overcome these limitations, Janicki and Koutny proposed the {\em comtrace} ({\em com}bined {\em trace}) notion \cite{JK95}. First the set of all possible steps that generates step sequences are identified by a relation $sim$, which is called {\em simultaneity}. Second a congruence relation is determined by a relation $ser$, which is called {\em serializability} and is in general \emph{not} symmetric. Then a comtrace is defined as a finite set of congruent step sequences. Comtraces were invented to provide a formal linguistic counterpart of {\em stratified order structures} (\emph{so-structures}), analogously to how traces relate to partial orders.  

A so-structure \cite{GP,JK0,JK95,JK} is a triple $(X,\prec,\sqsubset)$, where $\prec$ and $\sqsubset$ are binary relations on the set $X$. So-structures were invented to model both the ``earlier than" (the relation $\prec$) and the ``not later than" (the relation $\sqsubset$) relationships, under the assumption that all system runs are modelled by stratified partial orders, i.e., step sequences. They have been successfully applied to model inhibitor and priority systems, asynchronous races, synthesis problems, etc. (see for example \cite{JK95,PK,JL08,JL08b,KK,KK2}).

The paper \cite{JK95} contains a major result showing that every comtrace uniquely determines a labeled so-structure, and then use comtraces to provide a semantics of Petri nets with inhibitor arcs. However, so far comtraces are used   less often than so-structures, even though in many cases they appear to be more natural than so-structures. Perhaps this is due to the lack of a sufficiently developed quotient monoid theory for comtraces similar to that of traces.

However, neither comtraces nor so-structures are enough to model the ``non-simultaneously" relationship, which could be defined by the set of step sequences $\{\{a\}\{b\},\{b\}\{a\}\}$ with the additional assumption that the step $\{a,b\}$ is not allowed.
In fact,
both comtraces and so-structures can adequately model concurrent histories only when paradigm $\pi_3$ of \cite{J4,JK} is satisfied. Intuitively, paradigm $\pi_3$  formalizes the class of concurrent histories satisfying the condition that  if both $\{a\}\{b\}$ and $\{b\}\{a\}$ belong to the concurrent history, then so does $\{a,b\}$ (i.e., these three step sequences $\{a\}\{b\}$, $\{b\}\{a\}$ and $\{a,b\}$ are all equivalent observations).

 To model the general case that includes the ``non-simultaneously" relationship, we need the concept of {\em generalized stratified order structures} (\emph{gso-structures}), which were introduced and analyzed by Guo and Janicki in \cite{GJ,J0}. A gso-structure is a triple $\left(X,\com,\sqsubset\right)$, where $\com$ and $\sqsubset$ are binary relations on $X$ modelling the ``non-simultaneously" and ``not later than" relationships respectively, under the assumption that all system runs are modelled by stratified partial orders.

To provide the reader with a high level view of  the main motivation and intuition behind the use of so-structures as well as the need of gso-structures, we will consider a motivating example (adapted from \cite{J0}).

\subsection{A motivating example}
We will illustrate our basic concepts and constructions by analyzing four simple concurrent programs. Three of these programs will involve the concepts of simultaneous executions, which is essential to our model. We would like to point out that the theory presented in this paper is especially a   for models where simultaneity is well justified, for example for the models with a discrete time.  

All four programs in this example are written using a mixture of {\em cobegin, coend} and a version of {\em concurrent  guarded commands}.
\newpage
\begin{example} \mbox{}\\

\begin{small}{\tt
P1: begin int x,y;\\
\indent \hspace{7mm}\indent    a: begin x:=0; y:=0 end;\\
\indent \hspace{7mm}\indent       cobegin b: x:=x+1, c: y:=y+1 coend\\
\indent \hspace{8mm}        end P1.\\

P2: begin int x,y;\\
\indent \hspace{7mm}\indent    a: begin x:=0; y:=0 end;\\
\indent \hspace{7mm}\indent    cobegin  b: x=0 $\rightarrow$ y:=y+1, c: x:=x+1   coend\\
\indent \hspace{8mm}        end P2. \\

P3: begin int x,y;\\
\indent \hspace{7mm}\indent    a: begin x:=0; y:=0 end;\\
\indent \hspace{7mm}\indent    cobegin  b: y=0 $\rightarrow$ x:=x+1, c: x=0 $\rightarrow$ y:=y+1  coend\\
\indent \hspace{8mm}        end P3.\\

P4: begin int x;\\
\indent \hspace{7mm}\indent    a: x:=0;\\
\indent \hspace{7mm}\indent       cobegin b: x:=x+1, c: x:=x+2 coend\\
\indent \hspace{8mm}        end P4.
} \end{small}\\

Each program is a different composition of three events (actions) called $a$, $b$, and $c$ ($a_i$, $b_i$, $c_i$, $i=1,\ldots ,4$, to be exact, but a restriction to $a$, $b$, $c$ does not change the validity of the analysis below, while simplifying the notation). Transition systems modelling these programs are shown in \figref{motiv1}. \EOD
\label{ex:motiveprog}
\end{example}

\begin{figure}[!h]
\begin{minipage}[t]{0.25\textwidth}
\[\SelectTips{cm}{}
\xymatrix@-1pc{
		&*++[o][F-]{\bullet}\ar[d]^{a}&\\
		&*+++[o][F-]{}\ar[dl]_{b}\ar[dr]^{c}\ar[dd]^{\set{b,c}} &\\
*+++[o][F-]{}\ar[dr]_{c} &	&*+++[o][F-]{}\ar[dl]^{b}\\
		&*+++[o][F=]{}	&
}\]
\caption*{$A_{1}$}\end{minipage}
\begin{minipage}[t]{0.25\textwidth}
\[\SelectTips{cm}{}
\xymatrix@-1pc{
		&*++[o][F-]{\bullet}\ar[d]^{a}\\
		&*+++[o][F-]{}\ar[dl]_{b}\ar[dd]^{\set{b,c}} &\\
*+++[o][F-]{}\ar[dr]_{c} &	\\
		&*+++[o][F=]{}	
}\]
\caption*{$A_{2}$}\end{minipage}
\begin{minipage}[t]{0.2\textwidth}\[{
\SelectTips{cm}{}
\xymatrix@-1pc{
	&	*++[o][F-]{\bullet}\ar[d]^{a}\\
	&	*+++[o][F-]{}\ar[ddd]^{\set{b,c}} &\\
	&	 	\\ &\\
	&	*+++[o][F=]{}	
}
}\]
\caption*{$A_{3}$}\end{minipage}
\begin{minipage}[t]{0.25\textwidth}
\[\SelectTips{cm}{}
\xymatrix@-1pc{
		&*++[o][F-]{\bullet}\ar[d]^{a}&\\
		&*+++[o][F-]{}\ar[dl]_{b}\ar[dr]^{c} &\\
*+++[o][F-]{}\ar[dr]_{c} &	&*+++[o][F-]{}\ar[dl]^{b}\\
		&*+++[o][F=]{}	&
}\]
\caption*{$A_{4}$}\end{minipage}
\vspace{0.5cm}

\begin{footnotesize}
\begin{minipage}[t]{0.25\textwidth}
\centering
$\PO_{1}=\set{(a,b),(a,c)}$\\
$\sq_{1}=\set{(a,b),(a,c)}$\\
$\com_{1}=\sq_{1}\cup\sq_{1}^{-1}$\\
$obs(P_1)\asymp obs(A_{1})$\\
$\asymp\{\PO_{1}\}\asymp\{\PO_{1},\sq_1\}$\\
$\asymp\{\com_1,\sq_1\}$
\end{minipage}
\begin{minipage}[t]{0.25\textwidth}
\centering
$\PO_{2}=\set{(a,b),(a,c)}$\\
$\sq_{2}=\set{(a,b),(a,c),(b,c)}$\\
$\com_{2}=\PO_{2}\cup\PO_{2}^{-1}$\\
$obs(P_2)\asymp obs(A_{2})$\\
$\asymp\{\PO_{2},\sq_2\}$\\
$\asymp\{\com_2,\sq_2\}$
\end{minipage}
\begin{minipage}[t]{0.2\textwidth}
\centering
$\PO_{3}=\set{(a,b),(a,c)}$\\
$\sq_{3}=\set{(a,b),(a,c),$\\$(b,c),(c,b)}$\\
$\com_{3}=\PO_{3}\cup\PO_{3}^{-1}$\\
$obs(P_3)\asymp obs(A_{3})$\\
$\asymp\{\PO_{3},\sq_3\}$\\
$\asymp\{\com_3,\sq_3\}$
\end{minipage}
\begin{minipage}[t]{0.25\textwidth}
\centering
$\PO_{4}=\set{(a,b),(a,c)}$\\
$\sq_{4}=\set{(a,b),(a,c)}$\\
$\com_{4}=\set{(a,b),(b,a),$\\$(a,c),(c,a),(b,c),(c,b)}$\\
$obs(P_4)\asymp obs(A_{4})$\\
$\asymp\{\com_4,\sq_4\}$
\end{minipage}
\end{footnotesize}
\caption{\label{fig:motiv1} Examples of \emph{causality}, \emph{weak causality}, and \emph{commutativity}. Each program $P_{i}$ can be modelled by a labeled transition system (automaton) $A_i$. The step $\set{a,b}$ denotes the \emph{simultaneous} execution of $a$ and $b$.}
\end{figure}

Let $obs(P_i)$ denote the set of all program runs involving the actions $a,b,c$  that can be observed. Assume that simultaneous executions can be observed. In this simple case all runs (or observations) can be modelled by {\em step sequences} . Let us denote $o_1=\set{a}\set{b}\set{c}$, $o_2=\set{a}\set{c}\set{b}$, $o_3=\set{a}\set{b,c}$. Each $o_i$ can be equivalently seen as a stratified partial order $o_i=(\{a,b,c\},\stackrel{o_i}{\rightarrow})$ where:

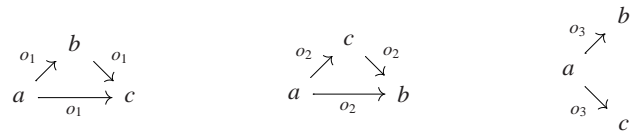
\begin{center}{\small
${
\SelectTips{cm}{10}
\divide\dgARROWLENGTH by2
\begin{diagram}
\node[2]{b}\arrow{se,t}{o_1}\\
\node{a}\arrow{ne,t}{o_1}\arrow[2]{e,b}{o_1}\node[2]{c}
\end{diagram}}$ \hspace{1.5cm}
${
\SelectTips{cm}{10}
\divide\dgARROWLENGTH by2
\begin{diagram}
\node[2]{c}\arrow{se,t}{o_2}\\
\node{a}\arrow{ne,t}{o_2}\arrow[2]{e,b}{o_2}\node[2]{b}
\end{diagram}}$ \hspace{1.5cm}
${
\SelectTips{cm}{10}
\divide\dgARROWLENGTH by2
\begin{diagram}
\node[2]{b}\\
\node{a}\arrow{ne,t}{o_3}\arrow{se,b}{o_3}\\
\node[2]{c}
\end{diagram}}$}
\end{center}

We can now write $obs(P_1)=\set{o_1,o_2,o_3}$, $obs(P_2)=\set{o_1,o_3}$, $obs(P_3)=\set{ o_3}$, $obs(P_4)=\set{ o_1, o_2}$. Note that for every $i=1,\ldots ,4$, all runs from the  set $obs(P_i)$ yield exactly the same outcome. Hence, each $obs(P_i)$ is called the {\em concurrent history} of $P_{i}$.

An abstract model of such an outcome is called a {\em concurrent
behavior}, and now we will discuss how causality, weak causality and commutativity relations are used to construct concurrent behavior.

\subsubsection{Program $P_{1}$}
In the set $obs(P_1)$, for each run,
$a$ always precedes both $b$ and $c$, and there is no {\em causal}
relationship between $b$ and $c$. This {\em causality} relation,
$\PO$, is the partial order defined as $\PO = \{(a,b),(a,c)\}$. In general
$\PO$ is defined by: $x\PO y$ iff for each run $o$ we have $x\stackrel{o}{\rightarrow} y$.
 Hence for $P_1$, $\PO$ is the intersection of $o_1$, $o_2$ and $o_3$, and $\set{o_1,o_2,o_3}$ is the set of all stratified extensions of the relation $\PO$.

Thus, in this case, the causality relation $\PO$ models the
concurrent behavior corresponding to the set of (equivalent) runs
$obs(P_1)$. We will say that $obs(P_1)$ and $\PO$ are {\em tantamount}\footnote{Following \cite{J0}, we are using the word ``tantamount" instead of ``equivalent" as the latter usually implies that the entities are of the same type, as ``equivalent automata", ``equivalent expressions", etc. Tantamount entities can be of different types.}
 and write
$obs(P_1)\asymp \set{\PO}$ or $obs(P_1)\asymp (\set{a,b,c},\PO)$. Having
$obs(P_1)$ one may construct $\PO$ (as an intersection of all orders from $obs(P_1)$), and then
reconstruct $obs(P_1)$ (as the set of all stratified extensions of $\PO$). This is a classical case
 of the ``true" concurrency approach, where concurrent behavior is modelled by a causality relation. \\

Before considering the remaining cases, note that the causality relation $\PO$ is
exactly the same in all four cases, i.e., $\PO_i \; = \{(a,b),(a,c)\}$, for $i=1,\ldots ,4$,
so we may omit the index $i$.

\subsubsection{Programs $P_{2}$ and $P_{3}$}
To deal with $obs(P_2)$ and $obs(P_3)$, $\PO$ is insufficient because $o_{2}\notin obs(P_2)$ and $o_{1},o_{2}\notin obs(P_2)$.
Thus, we need a weak causality relation $\sqsubset$  defined in this context as $x\sqsubset y$ iff for each run $o$ we have $\neg(y \stackrel{o}{\rightarrow}x)$ ($x$ {\em is never executed after} $y$).
For our four cases we have
$\sq_2=\{(a,b),(a,c),(b,c)\}$, $\sq_1=\sq_4=\PO$, and
$\sq_3=\{(a,b),(a,c),(b,c),(c,b)\}$.
Notice again that for $i=2,3$, the pair of relations $\{\PO,\sq_i\}$ and the set
$obs(P_i)$ are {\em tantamount} as each is definable from the other. (The set
$obs(P_i)$ can be defined as the greatest set $\mathit{PO}$ of partial orders
built from $a$, $b$ and $c$ satisfying
$x\PO y \Rightarrow \forall o \in \mathit{PO}.\;x\stackrel{o}{\rightarrow}y$ and
$x\sq_i y \Rightarrow \forall o \in \mathit{PO}.\; \neg(y\stackrel{o}{\rightarrow}x)$.)

Hence again in these cases ($i=2,3$) $obs(P_i)$ and $\{\PO,\sq_i\}$ are {\em tantamount}, $obs(P_i)\asymp\{\PO,\sq_i\}$, and so the pair $\{\PO,\sq_i\}$, $i=2,3$, models the concurrent behavior described by $obs(P_i)$. Note that $\sq_i$ alone is not sufficient, since (for instance) $obs(P_2)$ and $obs(P_2)\cup\{\{a,b,c\}\}$ define the same relation $\sq$.

\subsubsection{Program $P_{4}$}
The causality relation $\PO$ does not model the concurrent behavior of $P_{4}$ correctly\footnote{ Unless we assume that simultaneity is not allowed, or not observed, in which case $obs(P_1)=obs(P_4)=\set{o_1,o_2}$, $obs(P_2)= \{o_1\}$, $obs(P_3)=\emptyset$.} since $o_3$ does not belong to $obs(P_4)$. The commutativity relation $\com$ is defined in this context as $x\com y$ iff for each run $o$ either
$x\stackrel{o}{\rightarrow}y$ or $y\stackrel{o}{\rightarrow}x$. For
the set $obs(P_4)$, the relation $\com_4$ looks like $\com_4 =
\{(a,b),(b,a),(a,c),(c,a),(b,c),(c,b)\}$. The pair of relations
$\{\com_4,\PO\}$ and the set $obs(P_4)$ are {\em tantamount} as each is definable from the other. (The set $obs(P_4)$ is
the greatest set $\mathit{PO}$ of partial orders built from $a$, $b$ and $c$
satisfying $x\com_4 y \Rightarrow \forall o \in \mathit{PO}.\;
x\stackrel{o}{\rightarrow}y \vee y\stackrel{o}{\rightarrow}x$ and $
x\PO y \Rightarrow\forall o \in \mathit{PO}.\;x\stackrel{o}{\rightarrow}y.$)
In other words, $obs(P_4)$ and $\{\com_4,\PO\}$ are tantamount, so we may say that in this case the
relations $\{\com_4,\PO\}$ model the concurrent behavior described by
$obs(P_4)$.

Note that $\com_1\;=\; \PO \cup \PO^{-1}$ and the pair
  $\{\com_1,\PO\}$ also model the concurrent behavior
described by $obs(P_1)$.

\subsubsection{Summary of Analysis of $P_1, P_2, P_3$ and $P_{4}$}
For each $P_i$ the state transition model $A_{i}$ and their respective concurrent histories and concurrent behaviors are summarized in \figref{motiv1}. Thus, we can make the following observations:
\begin{enumerate}
\item $obs(P_1)$ can be modelled by the relation $\PO$ alone, and $obs(P_1)\asymp\{\PO\}$.
\item  $obs(P_i)$, for $i=1,2,3$ can also be modelled by the appropriate pairs of relations
    $\{\PO,\sq_i\}$, and $obs(P_i)\asymp\{\PO,\sq_i\}$.

\item all sets of observations $obs(P_i)$, for $i=1,2,3,4$ are modelled by the appropriate pairs of relations
    $\{\com_i,\sq_i\}$, and $obs(P_i)\asymp\{\com_i,\sq_i\}$.
\end{enumerate}

Note that the relation $\PO$ is not independent from the relations $\com$, $\sq$, since it can be proven (see \cite{J4}) that $\PO\: =\: \com\: \cap\: \sq$. Intuitively, since $\com$ and $\sq$ are the abstraction of the ``earlier than or later than'' and ``not later than'' relations, it follows that their intersection is the abstraction of the ``earlier than'' relation.

\subsubsection{Intuition for comtraces and generalized comtraces}
We may also try to model the concurrent behaviors of the programs $P_1$, $P_2$, $P_3$ and $P_4$ only in {\em terms of algebra of step sequences}. To do this we need to introduce an equivalence relation on step sequences such that the sets $obs(P_i)$, for $i=1,\ldots,4$, interpreted as {\em sets of step sequences} and not partial orders, are appropriate equivalence classes. A particular instance of this equivalence relation should depend on the structure of a particular program, or its labeled transition system representation.

It turns out that in such an approach the program $P_4$ needs to be treated differently than $P_1, P_2$ and $P_3$.
In order to avoid ambiguity, we will write $obs_{\rm step}(P_i)$ to denote the same set of system runs as $obs(P_i)$, but with runs now modelled by step sequences instead of partial orders.

For all four cases we need two relations $sim_i$ and $ser_i$, $i=1,\ldots,4$ on the set $\{a,b,c\}$. The relations $sim_i$, called {\em simultaneity}, are symmetric and indicate which actions can be executed simultaneously, i.e. in one step. It is easy to see that $sim_1=sim_2=sim_3 = \{(b,c),(c,b)\} $, but $sim_4=\emptyset$. The relations $ser_i$, called {\em serializability}, may not be symmetric, must satisfy $ser_i\subseteq sim_i$, and indicate how steps can equivalently be executed in some sequence. In principle if $(\alpha,\beta)\in ser$ then the step $\{\alpha,\beta\}$ is equivalent to the sequence $\{\alpha\}\{\beta\}$. For our four cases we have
$ser_1=sim_1 = \{(b,c),(c,b)\}$, $ser_2=\{(b,c)\}$, $ser_3=ser_4=\emptyset$.

Let $A,B,C$ be steps such that $A=B\cup C$ and $B\cap C=\emptyset$. For example $A=\{b,c\}$, $B=\{b\}$ and $C=\{c\}$.
We will say that the step $A$ and the step sequence $BC$ are {\em equivalent}, $A \eqa_i BC$, if $B\times C \subseteq sim_i$.
For example we have $\{b,c\}\eqa_i \{b\}\{c\}$ for $i=1,2$ and
$\{b,c\}\eqa_i \{c\}\{b\}$ for $i=1$. The relations $\eqa_3$ and $\eqa_4$ are empty.

Let $\equiv_i$ be the smallest equivalence relation on the whole set of events containing $ \eqa_i $, and for each step sequence $ A_1\ldots A_k $, let $[A_1\ldots A_k]_{\equiv_i}$ denote the equivalence class of $\equiv_i $ containing the step sequence $ A_1\ldots A_k $.

For our four cases, we have:
\begin{enumerate}
\item$ [\{a\}\{b\}\{c\}]_{\equiv_1}=\{\{a\}\{b\}\{c\},\{a\}\{c\}\{b\},\{a\}\{b,c\}\} = obs_{\rm step}(P_1)\asymp  obs(P_1)  $
\item$ [\{a\}\{b\}\{c\}]_{\equiv_2}=\{\{a\}\{b\}\{c\},\{a\}\{b,c\}\} = obs_{\rm step}(P_2)\asymp  obs(P_2)  $
\item$ [\{a\}\{b\}\{c\}]_{\equiv_3}=\{ \{a\}\{b,c\}\} = obs_{\rm step}(P_3) \asymp  obs(P_3)  $
\item$ [\{a\}\{b\}\{c\}]_{\equiv_4}=\{\{a\}\{b\}\{c\}\} \not= obs_{\rm step}(P_4) $

\end{enumerate}
Strictly speaking the statement $obs_{\rm step}(P_i)=obs(P_i)$ is false, but obviously $obs_{\rm step}(P_i) \asymp  obs(P_i) $, for $i=1,...,4$.

For $i=1,\ldots,3$, equivalence classes of each relation $\equiv_i$ are generated by relations $sim_i$ and $ser_i$. These equivalence classes  are called {\em comtraces} (introduced in \cite{JK95} as a generalization of Mazurkiewicz traces) and can be used to model concurrent histories of the systems or programs like $P_1,P_2$ and $P_3$.

In order to model the concurrent history of $P_4$ with equivalent step sequences, we need a third relation $inl_4$ on the set of events $\{a,b,c\}$ that is symmetric and satisfies $inl_4\cap sim_4=\emptyset$. The relation $inl_4$ is called {\em interleaving}, and if $(x,y)\in inl$ then events $x$ and $y$ cannot be executed simultaneously, but the execution of $x$ followed $y$ and the execution of $y$ followed by $x$ are equivalent. For program $P_4$ we have $inl_4 =\{(b,c),(c,b)\}$.

We can now define a relation $ \eqa'_4$ on step sequences of length two, as
$BC \eqa'_4 CB$ if $B\times C\subseteq inl$, which for this simple case gives $ \eqa'_4 = \bset{ (\{b\}\{c\},\{c\}\{b\}) ,
(\{c\}\{b\},\{b\}\{c\})}$. Let $\equiv_4$ be the smallest equivalence relation on the whole set of events containing $ \eqa_4 $ and $ \eqa'_4 $. Then we have
\[[\{a\}\{b\}\{c\}]_{\equiv_4}=\{\{a\}\{b\}\{c\},\{a\}\{c\}\{b\}\} = obs_{step}(P_4)\asymp  obs(P_4).\]

Equivalence classes of relations like $\equiv_4$, generated by the relations like $sim_4$, $ser_4$ {\em and} $inl_4$ are called {\em generalized comtraces} ({\em g-comtraces}, introduced in \cite{JL}) and they can be used to model concurrent histories of the systems or programs like $P_4$.

\subsection{Summary of contributions}
This paper is an expansion and revision of our results from \cite{JL,Le08}.  We propose a formal-language counterpart of gso-structures, called {\em generalized comtraces} (\emph{g-comtraces}). We will revisit and expand the algebraic theory of comtraces, especially various types of canonical forms and the formal relationship between traces and comtraces. We analyze in detail the properties of g-comtraces, their canonical representations, and most importantly the formal relationship between g-comtraces and gso-structures.

\subsection{Organization}
 The content of the paper is organized as follows. In the next section, we review some basic concepts of order theory and monoid theory.
\secref{eq} recalls the concept of Mazurkiewicz traces and discusses its relationship to finite partial orders. \secref{relsurvey} surveys some basic background on the relational structures model of concurrency \cite{GP,JK0,JK95,JK,GJ,J0}. 

Comtraces are defined
and their relationship to traces is discussed in \secref{comtraces},  and g-comtraces
are introduced in
\secref{gcomtraces}.

Various basic algebraic properties of both comtrace and g-comtrace congruences are discussed in \secref{algebraic}.
\secref{can} is devoted to canonical representations of traces, comtraces and g-comtraces. In \secref{com2strat} we recall some results on the so-structures defined by comtraces. The gso-structures generated by g-comtraces are defined and analyzed in \secref{gcom2strat}. Concluding remarks are made in \secref{conclusion}. We also include two appendices containing some long and  technical proofs of results from \secref{gcom2strat}.

\section{Orders, Monoids, Sequences and Step Sequences \label{sec:background}}
In this section, we recall some standard notations, definitions and results which are used extensively in this paper.

\subsection{Relations, orders and equivalences}

The \emph{powerset} of a set $X$ will be denoted by $\PS{X}$. The set of all {\em non-empty} subsets of $X$  will be denoted by $\PSB{X}$. In other words, 
$\PSB{X}\df\PS{X}\setminus\set{\emptyset}.$  

Let $f:A\rightarrow B$ be a function, then for every set $C\subseteq A$, we write $f[C]$ to denote the image of the set $C$ under $f$, i.e., $f[C] \df \set{f(x) \mid x\in C}$.

We let $\emph{id}_X$ denote the \emph{identity relation} on a set $X$. We write $R\circ S$ to denote the \emph{composition} of relations $R$ and $S$. We also write $R^{+}$ and $R^{*}$ to denote the \emph{(irreflexive) transitive closure} and \emph{reflexive transitive closure} of $R$ respectively.

A binary relation $R\subseteq X\times X$ is an \emph{equivalence relation} relation on $X$ iff it is reflexive, symmetric and transitive. If $R$ is an equivalence relation, we write $[x]_R$ to denote the equivalence class of $x$ with respect to $R$, and the set of all equivalence classes in $X$  is denoted as $X/R$ and called the \emph{quotient set} of $X$ by $R$.  We drop the subscript and write $[x]$ to denote the equivalence class of $x$ when $R$ is clear from the context. 

A binary relation $\prec \;\subseteq X \times X$ is a
{\em partial order} iff $R$ is {\em irreflexive} and {\em transitive}.
The pair $(X,\prec)$ in this case is called a \emph{partially ordered set} (\emph{poset}). The pair $(X,\prec)$ is called a \emph{finite poset} if $X$ is finite. For convenience, we define:
\begin{align*}
\simeq_\prec  &\df  \bset{(a,b)\in X\times X\mid a\not\prec b \;\wedge\; b\not\prec a}& \text{\emph{(incomparable)}}\\
\frown_\prec  &\df \bset{(a,b)\in X\times X\mid  a \simeq_\prec b \;\wedge\; a\neq b} &
\text{\emph{(distinctly incomparable)}}\\
\prec^\frown  &\df \bset{(a,b)\in X\times X\mid a \prec b \;\vee\; a\frown_\prec b} & \text{\emph{(not greater)}}
\end{align*}

A poset $(X,\prec)$ is {\em total} iff $\frown_\prec$ is empty; and {\em stratified} iff $\simeq_\prec$ is an equivalence relation. Evidently every total order is stratified.

Let $\prec_{1}$ and $\prec_{2}$ be partial orders on a set $X$. Then $\prec_2$ is an {\em extension} of $\prec_1$ if $\prec_1 \subseteq \prec_2$. The relation $\prec_2$ is a \emph{total extension} (\emph{stratified extension}) of $\prec_1$ if $\prec_2$ is total (stratified) and $\prec_1 \subseteq \prec_2$. \\

For a poset $(X,\prec)$, we define
\begin{align*}
\textit{Total}_{X}(\prec)&\df\set{\lhd\subseteq X\times X\mid \lhd\text{  is a total extension of} \prec}.
\end{align*}

\begin{theorem}[Szpilrajn \cite{Szp}]
For every poset $(X,\prec)$, $\prec = \bigcap_{\lhd\in\textit{Total}_{X}(\prec)}\lhd.$
\END
\end{theorem}

Szpilrajn's theorem states that every partial order can be uniquely reconstructed by taking the intersection of all of its total extensions.

\subsection{Monoids and equational monoids}
A triple $(X,\ast,\mathbbm{1})$, where $X$ is a set, $\ast$ is a total binary operation on $X$, and $\mathbbm{1}\in X$,
is called a {\em monoid}, if  $(a\ast b)\ast c = a\ast (b\ast c)$ and $a\ast \mathbbm{1} = \mathbbm{1}\ast a = a$,
for all $a,b,c \in X$. 

A equivalence relation $\sim \; \subseteq X\times X$ is a {\em congruence} in the monoid $(X,\ast,\mathbbm{1})$ if for all elements $a_1,a_2,b_1,b_2$ of $X$,
$a_1\sim b_1 \wedge a_2\sim b_2 \Rightarrow (a_1\ast a_2) \sim (b_1\ast b_2)$.

The triple $(X/\!\sim,\circledast,[\mathbbm{1}])$, where
$[a] \circledast [b] = [a\ast b]$, is called the {\em quotient monoid} of
$(X,\ast,1$) under the congruence $\sim$.
The mapping $\phi:X\rightarrow X/\!\sim$ defined as $\phi(a) = [a]$ is called the
{\em natural homomorphism} generated by the congruence $\sim$.
We usually omit the symbols $\ast$ and $\circledast$.

\begin{definition}[Equation monoid]
Given a monoid $M=(X,\ast,\mathbbm{1})$ and a finite set of {\em equations} $\It{EQ}=\{\;x_i=y_i\mid i=1,\ldots ,n\;\}$, define $\equiv_{\It{EQ}}$  to be the least congruence on $M$ satisfying 
\[x_i = y_i \implies x_i\equiv_{\It{EQ}} y_i\] for every equation $x_i = y_i \in \It{EQ}$. We call the relation $\equiv_{\It{EQ}}$  the congruence defined by the set of equation $\It{EQ}$, or  $\It{EQ}$-congruence.
The quotient monoid $M_{\equiv_{\It{EQ}}} \;=\; (X/\!\!\equiv_{\It{EQ}}, \circledast, [\mathbbm{1}])$, where $[x]\circledast[y] = [x \ast y]$, is called an {\em equational monoid}. \EOD
\end{definition}

The following folklore result shows that the relation $\equiv_{\It{EQ}}$ can also be \emph{uniquely} defined in an explicit way.
\begin{proposition}[cf. \cite{Le08}]
Given a monoid $M=(X,\ast,\mathbbm{1})$ and a set of equations $\It{EQ}$, define the relation $\approx \;\subseteq X \times X$ as:
\[x \approx y \iffdf \exists\; x_1, x_2 \in X.\; \exists\; (u=w) \in \It{EQ}.\; x=x_1\!\ast\! u \!\ast\! x_2 \wedge y=x_1\!\ast\! w\!\ast\! x_2,\]
then the $\It{EQ}$-congruence  $\equiv$  is $(\approx \cup \approx^{-1})^*$, the symmetric irreflexive transitive closure of $\eqa$.\qed
\label{prop:p1}
\end{proposition}

We will see later in this paper that monoids of traces, comtraces and generalized comtraces are all special cases of equational monoids.

\subsection{Sequences, step sequences and partial orders}
\label{s:sts-so}
By an {\em alphabet} we shall understand any finite set. For an alphabet $\Sigma$, let $\Sigma^*$ denote the set of all \emph{finite} sequences of elements (words) of $\Sigma$, $\lambda$ denotes the empty sequence, and any subset of $\Sigma^*$ is called a {\em language}. In the scope of this paper, we only deal with \emph{finite} sequences.  Let the operator $\_\cdot\_$ denote the sequence concatenation (usually omitted). Since the sequence concatenation operator is associative and $\lambda$ is neutral, the triple $(\Sigma^*,\cdot,\lambda)$  is a {\em monoid} (of sequences). \\

Consider an alphabet $\E \subseteq \PSB{X}$ for some alphabet $\Sigma$.
The elements of $\E$ are called {\em steps} and
the elements of $\E^*$ are called {\em step sequences}. For example if
$\E = \{ \{a,b,c\}, \{a,b\}, \{a\}, \{c\} \}$ then
$\{a,b\}\{c\}\{a,b,c\} \in \E^*$ is a step sequence.
The triple $(\E^*,\cdot,\lambda)$, is a {\em monoid} (of step sequences), since the step sequence concatenation is associative and $\lambda$ is neutral.\\

We will now show the formal relationship between step sequences and stratified orders. Let $t=A_1\ldots A_k$ be a step sequence in $\E^{*}$. We define  $|t|_a$, the number of occurrences of an event $a$ in $t$, as $|t|_a \df \bigl\lvert\bset{A_i\mid 1\le i \le k \wedge a\in A_i}\bigr\rvert$, where $|X|$ denotes the cardinality of the set $X$. 
\begin{itemize}
\item We can uniquely construct its \emph{enumerated step sequence} $\h{t}$ as
\begin{align*}
\h{t}\df \h{A_1}\ldots \h{A_k}\text{, where } \h{A_i} \df \Bset{e^{(|A_1\ldots A_{i-1}|_e+1)}\;\bigl\lvert\; e\in A_i\bigr.}.
\end{align*}
We call such $\alpha=e^{(i)}\in \h{A_i} $ an \emph{event occurrence} of $e$.  E.g., if $t= \set{a,b}\set{b,c}\set{c,a}\set{a}$,
then  $\overline{t}=\bset{a^{(1)},b^{(1)}}\bset{b^{(2)},c^{(1)}}\bset{a^{(2)},c^{(2)}}\bset{a^{(3)}}$ is its enumerated step sequence. 

\item Let $\Sigma_t=\bigcup_{i=1}^k \h{A_i}$ denote the set of all event occurrences in all steps of $t$. For example, when
$t= \set{a,b}\set{b,c}\set{c,a}\set{a}$,  we have $\Sigma_t=\bset{ a^{(1)},a^{(2)},a^{(3)}, b^{(1)}, b^{(2)},
c^{(1)},c^{(2)} }.$

\item Define $l:\Sigma_{t}\rightarrow \Sigma$ to be the function that returns the label of an even ocurrrence. In other words, for each event occurrence $\alpha=e^{(i)}$, $l(\alpha)$ returns the \emph{label} $e$ of $\alpha$. From an enumerated step sequence $\h{t} = \h{A_1}\ldots \h{A_k}$, we can uniquely recover its step sequence as $t =  l[\,\h{A_1}\,]\ldots l[\,\h{A_k}\,].$

\item For each $\alpha \in \Sigma_t$, let $pos_t(\alpha)$ denote the index number of the step where $\alpha$ occurs, i.e., if $\alpha\in \h{A_j}$ then $pos_t(\alpha)=j$. For our example, $pos_t(a^{(2)})=3$,  $pos_t(b^{(2)})=2$, etc.\\
\end{itemize}

Given a step sequence $u$, we define two relations $\lhd_u,\simeq_u\subseteq\Sigma_u\times\Sigma_{u}$ as:
\begin{align*}
\alpha \lhd_u \beta \iffdf pos_u(\alpha)<pos_u(\beta)
\text{\mbox{\hspace{5mm}} and \mbox{\hspace{5mm}} }
\alpha\simeq_u \beta \iffdf pos_u(\alpha)=pos_u(\beta).
\end{align*}
Since $\lhd_u^\frown$ is the union of $\lhd_u$ and $\frown_u$, we have 
\[\alpha \lhd_u^\frown \beta \iff (\alpha\not=\beta\wedge pos_u(\alpha)\le pos_u(\beta)).\]
The two propositions below are folklore results (see \cite{Le08} for detailed proofs), which are fundamental for understanding why stratified partial orders and step sequences are two interchangeable concepts. The first proposition shows that $\lhd_u$ is indeed a stratified order.

\begin{proposition}
Given a step sequence $u$, the relation $\simeq_u$ is an equivalence  relation and $\lhd_u$ is a stratified order. \qed
\label{corSSOS}
\end{proposition}

We will call  $\lhd_u$ the stratified order {\em generated by the step sequence} $u$. Conversely, let $\lhd$ be a stratified order on a set $\Sigma$. Then the second proposition says:
\begin{proposition}
If $\lhd$ is a stratified order on a set $\Sigma$ and $A,B$ are two distinct equivalence classes of $\simeq_{\lhd}$, then either
 $A\times B \subseteq \lhd$ or $B\times A \subseteq \lhd$. \qed
\label{OrderOnEC}
\end{proposition}

In other words, Proposition \ref{OrderOnEC} implies that if  we define a binary relation $\widehat{\lhd}$ on the quotient set $\quotient{\Sigma}{\simeq_{\lhd}}$ as  \[A\; \widehat{\lhd}\; B \iffdf A\times B \subseteq \lhd,\]
then $\widehat{\lhd}$ totally orders $\quotient{\Sigma}{\simeq_{\lhd}}$ into a sequence of equivalence classes $\Omega_\lhd=B_1\ldots B_k$ ($k\ge 0$). We will call the sequence $\Omega_\lhd$ as the step sequence \emph{representing} $\lhd$.  \\

Since sequences are a special case of step sequences and total orders are a special case of stratified orders, the above results can be applied to sequences and finite total orders as well. Hence, for each sequence $x\in\Sigma^*$, we let $\lhd_x$ denote the {\em total order generated} by $x$, and for every total order $\lhd$, we let $\Omega_\lhd$ denote the {\em sequence generating} $\lhd$. Furthermore, $\Sigma_x$ will denote the alphabet of the sequence $x$.

\section{Traces vs. Partial Orders \label{sec:eq}}
Traces or partially commutative monoids \cite{Foa,Di,Ma1,Ma2} are {\em equational monoids over sequences}.
In the previous section we have shown how sequences correspond to finite total orders and how step sequences correspond to finite stratified orders. In this section we discuss the relationship between traces and finite partial orders.

The theory of traces has been utilized to tackle problems from    diverse areas including combinatorics, graph theory, algebra, logic and, especially (due to the relationship to partial orders) concurrency theory \cite{Di,Ma1,Ma2}.

Since traces constitute a {\em sequence representation of partial orders}, they can effectively model ``true concurrency" in various aspects of concurrency theory using simple and intuitive means. We will now recall the definition of a \emph{trace monoid}.

\begin{definition}[\cite{Di,Ma2}]
Let $M=(E^*,\ast,\lambda)$ be a {\em monoid} generated by finite $E$, and let the relation $ind \subseteq E\times E$ be an
irreflexive and symmetric relation (called {\em independency} or {\em commutation}), and
$\It{EQ} \df \{ ab=ba \mid (a,b)\in ind \}.$ Let $\equiv_{ind}$, called \emph{trace congruence}, be the congruence defined by $\It{EQ}$. Then the equational monoid $M_{\equiv_{ind}} = \bigl( E^*/\!\!\equiv_{ind}, \circledast, [\lambda]\bigr)$ is a  monoid of {\em traces} (or a {\em free partially commutative monoid}). The pair $(E,ind)$ is called a {\em trace alphabet}.
\EOD
\end{definition}

We will omit the subscript $ind$ from trace congruence and write $\equiv$ if it causes no ambiguity.

\begin{example}
Let $E=\{a,b,c\}$,
$ind=\{(b,c),(c,b)\}$, i.e., $\It{EQ}=\{\;bc=cb\;\}$\footnote{Strictly speaking $\It{EQ}=\{\;bc=cb,cb=bc\;\}$ but standardly we consider the equations $bc=cb$ and $cb=bc$ as identical.}. 
For example, $abcbca\equiv accbba$ (since $abcbca\approx acbbca \approx acbcba \approx accbba$). Also we have $\mathbf{t}_1=[abcbca]=\\ \{abcbca,abccba,acbbca,acbcba, abbcca,accbba\}$,  $\mathbf{t}_2 =[abc]=\{abc,acb\}$ and $\mathbf{t}_3=[bca]=\{bca,cba\}$
 are traces. Note that  $\mathbf{t}_1 = \mathbf{t}_2\; \circledast\; \mathbf{t}_3$ since $[abcbca]=[abc]\circledast[bca]$. \EOD
\label{ex:traces1}
\end{example}

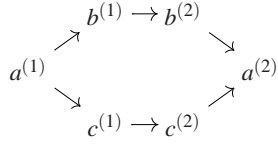
\begin{figure}[t]
  \begin{center}
	\begin{small}$
\SelectTips{cm}{10}
\xymatrix@C=1em@R=0.7em{
		&b^{(1)} \ar[r]  &b^{(2)} \ar[dr] & \\
a^{(1)} \ar[ur]\ar[dr]  &		 &		&a^{(2)}\\
		&c^{(1)} \ar[r]  &c^{(2)} \ar[ur] &
}$\end{small} 

  \end{center}
  \caption{Partial order generated by the trace $[abcbca]$}
\label{fig:trace2po}
\end{figure}

Each trace can be interpreted as a finite partial order. Let $\mathbf{t}=\{x_1,\ldots ,x_k\}$ be a trace, and let $\lhd_{x_i}$ denotes the total order induced by the sequence $x_i$, $i=1,\ldots ,k$.
Note that $\Sigma_{x_i} = \Sigma_{x_j}$ for all $i,j=1,\ldots ,n$, so we can define $\Sigma_t = \Sigma_{x_i}$, $i=1,\ldots ,n$. For example, the set of event occurrences of the trace $\mathbf{t}_1$ from Example \ref{ex:traces1} is
$\Sigma_{\mathbf{t}_1} =\bset{a^{(1)},b^{(1)},c^{(1)},a^{(2)},b^{(2)},c^{(2)}}$.
Each $\lhd_i$ is a total order on $\Sigma_\mathbf{t}$.
The partial order  generated   by $\mathbf{t}$ can then be defined as $\prec_\mathbf{t} = \bigcap_{i=1}^k \lhd_{x_i}$. 
In fact, the set $\set{\lhd_{x_1},\ldots,\lhd_{x_k}}$ consists of all total extensions of $\PO_{\mathbf{t}}$ (see \cite{Ma1,Ma2}). Thus, the trace $\mathbf{t}_1=[abcbca]$ from Example \ref{ex:traces1} can be interpreted as the partial order $\prec_{\mathbf{t}_1}$ depicted in \figref{trace2po} (arcs inferred from transitivity are omitted for simplicity).

\begin{remark}
Given a sequence $s$, to construct the partial order $\PO_{[s]}$ generated by $[s]$, we \emph{do not} need to build up to exponentially many elements of $[s]$. We can simply construct the direct acyclic graph $(\Sigma_{[s]},\PO_s)$, where $x^{(i)} \PO_s y^{(j)}$ iff $x^{(i)}$ occurs before $y^{(j)}$ on the sequence $s$ and $(x,y)\not\in ind$. The relation $\PO_s$ is usually \emph{not} the same as the partial order $\PO_{[s]}$. However, after applying the \emph{transitive closure} operator, we have $\PO_{[s]} = \PO_{s}^+$ (cf. \cite{Di}). We will later see how this idea is generalized to the construction of so-structures and gso-structures from their ``trace'' representations. Note that to do so, it is inevitable that we have to generalize the \emph{transitive closure} operator to these order structures. \EOD
\label{rem:trace}
\end{remark}

From the concurrency point of view, the trace quotient monoid representation has a fundamental advantage over its labeled poset representation when studying the formal linguistic aspects of concurrent behaviors, e.g., Ochma\'nski's characterization of recognizable trace language \cite{Och} and Zielonka's theory of asynchronous automata \cite{Zie}.  For more details on traces and their various properties, the reader is referred to the monograph \cite{Di}.
The reader is also referred to \cite{BK91} for interesting discussions on the trade-offs: traces vs. labeled partial order models that allow auto-concurrency, e.g., pomsets.

\section{Relational Structures Model of Concurrency \label{sec:relsurvey}}
Even though partial orders are one of the main tools for modelling ``true concurrency,'' they have some limitations. While they can sufficiently model the ``earlier than" relationship, they can  model neither the ``not later than" relationship nor the ``non-simultaneously" relationship. It was shown in \cite{J4} that any reasonable concurrent behavior can be modelled by an appropriate {\em pair of relations}. This leads to the theory of {\em relational  structures models of concurrency} \cite{JK,GJ,J0} (see \cite{J0} for a detailed bibliography and history).

In this section, we review the theory of {\em stratified order structures} of \cite{JK} and {\em generalized stratified order structures} of \cite{GJ,J0}.
The former can model both the ``earlier than" and the ``not later than" relationships, but not the ``non-simultaneously" relationship. The latter can model all three relationships.

While traces provide sequence representations of causal partial orders, their extensions, comtraces and generalized comtraces discussed in the following sections, are {\em step sequence} representations of  stratified order structures and generalized stratified order structures respectively.

Since the theory of relational order structures is far less known than the theory of causal partial orders, we will not only give appropriate definitions but also introduce some intuition and motivation behind those definitions using simple examples.\\

We start with the concept of an {\em observation}:
\begin{quote}
An {\em observation} (also called a \emph{run} or an \emph{instance of concurrent behavior}) is an abstract model of the execution of a concurrent system.
\end{quote}

It was argued in \cite{J4} that {\em an observation must be a total, stratified or interval order} (interval orders are not used in this paper). Totally ordered observations can be represented by sequences while stratified observations can be represented by step sequences.\\

The next concept is a {\em concurrent behavior}:
\begin{quote}
A {\em concurrent behavior} ({\em concurrent history}) is a set of equivalent observations.
\end{quote}

When totally ordered observations are sufficient to define whole concurrent behaviors, then the concurrent behaviors can entirely be described by causal partial orders. However if  concurrent behaviors consist of more sophisticated sets of stratified observations, e.g., to model the ``not later than" relationship or the ``non-simultaneously'' relationship,  then we need  relational structures \cite{J4}.

\subsection{Stratified order structure}
By a \emph{relational structure}, we mean a triple $T=(X,R_1,R_2)$, where $X$ is a set and $R_1$, $R_2$ are binary relations on $X$. A relational structure $T'=(X',R'_1,R'_2)$ is an \emph{extension} of $T$, denoted as $T\subseteq T'$, iff $X=X'$, $R_1\subseteq R_1'$ and $R_2\subseteq R_2'$.

\begin{definition}[stratified order structure \cite{JK}]
A \emph{stratified order structure} (\emph{so-structure}) is a relational structure $S=(X,\prec,\sqsubset),$ such that for all $a,b,c\in X$, the following hold:
\begin{align*}
\text{S1:\hspace{5mm}}& a\not\sqsubset a & \text{S3:\hspace{5mm}}& a\sqsubset b \sqsubset c \;\wedge\; a \not= c \implies a\sqsubset c\\
\text{S2:\hspace{5mm}}& a \prec b \implies a \sqsubset b & \text{S4:\hspace{5mm}}&
a\sqsubset b \prec c \;\vee\; a\prec b \sqsubset c \implies a\prec c
\end{align*}
When $X$ is finite, $S$ is called a \emph{finite so-structure}. \EOD
\label{def:sos}
\end{definition}

Note that the axioms S1--S4 imply that $(X,\PO)$ is a poset and $a\prec b\Rightarrow b\not\sqsubset a.$  The relation $\prec$ is called \textit{causality} and represents the ``earlier than" relationship, and the relation $\sqsubset$ is called \textit{weak causality} and represents the ``not later than" relationship. The axioms S1--S4 model the mutual relationship between ``earlier than" and ``not later than" relations, {\em provided that the system runs are modelled by stratified orders}.

The concept of so-structures were independently introduced in \cite{GP} and \cite{JK0} (the axioms are slightly different from S1--S4, although equivalent). Their comprehensive theory has been presented in \cite{JK}. They have been successfully applied to model inhibitor and priority systems, asynchronous races, synthesis problems, etc. (see for example \cite{JK95,PK,JL08,JL08b,JLN,KK,KK2}). The name follows from the following result.

\begin{proposition}[\cite{J4}]
For every stratified order $\lhd$ on $X$, the triple $S_\lhd=(X,\lhd,\lhd^\frown)$ is a so-structure.\END
\label{p:soss}
\end{proposition}

\begin{definition}[stratified extension of so-structure \cite{JK}]
A {\em stratified} order $\lhd$ on $X$ is a {\em stratified extension} of a so-structure $S=(X,\prec,\sqsubset)$ if for all $\alpha,\beta\in X$,
\begin{align*}
\alpha \prec \beta \implies \alpha\lhd \beta \text{\hspace*{1cm} and \hspace*{1cm}}
\alpha\sqsubset \beta \implies \alpha\lhd^\frown \beta
\end{align*}
The set of all stratified extensions of $S$  is denoted as  $ext(S)$. \EOD
\label{def:extsos}
\end{definition}

According to Szpilrajn's theorem, every poset can be reconstructed by taking the intersection of all of its total extensions. A similar result holds for so-structures and stratified extensions.

\begin{theorem}[{\cite[Theorem 2.9]{JK}}]
Let $S=(X,\PO,\sq)$ be a so-structure. Then
\begin{center}
${ S=\Bigl(X,\bigcap_{\lhd\;\in\; ext(S)}\lhd,\bigcap_{\lhd\;\in\; ext(S)}\lhd^\frown\Bigr)}.$          \end{center}
\mEND
\label{theo:SzpStrat}
\end{theorem}

The set $ext(S)$ also has the following internal property that will be useful in various proofs.

\begin{theorem}[\cite{J4}]
Let $S=(X,\PO,\sq)$ be a so-structure. Then for every $a,b\in X$,
\begin{center}
$(\exists \lhd\in ext(S).\; a\lhd b) \wedge  (\exists \lhd\in ext(S).\; b\lhd a) \implies \exists \lhd\in ext(S).\; a\frown_\lhd b.$
\end{center}
\mEND
\label{pi3}
\end{theorem}

 The classification of concurrent behaviors provided in \cite{J4} says that a concurrent behavior conforms to the paradigm\footnote{A {\em paradigm} is a supposition or statement about the structure of a concurrent behavior (concurrent history) involving a treatment of {\em simultaneity}. See \cite{J0,J4} for more details.}
$\pi_3$ if it has the same property as stated in Theorem \ref{pi3} for $ext(S)$.
In other words, Theorem \ref{pi3} states that the set $ext(S)$ conforms to the paradigm $\pi_3$.

\subsection{Generalized stratified order structure}
The stratified order structures can adequately model concurrent histories only when the paradigm $\pi_3$ is satisfied. For the general case, we need {\em gso-structures} introduced in \cite{GJ} also under the assumption that the system runs are defined as stratified orders.

\begin{definition}[generalized stratified order structure \cite{GJ,J0}]
\begin{sloppypar}
A \emph{generalized stratified order structure} ({\em gso-structure}) is a relational structure $G=(X,\com,\sq)$ such that  $\sqsubset$ is irreflexive, $\com$ is symmetric and irreflexive, and the triple $S_{G}=(X,\prec_{G},\sqsubset)$, where $\prec_{G}\;=\;\com\cap\sqsubset$, is a so-structure, called the \textit{so-structure induced by $G$}.  When $X$ is finite, $G$ is called a finite gso-structure.
 \EOD
\end{sloppypar}
\label{def:gsos}
\end{definition}

The relation $\com$ is called \textit{commutativity} and represents the ``non-simultaneously" relationship, while the relation $\sqsubset$ is called \textit{weak causality} and represents the ``not later than" relationship.\\

For a binary relation $R$ on $X$, we let $\sym{R}\df R\cup R^{-1}$ denote the \emph{symmetric closure} of $R$.

\begin{definition}[stratified extension of gso-structure \cite{GJ,J0}]
A stratified order $\lhd$ on $X$ is a {\em stratified extension} of a gso-structure $G=(X,\com,\sq)$ if
for all $\alpha,\beta\in X$,
\begin{align*}
\alpha \com \beta \implies \alpha\sym{\lhd} \beta \text{\hspace*{1cm} and \hspace*{1cm}}
\alpha\sq \beta \implies \alpha\lhd^\frown \beta
\end{align*}
The set of all stratified extensions of $G$ is denoted as $ext(G)$. \EOD
\label{def:gsosext}
\end{definition}

Every gso-structure can also be uniquely reconstructed from its stratified extensions. The generalization of Szpilrajn's theorem for gso-structures can be stated as following.

\begin{theorem}[\cite{GJ,J0}]
Let $G=(X,\com,\sq)$ be a gso-structure. Then
\begin{center}
${ G=\Bigl(X,\bigcap_{\lhd\;\in\; ext(G)}\sym{\lhd},\bigcap_{\lhd\;\in\; ext(G)}\lhd^\frown\Bigr)}.$\end{center}
\mEND
\label{theo:SzpGStrat}
\end{theorem}

The gso-structures {\em do not} have an equivalent of Theorem \ref{pi3}. As a counter-example consider $G=(\{a,b,c\},\com_4,\sq_4)$ where $\com_4$ and $\sq_4$ are those from Figure \ref{fig:motiv1}. Hence $ext(G)=obs(P_4)=\{o_1,o_2\}$, where $o_1=\{a\}\{b\}\{c\}$ and $o_2=\{a\}\{c\}\{b\}$. For this gso-structure we have $b\stackrel{o_1}{\rightarrow}c$ and $c\stackrel{o_2}{\rightarrow}b$, but neither $o_1$ nor $o_2$ contains the step $\{b,c\}$, so Theorem \ref{pi3} does not hold. The lack of  an equivalent of Theorem \ref{pi3} makes proving properties about gso-structures more difficult, but they can model the most general concurrent behaviors provided that observations are modelled by stratified orders \cite{J0}.

\section{Comtraces \label{sec:comtraces}}

The standard definition of a free monoid $(E^*,\ast,\lambda)$ assumes that the elements of $E$ have no
internal structure (or their internal structure does not affect any monoidal properties), and they
are often called `letters', `symbols', `names', etc.
When we assume the elements of $E$ have some internal structure, for instance that they are sets,
this internal structure may be used when defining the set of equations $\It{EQ}$. This idea is exploited in the concept of a {\em comtrace}.\\

{\em Comtraces} ({\em com}bined {\em traces}), introduced in \cite{JK95} as an extension of
traces to distinguish between ``earlier than" and ``not later than" phenomena,
are equational monoids of step sequence monoids.
The equations $\It{EQ}$ are in this case defined
implicitly via two relations: {\em simultaneity} and {\em serializability}.

\begin{definition}[comtrace alphabet \cite{JK95}] Let $E$ be a finite set (of events) and let $ser \subseteq sim \subset E\times E$ be two relations called \emph{serializability} and \emph{simultaneity} respectively and the relation $sim$ is irreflexive and symmetric. Then the triple $(E,sim,ser)$ is called a \emph{comtrace alphabet}. \EOD
\label{def:comalpha}
\end{definition}

Intuitively, if $(a,b)\in sim$ then $a$ and $b$ can occur simultaneously
(or be a part of a {\em synchronous} occurrence in the sense of \cite{JL08}),
while $(a,b)\in ser$ means that $a$ and $b$ may occur simultaneously
and also $a$ may occur before $b$ (i.e., both executions are equivalent). We define $\E$,
the set of all (potential) {\em steps},
 as the set of all cliques of
the graph $(E,sim)$, i.e.,
\[\E \df \bset{ A \mid A\neq\emptyset \;\wedge\; \forall a,b\in A.\; (a=b \vee (a,b)\in sim)}.\]

\begin{definition}[Comtrace congruence \cite{JK95}]
Let $\theta=(E,sim,ser)$ be a comtrace alphabet and let $\equiv_{ser}$, called \emph{comtrace congruence}, be the $\It{EQ}$-congruence defined by the set of equations
\[\It{EQ}\df\{ A=BC \mid A=B\cup C\in \E \;\wedge\; B\times C \subseteq ser \}.\]
Then the equational monoid $(\E^*/\!\!\equiv_{ser},\circledast,[\lambda])$ is called a monoid of {\em comtraces} over $\theta$. \EOD
\label{def:commonoid}
\end{definition}

Since $ser$ is irreflexive, for each $(A=BC)\in \It{EQ}$ we have $B\cap C=\emptyset$. By \propref{p1}, the comtrace congruence relation can also be defined explicitly in non-equational form as follows.
\begin{proposition}
Let $\theta=(E,sim,ser)$ be a comtrace alphabet and let $\E^*$ be the set of all step sequences defined on $\theta$. Let $\eqa_{ser}\;\subseteq\; \E^*\times\E^*$ be the relation comprising all pairs $(t,u)$ of step sequences such that $t=wAz$ and $u=wBCz$, where $w,z\in\E^*$ and $A$, $B$, $C$ are steps satisfying $B\cup C \;= \;A$ and $B\times C\;\subseteq\; ser$. Then  $\eqb_{ser}= \left(\eqa_{ser}\cup\eqa^{-1}_{ser}\right)^*$.\END

\label{prop:comequi}
\end{proposition}

We will omit the subscript $ser$ from comtrace congruence and $\eqa_{ser}$, and only write $\equiv$ and $\eqa$ if it causes no ambiguity.
\begin{example}
Let $E=\set{a,b,c}$ where $a$, $b$ and $c$ are three atomic operations, where
\begin{align*}
a:\ \ \  y&:= x+y & b:\ \ \  x&:= y+2 & c:\ \ \  y&:= y+1
\end{align*}
Assume simultaneous reading is allowed, but simultaneous writing is not allowed. Then the events $b$ and $c$ can be performed simultaneously, and the execution of the step $\set{b,c}$ gives the same outcome as executing $b$ followed by $c$. The events $a$ and $b$ can also be performed simultaneously, but the outcome of executing the step $\set{a,b}$ is not the same as executing $a$ followed by $b$, or $b$ followed by $a$. Note that although executing the steps $\set{a,b}$ and $\set{b,c}$ is allowed, we cannot execute the step $\set{a,c}$ since that would require writing on the same variable $y$. 

Let $E=\set{a,b,c}$ be the set of events. Then we can define the comtrace alphabet $\theta=(E,sim,ser)$, where $sim = \bset{(a,b),(b,a),(b,c),(c,b)}$ and $ser=\set{(b,c)}$. Thus the set of all possible steps is \[\E_{\theta}=\bset{\{a\},\{b\},\{c\},\set{a,b},\set{b,c}}.\]
We observe that the set $\textbf{t}=[\{a\}\set{a,b}\{b,c\}] =\bset{ \{a\}\set{a,b}\{b,c\},\{a\}\set{a,b}\{b\}\{c\}}$ is a comtrace.
But the step sequence $\{a\}\set{a,b}\{c\}\{b\}$ is not an element of $\textbf{t}$ because $(c,b)\not\in ser$. \EOD
\label{ex:comtrace1}
\end{example}

Even though traces are quotient monoids over sequences and comtraces are
quotient monoids over step sequences (and the fact that steps are sets is used in the
definition of quotient congruence), traces can be regarded as a special case
of comtraces. In principle, each trace commutativity equation $ab=ba$ corresponds to two
comtrace equations $\{a,b\}=\{a\}\{b\}$ and $\{a,b\}=\{b\}\{a\}$. This relationship
can formally be formulated as follows.

Let $(E,ind)$ and $(E,sim,ser)$ be trace and comtrace alphabets respectively.
For each sequence $x=a_1\ldots a_n \in E^*$, we define
$x^{\{\}} = \{a_1\}\ldots \{a_n\}$ to be its corresponding step sequence, which in this case consists of only singleton steps.

\begin{lemma}
\begin{enumerate}
\item Assume $ser=sim$. Then for each comtrace $\mathbf{t}\in \E^*/{\!\equiv_{ser}}$ there exists
a step sequence $x= \set{a_1}\ldots\set{a_k} \in \E^*$ such that
$\mathbf{t}=[x]_{\equiv_{ser}}$.
\item If $ser=sim=ind$, then for each $x,y \in E^*$, we have
$x\equiv_{ind} y \iff x^{\{\}}\equiv_{ser} y^{\{\}}$.
\end{enumerate}
\label{le:com2maz}
\end{lemma}
\begin{proof} (1) follows from the fact that if $ser=sim$, then for each $A=\{a_1,...,a_k\}\in \E$, we have $A \equiv_{ser} 
\set{a_1}\ldots\set{a_k}$. (2) is a simple consequence of the definition of $x^{\{\}}$. \qed 
\end{proof}

Let $\mathbf{t}$ be a trace over $(E,ind)$ and let $\mathbf{v}$ be a comtrace over $(E,sim,ser)$. We say that $\mathbf{t}$ and $\mathbf{v}$ are {\em tantamount} if
$sim=ser=ind$ and there is $x\in E^*$ such that $\mathbf{t}=[x]_{\equiv_{ind}}$ and $\mathbf{v}=[x^{\{\}}]_{\equiv_{ser}}$. If a trace $\mathbf{t}$ and a comtrace $\mathbf{v}$ are equivalent we will write $\mathbf{t} \TCT \mathbf{v}$. Note that Lemma \ref{le:com2maz} guarantees that this definition is valid.

\begin{proposition}
Let ${\bf t},{\bf r}$ be traces and ${\bf v},{\bf w}$ be comtraces. Then
\begin{enumerate}
\item
${\bf t} \TCT {\bf v} \;\wedge\;  {\bf t} \TCT {\bf w} \implies {\bf v}={\bf w}$.
\item
${\bf t} \TCT {\bf v} \;\wedge\;  {\bf r} \TCT {\bf v} \implies {\bf t}={\bf r}$. 
\end{enumerate}
\end{proposition}
\begin{proof}
\begin{enumerate}
\item 
${\bf t} \TCT {\bf v}$ means that there is $x\in E^*$ such that
$\mathbf{t}= [x]_{\equiv_{ind}}$ and $\mathbf{v}=[x^{\{\}}]_{\equiv_{ser}}$, and 
${\bf t} \TCT {\bf w}$ means that there is $y\in E^*$ such that
$\mathbf{t}= [y]_{\equiv_{ind}}$ and $\mathbf{w}=[y^{\{\}}]_{\equiv_{ser}}$. Since $\mathbf{t}= [x]_{\equiv_{ind}} = [y]_{\equiv_{ind}}$ then $x\equiv_{ind} y$ and by Lemma \ref{le:com2maz}(2), 
$x^{\{\}}\equiv_{ser} y^{\{\}}$, i.e. ${\bf v} = {\bf w}$. 

\item Similarly as (1). \END
\end{enumerate}
\end{proof}

Equivalent traces and comtraces generate identical partial orders. However, we will postpone the discussion of this issue to Section~ \ref{sec:com2strat}. Hence {\em traces can be regarded as a special case of comtraces}.\\

Note that comtrace  might be a useful notion to formalize the concept of
{\em synchrony} (cf. \cite{JL08}).
In principle,  events $a_1,\ldots ,a_k$ are {\em synchronous} if they can be
executed in one step $\{a_1,\ldots ,a_k\}$ but this execution cannot be modelled by
any sequence of proper subsets of $\{a_1,\ldots ,a_k\}$.
Note that in general `synchrony' is not necessarily `simultaneity' as it does not
include the concept of time \cite{JL}. It appears, however, that the mathematics
 to deal with synchrony are   close to that to deal with simultaneity.

\begin{definition}[independency and synchrony]
Let $(E,sim,ser)$ be a given comtrace alphabet. We define the relations $ind$, $syn$ and the set
$\E_{syn}$ as follows:
\begin{itemize}
\item $ind \subseteq E\times E$, called {\em independency}, and defined as $ind = ser \cap ser^{-1}$,
\item $syn \subseteq E\times E$, called {\em synchrony}, and defined as:\\
\mbox{\hspace{3.5cm}}$(a,b)\in syn \iffdf (a,b)\in sim \setminus \sym{ser},$
\item $\E_{syn} \subseteq \E$, called {\em synchronous steps}, and defined as:\\
\mbox{\hspace{3cm}}$A \in \E_{syn} \iffdf A\neq\emptyset \wedge (\forall a,b \in A. \;\;(a,b)\in syn).$ \EOD
\end{itemize}
\label{def:synsteps}
\end{definition}

If $(a,b)\in ind$ then $a$ and $b$ are {\em independent}, i.e., executing them either
simultaneously, or $a$ followed by $b$, or $b$ followed by $a$, will yield exactly the same result.
If $(a,b)\in syn$ then $a$ and $b$ are {\em synchronous}, which means they might be
executed in one step, either $\{a,b\}$ or as a part of bigger step, but such an execution of $\{a,b\}$
is not equivalent to either $a$ followed by $b$, or $b$ followed by $a$. In principle,
the relation $syn$ is a counterpart of `synchrony' (cf. \cite{JL08}).
If $A \in \E_{syn}$, then the set of events $A$ can be executed as one step, but it
{\em cannot} be simulated by any sequence of its subsets.

\begin{example}
Assume we have $E=\set{a,b,c,d,e}$,
$sim=\set{(a,b),(b,a),(a,c),(c,a),(a,d),(d,a)}$, and
$ser= \set{(a,b),(b,a),(a,c) }$. Hence, $\E =\set{\{a,b\},\{a,c\},\{a,d\},\{a\},\{b\},\{c\},\{e\}}$, and
\begin{align*}
ind &= \set{(a,b),(b,a)} & syn &= \{(a,d),(d,a)\} & \E_{syn} &=\{\{a,d\}\}
\end{align*}

Since $\{a,d\}\in \E_{syn}$, the step $\set{a,d}$ {\em cannot} be split into smaller steps. For example the comtraces ${\bf x}_1=[ \set{a,b}\{c\}\{a\} ]$,
${\bf x}_2=[ \{e\}\{a,d\}\set{a,c} ]$, and
${\bf x}_3=[ \set{a,b}\{c\}\{a\}\{e\}\{a,d\}\set{a,c} ]$
 are respectively the following sets of step sequences:
\begin{align*}
{\bf x}_1&=\bset{ \set{a,b}\{c\}\{a\}, \set{a}\set{b}\{c\}\{a\},\set{b}\set{a}\set{c}\set{a}, \set{b}\set{a,c}\set{a}}\\
{\bf x}_2&=\bset{  \{e\}\{a,d\}\set{a,c} ,  \{e\}\{a,d\}\set{a}\set{c} }\\
{\bf x}_3&= \left\{
\begin{array}{c}
\set{a,b}\{c\}\{a\}\{e\}\{a,d\}\set{a,c} , \{a\}\{b\}\{c\}\{a\}\{e\}\{a,d\}\set{a,c},\\
\{b\}\{a\}\{c\}\{a\}\{e\}\{a,d\}\set{a,c} , \{b\}\{a,c\}\{a\}\{e\}\{a,d\}\set{a,c},\\
\set{a,b}\{c\}\{a\}\{e\}\{a,d\}\set{a}\set{c},\{a\}\{b\}\{c\}\{a\}\{e\}\{a,d\}\set{a}\set{c} ,\\
\{b\}\{a\}\{c\}\{a\}\{e\}\{a,d\}\set{a}\set{c} , \{b\}\{a,c\}\{a\}\{e\}\{a,d\}\set{a}\set{c} 
\end{array}
\right\}
\end{align*}
We also have ${\bf x}_3 = {\bf x}_1\circledast {\bf x}_2$. Note that since $(c,a)\notin ser$, $\set{a,c}\; \equiv_{ser} \;\set{a}\set{c}\; \not\equiv_{ser} \;\set{c}\set{a}$.  \EOD
 \label{e4}
\end{example}

We can easily extend the concepts of comtraces to the level of languages, with potential applications similar to traces.  For any step sequence language $L$, we define a comtrace language $[L]_{\Theta}$
(or just $[L]$) to be the set $\set{[u] \mid u \in L}$.  The languages of comtraces  provide a bridge between operational and structural semantics.
In other words, if a step sequence language $L$ describes an operational semantics of a given concurrent system, we
only need to derive the comtrace alphabet $(E,sim,ser)$ from the system, and the comtrace language $[L]$ defines the structural semantics of the system.

\begin{example} Consider the following simple concurrent system $\mathsf{Priority}$, which comprises two sequential
subsystems such that
\begin{itemize}
\item the first subsystem can cyclically engage in event $a$ followed by event $b$,
\item the second subsystem can cyclically engage in event $b$ or in event $c$,
\item the two systems synchronize by means of handshake communication,
\item there is a priority constraint stating that if it is possible to execute event $b$, then
$c$ must not be executed.
\end{itemize}

This example has often been analyzed in the literature
(cf. \cite{JK1}), usually under the interpretation
that $a=\textrm{`Error Message'}$, $b=\textrm{`Stop And Restart'}$, and $c=\textrm{`Some Action'}$.
It can be formally specified in various notations including
\emph{Priority} and \emph{Inhibitor Nets} (cf. \cite{JK0,JK}). Its operational semantics
(easily found in any model) can be defined by the following step sequence language
\[L_{\mathsf{Priority}}\df\emph{Pref}\bigl((\set{c}^*\cup\set{a}\set{b}\cup\set{a,c}\set{b})^*\bigr),\]
where $\mathit{Pref}(L)\df\bigcup_{w\in L}\set{u\in L\mid \exists v.\ uv=w}$ denotes the \emph{prefix closure} of  $L$.

The rules for deriving the comtrace alphabet $(E,sim,ser)$ depend on the model, and for $\mathsf{Priority}$, the set of possible steps is $\E = \bset{\set{a},\set{b},\set{c},\set{a,c}}$, and $ser=\set{(c,a)}$ and $sim=\set{(a,c),(c,a)}$.
Then, $[L_{\mathsf{Priority}}]$ defines the structural comtrace semantics of $\mathsf{Priority}$. For instance, the comtrace 
$[\set{a,c}\set{b}] = \bset{ \set{c}\set{a}\set{b},\set{a,c}\set{b}}$ is in the language $[L_{\mathsf{Priority}}]$. \EOD
\end{example}

\section{Generalized Comtraces \label{sec:gcomtraces}}

There are reasonable concurrent behaviors that cannot be modelled by any comtrace. Let us analyze the following example.

\begin{example}
Let $E=\set{a,b,c}$ where $a$, $b$ and $c$ are three atomic operations defined as follows
(we assume simultaneous reading is allowed):
\begin{align*}
a:\ \ \ x&:= x+1 & b:\ \ \ x&:= x+2 & c:\ \ \ y&:= y+1
\end{align*}
It is reasonable to consider them all as `concurrent' as any order of their executions
yields exactly the same results (see \cite{J4,JK} for more motivation and formal considerations as well as the program $P_4$ of Example \ref{ex:motiveprog}). Assume that simultaneous reading  is allowed, but simultaneous writing  is not. Then
while simultaneous execution of $\{a,c\}$ and $\{b,c\}$ are allowed,
we cannot execute $\{a,b\}$, since we simultaneous writing on the same variable $x$ is not allowed.

The set of all equivalent executions (or runs) involving one occurrence of the operations $a$, $b$ and $c$, and modelling the above case,
\begin{align*}
 {\bf x}= \left\{\begin{array}{c}
\{a\}\{b\}\{c\},\{a\}\{c\}\{b\},\{b\}\{a\}\{c\},\{b\}\{c\}\{a\},\{c\}\{a\}\{b\}, \\
\{c\}\{b\}\{a\},\{a,c\}\{b\}, \{b,c\}\{a\}, \{b\}\{a,c\}, \{a\}\{b,c\}\end{array}\right\},
\end{align*}
is a valid concurrent history \cite{J4,JK}. However $x$ is {\em not} a comtrace. The problem is that we have $\{a\}\{b\}\eqb\{b\}\{a\}$ but $\{a,b\}$ {\em is not} a valid step, so comtrace cannot represent this situation. \EOD
 \label{ex:gcom}
\end{example}

 In this section, we
will introduce the {\em generalized comtrace} notion ({\em g-comtrace}), an extension of comtrace, which is also defined over step sequences. The g-comtraces will be able to model ``non-simultaneously" relationship similar to the one from Example \ref{ex:gcom}.

\begin{definition}[generalized comtrace alphabet]
\begin{sloppypar}
Let $E$ be a finite set (of events). Let $ser$, $sim$ and $\It{inl}$  be three relations on $E$
called {\em serializability}, {\em simultaneity} and {\em interleaving} respectively satisfying the following conditions:
\begin{itemize}
 \item $sim$ and $inl$ are irreflexive and symmetric,
 \item $ser \subseteq sim$, and
 \item $sim \cap inl =\emptyset.$
\end{itemize}
Then the triple $(E,sim,ser,inl)$ is called a {\em g-comtrace alphabet}.\EOD
\end{sloppypar}
\end{definition}

The interpretation of  the relations $sim$ and $ser$ is as in \defref{comalpha}, and $(a,b)\in inl$ means
$a$ and $b$ cannot occur simultaneously, but their occurrence in any order is equivalent.
As for comtraces, we define the set $\E$ of all (potential) {\em steps} 
 as the set of all cliques of
the graph $(E,sim)$.

\begin{definition}[generalized comtrace congruence]
Let $\Theta=(E,sim,ser,inl)$ be a g-comtrace alphabet and let $\equiv_{\set{ser,inl}}$, called \emph{g-comtrace congruence}, be the $\It{EQ}$-congruence defined by the set of equations $\It{EQ}=\It{EQ}_1\cup \It{EQ}_2$, where
\begin{align*}
\It{EQ}_1&\df\{ A=BC \mid A=B\cup C\in \E \;\wedge\; B\times C \subseteq ser \},\\
\It{EQ}_2&\df\{ BA=AB \mid A\in \E \;\wedge\; B\in \E \;\wedge\; A\times B \subseteq inl \}.
\end{align*}
The equational monoid $\bigl(\E^*/\!\!\equiv_{\set{ser,inl}},\circledast,[\lambda]\bigr)$ is called a \emph{monoid of g-comtraces}  over $\Theta$. \EOD
\label{def:gcommonoid}
\end{definition}

Since $ser$ and $inl$ are irreflexive, $(A=BC)\in \It{EQ}_1$ implies $B\cap C=\emptyset$, and
$(AB=BA)\in \It{EQ}_2$ implies $ A\cap B=\emptyset$. Since $inl\cap sim=\emptyset$, we also have that if $(AB=BA)\in \It{EQ}_2$, then $A\cup B\notin \E$. \\

By \propref{p1}, the g-comtrace congruence relations can also be defined explicitly in non-equational form as follows.
\begin{definition}
Let $\Theta=(E,sim,ser,inl)$ be a g-comtrace alphabet and let $\E^*$ be the set of all step sequences defined on $\Theta$.
\begin{itemize}
 \item Let $\eqa_1\;\subseteq\; \E^*\times\E^*$ be the relation comprising all pairs $(t,u)$ of step sequences such that $t=wAz$ and $u=wBCz$ where $w,z\in\E^*$ and $A$, $B$, $C$ are steps satisfying
$B\cup C \;= \;A$ and  $B\times C\;\subseteq\; ser$.
 \item Let $\eqa_2\;\subseteq\; \E^*\times\E^*$ be the relation comprising all pairs $(t,u)$ of step sequences such that $t=wABz$ and $u=wBAz$ where $w,z\in\E^*$ and $A$, $B$ are steps satisfying $A\times B\;\subseteq\; inl$. 
\end{itemize}
We define $ \eqa_{\set{ser,inl}} $ as $ \eqa_{\set{ser,inl}} \df \eqa_1\cup \eqa_2$. \EOD
\label{def:eqa}
\end{definition}

\begin{proposition}
For each g-comtrace alphabet $\Theta=(E,sim,ser,inl)$
$$\eqb_{\set{ser,inl}}= \Bigl(\eqa_{\set{ser,inl}}\cup\eqa_{\set{ser,inl}}^{-1}\Bigr)^*.$$  
\label{prop:gcomequi}
\end{proposition}
\begin{proof} Follows from \propref{p1}. \qed
\end{proof}

The name ``generalized comtraces'' comes from the fact that when $inl=\emptyset$, \defref{gcommonoid} coincides with \defref{commonoid} of a comtrace monoid. We will omit the subscript $\set{ser,inl}$ from $\eqb_{\set{ser,inl}}$ and $\eqa_{\set{ser,inl}}$, and write $\equiv$ and $\eqa$ when causing no ambiguity.

\begin{example}
The set ${\bf x}$ from \exref{gcom} is a g-comtrace, where we have
 $E=\{a,b,c\}$, $ser=sim=\{(a,c),(c,a),(b,c),(c,b)\}$, $inl=\{(a,b),(b,a)\}$,
and $\E=\{\{a,c\},\{b,c\},\{a\},\{b\},\{c\}\}$. \EOD
\label{e6}
\end{example}

It is worth noting that there is an {\em important difference} between the equation $ab=ba$ for traces,
and the equation $\{a\}\{b\}=\{b\}\{a\}$ for g-comtrace monoids.
For traces, the equation $ab=ba$, when translated into step sequences, corresponds
to two equations $\{a,b\}=\{a\}\{b\}$ and $\{a,b\}=\{b\}\{a\}$, which implies $\{a\}\{b\}\equiv\set{a,b}\equiv\{b\}\{a\}$.
For g-comtrace monoids, the equation $\{a\}\{b\}=\{b\}\{a\}$
implies that $\{a,b\}$ {\em is not a step}, i.e., neither the equation $\{a,b\}=\{a\}\{b\}$ nor
the equation $\{a,b\}=\{b\}\{a\}$ belongs to the set of equations. In other words, for traces the equation $ab=ba$
means `independency', i.e., executing $a$ and $b$ in any order or simultaneously will yield the same consequence.
For g-comtrace monoids, the equation $\{a\}\{b\}=\{b\}\{a\}$
means that execution of $a$ and $b$ in any order yields the same result, but executing of $a$ and $b$ in any order is \emph{not} equivalent to executing them simultaneously.

\section{Algebraic Properties of Comtrace and Generalized Comtrace Congruences \label{sec:algebraic}}
Algebraic properties of trace congruence operations such as  \emph{left/right cancellation} and \emph{projection} are well understood. They are intuitive and simple tools with many applications \cite{Ma2}. In this section we will generalize these cancellation and projection properties to comtrace and g-comtrace. The basic obstacle is switching from sequences to step sequences.

\subsection{Properties of comtrace congruence}
Let us consider a comtrace alphabet $\theta=(E,sim,ser)$ where we reserve $\E$ to denote the set of all possible steps of $\theta$ throughout this section.\\

For each step sequence or enumerated step sequence $x = X_1\ldots X_k$, we define the \emph{step sequence weight} of $x$ as $\wei(x)\df\Sigma_{i=1}^k|X_i|$. We also define $\Al(x) \df \bigcup_{i=1}^k X_i$.\\

Due to the commutativity of the independency relation for traces, the \textit{mirror rule}, which says if two sequences are congruent, then their \textit{reverses} are also congruent, holds for \emph{trace congruence} \cite{Di}. Hence, in trace theory, we only need a \emph{right cancellation} operation to produce congruent \emph{subsequences} from congruent sequences, since the \emph{left cancellation} comes from the right cancellation of the reverses.

However, the \textit{mirror rule}  does not hold for comtrace congruence since the relation $ser$ is usually not commutative. \exref{comtrace1} works as a counter example since $\set{a}\set{b,c} \eqb \set{a}\set{b}\set{c}$ but $\set{b,c}\set{a} \not\eqb \set{c}\set{b}\set{a}$. Thus, we define separate left and right cancellation operators for comtraces.\\

Let $a \in E$, $A\in \E$ and $w\in\E^*$.  The operator $\RC$, \textit{step sequence right
cancellation}, is defined  as follows: \medskip\\
\mbox{\hspace{2.5cm}}$\lambda \RC a \df \lambda,$
\mbox{\hspace{1cm}}
$wA \RC a \df
\begin{cases}
(w \RC a)A &\text{ if } a \not\in A\\
w &\text{ if } A = \set{a}\\
w(A\setminus\set{a}) &\text{ otherwise.}
\end{cases}$\\

Symmetrically, a \textit{step sequence left cancellation} operator $\LC$ is defined as follows:
\begin{align*}
&\lambda \LC a \df \lambda, &Aw \LC a \df
\begin{cases}
A(w \LC a) & \text{ if }a \not\in A\\
w & \text{ if }A = \set{a}\\
(A\setminus\set{a})w & \text{otherwise.}
\end{cases}
\end{align*}

Finally, for each $D \subseteq E$, we define the function $\pi_D:\E^* \rightarrow \E^*$,
{\em step sequence projection} onto $D$, as follows:
\begin{align*}
&\pi_D(\lambda) \df \lambda, &\pi_D(wA) \df
\begin{cases}
\pi_D(w) & \text{  if }A \cap D = \emptyset\\
\pi_D(w)(A \cap D ) &\text{  otherwise.}
\end{cases}
\end{align*}

The algebraic properties of comtraces are   similar to those of traces \cite{Ma2}.

\begin{proposition}\mbox{}\\
\vspace{-0.5cm}
\begin{enumerate}
\item $u\eqb v \implies \wei(u) = \wei(v)$. \hfill {\em (step sequence weight equality)}
\item $u\eqb v \implies |u|_a = |v|_a$.  \hfill \emph{(event-preserving)}
\item $u \eqb v \implies u\RC a \eqb v\RC a$. \hfill {\em (right cancellation)}
\item $u \eqb v \implies u\LC a \eqb v\LC a$.  \hfill {\em (left cancellation)}
\item $u\eqb v \Longleftrightarrow \forall s,t\in \E^*.\;sut \eqb svt$. \hfill {\em (step subsequence cancellation)}
\item $u \eqb v \implies \pi_D(u) \eqb \pi_D(v)$. \hfill {\em (projection rule)}
\end{enumerate}
\label{prop:p3}
\end{proposition}
\begin{proof} The proofs use the same techniques as in \cite{Ma2}. We would like recall only the following key observation that simplifies the proof of this proposition: since $\eqb$ is the symmetric transitive closure of $\eqa$, it suffices to show that $u\approx v$ implies the right hand side of (1)--(6). The rest follows naturally from the definition of comtrace $\approx$ and the congruence $\eqb$. \qed
\end{proof}

Note that $(w\RC a)\RC b = (w\RC b)\RC a$, so we define
\begin{align*}
w\RC \{a_1,\ldots,a_k\}&\df\Bigl(\ldots\bigl((w\RC a_1)\RC a_2\bigr)\ldots\Bigr)\RC a_k, \text{ and}\\
w\RC A_1\ldots A_k &\df \Bigl(\ldots\bigl((w\RC A_1)\RC A_2\bigr)\ldots\Bigr)\RC A_k
\end{align*}

We define dually for $\LC$.
Hence Proposition \ref{prop:p3} (4) and (5) can be generalized as follows.

\begin{corollary} For all $u,v,x \in \E^*$, we have
\begin{enumerate}
\item $u\eqb v \implies u\RC x \eqb v\RC x$.
\item $u\eqb v \implies u\LC x \eqb v\LC x$. \END
\end{enumerate}
\label{c2}
\end{corollary}

\subsection{Properties of generalized comtrace congruence}
Using the same proof technique as in \propref{p3}, we can show that g-comtrace congruence has  the same algebraic properties as comtrace congruence.

\begin{proposition} Let $\E$ be the set of all steps over a g-comtrace alphabet
 $(E,sim,ser,inl)$ and  $u,v\in\E^*$. Then
\begin{enumerate}
\item $u\eqb v \implies \wei(u) = \wei(v)$. \hfill {\em (step sequence weight equality)}
\item $u\eqb v \implies |u|_a = |v|_a$.  \hfill \emph{(event-preserving)}
\item $u \eqb v \implies u\RC a \eqb v\RC a$. \hfill {\em (right cancellation)}
\item $u \eqb v \implies u\LC a \eqb v\LC a$.  \hfill {\em (left cancellation)}
\item $u\eqb v \Longleftrightarrow \forall s,t\in \E^*.\;sut \eqb svt$. \hfill {\em (step subsequence cancellation)}
\item $u \eqb v \implies \pi_D(u) \eqb \pi_D(v)$. \hfill {\em (projection rule)}\\\qed
\end{enumerate}

\label{prop:gcprop1}
\end{proposition}

\begin{corollary} For all step sequences $u, v, x$ over a g-comtrace alphabet $(E,sim,ser,inl)$,
\begin{enumerate}
\item $u\eqb v \implies u\RC x \eqb v\RC x$.
\item $u\eqb v \implies u\LC x \eqb v\LC x$. \END
\end{enumerate}
\end{corollary}

The following proposition ensures that if any relation from the set
$\set{\le,\ge,<,>,=,\not=}$ holds for the positions of two event occurrences  after applying cancellation or projection operations on a g-comtrace $[\h{u}]$, then it also holds for the whole $[\h{u}]$.

\begin{proposition}\begin{sloppypar}
Let $\h{u}$ be an enumerated step sequence over a g-comtrace alphabet $(E,sim,ser,inl)$ and $\alpha,\beta,\gamma\in\Sigma_u$ such that $\gamma\notin \set{\alpha,\beta}$. Let  $\calf{R}\in \set{\le,\ge,<,>,=,\not=}$. Then
\end{sloppypar}

\begin{enumerate}
 \item If $\forall \h{v}\in [\h{u}\LC\gamma].\; pos_{\h{v}}(\alpha)\; \calf{R}\; pos_{\h{v}}(\beta)$, then $\forall \h{w}\in [\h{u}].\; pos_{\h{w}}(\alpha)\; \calf{R}\; pos_{\h{w}}(\beta)$.
 \item If $\forall \h{v}\in [\h{u}\RC\gamma].\; pos_{\h{v}}(\alpha)\; \calf{R}\; pos_{\h{v}}(\beta)$, then $\forall \h{w}\in [\h{u}].\; pos_{\h{w}}(\alpha)\; \calf{R}\; pos_{\h{w}}(\beta)$.
 \item If $S\subseteq\Sigma_u$ such that $\set{\alpha,\beta}\subseteq S$, then 
 \[\Bigl(\forall \h{v}\in [\pi_S(\h{u})].\; pos_{\h{v}}(\alpha)\; \calf{R}\; pos_{\h{v}}(\beta)\Bigr)\implies \Bigl(\forall \h{w}\in [\h{u}].\; pos_{\h{w}}(\alpha)\; \calf{R}\; pos_{\h{w}}(\beta)\Bigr).\]
\end{enumerate}
\label{prop:invsubs}
\end{proposition}
\begin{proof}
\begin{enumerate}
\item Assume that
\begin{equation}
\forall \h{v}\in [\h{u}\LC\gamma].\; pos_{\h{v}}(\alpha)\; \calf{R}\; pos_{\h{v}}(\beta)
\label{eq:invsubs.1}
\end{equation}
Suppose for a contradiction that $\exists \h{w}\in [\h{u}].\; \neg(pos_{\h{w}}(\alpha)\; \calf{R}\; pos_{\h{w}}(\beta))$. Since $\gamma\notin \set{\alpha,\beta}$, we have $\neg(pos_{\h{w}\LC\gamma}(\alpha)\; \calf{R}\; pos_{\h{w}\LC\gamma}(\beta))$. But $\h{w}\in [\h{u}]$ implies $\h{w}\LC\gamma\eqb \h{u}\LC\gamma$. Hence, $\h{w}\LC\gamma\in[\h{u}\LC\gamma]$ and $\neg(pos_{\h{w}\LC\gamma}(\alpha)\; \calf{R}\; pos_{\h{w}\LC\gamma}(\beta))$, contradicting \eref{invsubs.1}.
\item Dually to part (1).
\item Assume that
\begin{equation}
\forall \h{v}\in [\pi_S(\h{u})].\; pos_{\h{v}}(\alpha)\; \calf{R}\; pos_{\h{v}}(\beta)
\label{eq:invsubs.2}
\end{equation}
Suppose for a contradiction that $\exists \h{w}\in [\h{u}].\; \neg(pos_{\h{w}}(\alpha)\; \calf{R}\; pos_{\h{w}}(\beta))$. Since $\set{\alpha,\beta}\subseteq S$, we have $\neg(pos_{\pi_S(\h{w})}(\alpha)\; \calf{R}\; pos_{\pi_S(\h{w})}(\beta))$. But $\h{w}\in [\h{v}]$ implies $\pi_S(\h{w})\eqb \pi_S(\h{u})$. Hence, $\pi_S(\h{w})\in[\pi_S(\h{u})]$ and $\neg(pos_{\pi_S(\h{w})}(\alpha)\; \calf{R}\; pos_{\pi_S(\h{w})}(\beta))$, contradicting  \eref{invsubs.2}.\END
\end{enumerate}
\end{proof}

Clearly the above results also hold for comtraces as they are just g-comtraces with $inl=\emptyset$.

\section{Maximally Concurrent and Canonical Representations \label{sec:can}}
In this section, we show that traces, comtraces and g-comtraces all have some special representations, that intuitively correspond to \emph{maximally concurrent execution of concurrent histories}, i.e., ``executing as much as possible in parallel.'' This kind of semantics is formally defined and analyzed for example in \cite{DJKL}. However such representations are truly unique only for comtraces. For traces and g-comtraces, unique (or {\em canonical}) representations are obtained by adding some arbitrary total ordering on their alphabets.

In this section we will start with the general case of g-comtraces and then consider comtraces and traces as a special case.

\subsection{Representations of generalized comtraces}
Let $\Theta=(E,sim,ser,inl)$ be a g-comtrace alphabet and $\E$ be
the set of all steps over $\Theta$. We will start with the most ``natural" definition which is the straightforward application of the approach used in \cite{DJKL} for an alternative version of traces called ``vector firing sequences" (see \cite{JL92,Shi}).

\begin{definition}[greedy maximally concurrent form]
A step sequence $u = A_1\ldots A_k \in \E^*$ is in
{\em greedy maximally concurrent form} ({\em GMC-form}) if and only if
for each $i=1,\ldots ,k$:
$$\bigl(B_iy_i \equiv A_i\ldots A_k\bigr) \implies |B_i|\leq |A_i|,$$
where for all $i=1,\ldots ,k$, $A_i, B_i \in \E$, and $y_i\in \E^*$. \EOD
\label{GMC}
\end{definition}

\begin{proposition}
For each g-comtrace ${\bf u}$ over $\Theta$ there is  a step sequence $u\in \E^*$ in GMC-form such that ${\bf u}=[u]$.
\label{ExGMC}
\end{proposition}
\begin{proof}
Let $u=A_1\ldots A_k$, where the steps $A_1,\ldots ,A_k$ are generated by the following simple greedy algorithm:
\begin{algorithmic}[1]
\STATE  Initialize $i\leftarrow 0$ and $u_0\leftarrow u$
\WHILE{$u_i\neq\lambda$}
\STATE $i \leftarrow i+1$
\STATE Find $A_i$ such that 
there exists $y$ such that $A_iy\equiv u_{i-1}$ and 
for each $Bz\equiv A_iy\equiv u_{i-1}$, $|B|\leq |A_i|$
\STATE $u_i \leftarrow u_{i-1} \LC A_i$
\ENDWHILE
\STATE $k\leftarrow i-1$.
\end{algorithmic}

Since $\wei(u_{i+1})<\wei(u_i)$ the above algorithm always terminates. Clearly $u = A_1\ldots A_k$ is in GMC-form and $u\in {\bf u}$. \END
\end{proof}

The algorithm from the proof of Proposition \ref{ExGMC} used to generate $A_1$ ,\ldots , $A_k$  justifies the prefix ``greedy" in Definition \ref{GMC}. However the GMC representation of g-comtraces is seldom unique and often not ``maximally concurrent". Consider the following two examples.
\begin{example}
Let $E=\{a,b,c\}$, $sim =\{(a,c),(c,a)\}$, $ser=sim$ and $inl=\{(a,b),(b,a)\}$ and ${\bf u}=[ \{a\}\{b\}\{c\} ] =\{\{a\}\{b\}\{c\},\{b\}\{a\}\{c\},
\{b\}\{a,c\}\}$. Note that 
both $\{a\}\{b\}\{c\} $ and $\{b\}\{a,c\}$ are
in GMC-form, but only $\{b\}\{a,c\}$ can intuitively be interpreted as maximally concurrent. \EOD
\label{GMC1}
\end{example}

\begin{example}
\mbox{}\\
\begin{minipage}{8cm}
Let $E=\{a,b,c,d,e\}$, and $sim=ser$, $inl$ be as in the picture on the right, and let ${\bf u}=[ \{a\}\{b,c,d,e\} ]$. One can easily verify by inspection that $\{a\}\{b,c,d,e\}$ is the shortest element of
${\bf u}$ and the only element of ${\bf u}$ in GMC-form is
$\{b,e,d\}\{a\}\{c\}$. The step sequence $\{b,e,d\}\{a\}\{c\}$ is longer and intuitively less maximally concurrent than the step sequence$\{a\}\{b,c,d,e\}$.\EOD
\end{minipage}\hfill
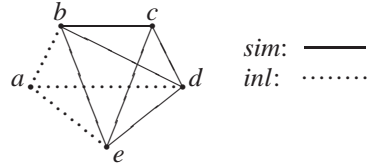
\begin{minipage}{4.5cm}
\vspace{-5mm}
\begin{picture}(5,5)
\setlength{\unitlength}{4mm}

\put(0,0){\circle*{0.2}}
\put(1,2){\circle*{0.2}}
\put(2.5,-2){\circle*{0.2}}
\put(4,2){\circle*{0.2}}
\put(5,0){\circle*{0.2}}
\drawline(1,2)(4,2)(5,0)(2.5,-2)(4,2)
\drawline(5,0)(1,2)(2.5,-2)
\dottedline[.]{0.3}(2.5,-2)(0,0)(1,2)
\dottedline[.]{0.3}(0,0)(5,0)

\put(-0.2,0){\makebox(0,0)[br]{$a$}}
\put(1.2,2.2){\makebox(0,0)[br]{$b$}}
\put(4.2,2.2){\makebox(0,0)[br]{$c$}}
\put(5.2,0){\makebox(0,0)[bl]{$d$}}

\put(2.7,-2.1){\makebox(0,0)[tl]{$e$}}

\put(7,1){\makebox(0,0)[bl]{$sim $:}}
\drawline(9,1.3)(11,1.3)

\put(7,0){\makebox(0,0)[bl]{$inl $:}}
\dottedline[.]{0.3}(9,0.3)(11,0.3)

\end{picture}
\end{minipage}
\label{GMC2}
\end{example}

Hence for g-comtraces the greedy maximal concurrency notion is not necessarily the global maximal concurrency notion, so we will try another approach.\\

Let $x=A_1\ldots A_k$ be a step sequence. We define $\len(A_1\ldots A_k)\df k$.
We also say that $A_i$ is {\em maximally concurrent in} $x$ if $B_iy_i\equiv A_i\ldots A_k \implies |B_i|\leq |A_i|$.
Note that $A_k$ {\em is always maximally concurrent in} $x$, which makes the following definition correct.

For every step sequence $x=A_1\ldots A_k$, let $mc(x)$ be the smallest $i$ such that $A_i$ is maximally concurrent in $x$.

\begin{definition}
A step sequence $u=A_1\ldots A_k$ is {\em maximally concurrent} ({\em MC-}) iff
\begin{enumerate}
\item $v\equiv u \;\implies\;  \len(u) \leq \len(v)$,
\item for all $i=1,\ldots ,k$ and for all $w$,
\begin{itemize}
 \item[] $\bigl(u_i=A_i\ldots A_k\equiv w\;\wedge\;\len(u_i) = \len(w)\bigr)\implies mc(u_i)\leq mc(w).$\EOD
\end{itemize}
\end{enumerate}
\label{MC}
\end{definition}

\begin{theorem}
For every g-comtrace ${\bf u}$, there exists a step sequence $u\in {\bf u}$ such that $u$ is maximally concurrent.
\label{t:mc}
\end{theorem}
\begin{proof}
Let $u_1\in {\bf u}$ be a step sequence such that for each $v$, $v\equiv u_1 \implies \len(u_1)\leq \len(v)$, and
$(v\equiv u_1 \;\wedge\;
\len(u_1)= \len(v)) \implies mc(u_1)\leq mc(v)$. Obviously such $u_1$ exists for every g-comtrace ${\bf u}$. Assume that $u_1=A_1w_1$ and $\len(u_1)=k$. Let $u_2$ be a step sequence satisfying
$u_2\equiv w_1$, $u_2\equiv v \implies \len(u_2)\leq \len(v)$, and
$(v\equiv u_2 \;\wedge\;
\len(u_2)= \len(v)) \implies mc(u_2)\leq mc(v)$. Assume that $u_2=A_2w_3$. We repeat this process $k-1$ times. Note that $u_k=A_k\in \E$. The step sequence $u=A_1\ldots A_k$ is maximally concurrent and $u\in {\bf u}$.
\END
\end{proof}

For the case of Example \ref{GMC1} the step sequence $\{b\}\{a,c\}$ is maximally concurrent and for the case of Example \ref{GMC2} the step sequence $\{a\}\{b,c,d,e\}$ is maximally concurrent. There may be more than one maximally concurrent step sequences in a g-comtrace. For example if $E=\{a,b\}$, $sim=ser=\emptyset$, $inl=\{(a,b),(b,a)\}$, then the g-comtrace $t=[\{a\}\{b\}]=\{\{a\}\{b\},\{b\}\{a\}\}$ and both
$\{a\}\{b\}$ and $\{b\}\{a\}$ are maximally concurrent.\\

Having a canonical (unique) representation is often useful in proving properties about g-comtraces since it allows us to \emph{uniquely identify} a g-comtrace. Furthermore, to be really useful in proofs, a canonical representation should be easy to construct and manipulate. For g-comtraces, it turns out that a natural way to get a canonical representation is: fix a total order on the alphabet, extend it to a lexicographical ordering on step sequences, and then simply choose the lexicographically least element.

\begin{definition}[lexicographical ordering]
Assume that we have a {\em total order} $<_E$ on $E$.
\begin{enumerate}
\item
We  define a \emph{step order} $\stor{<}$ on $\E$ as follows:
\begin{itemize}
 \item[] $A\stor{<}B \iffdf |A|>|B|\;\vee\;\bigl(|A|=|B|\;\wedge\; A\not= B
			\;\wedge\;  min_{<_E}(A\setminus B)<_E min_{<_E}(B\setminus A)\bigr),$
\end{itemize}
where $min_{<_E}(X)$ denotes the least element of the set $X\subseteq E$ w.r.t. $<_E$.
\item
Let $A_1\ldots A_n$ and $B_1\ldots B_m$ be two sequences in $\E^*$. We define a \emph{lexicographical order} $\lex{<}$ on step sequences in a natural way as the lexicographical order induced by $\stor{<}$,  i.e.,
\begin{itemize}
 \item[] $A_1\ldots A_n \lex{<} B_1\ldots B_m \iffdf \exists k>0\,\forall i<k .\bigl( A_i=B_i
			\wedge \left(A_k\stor{<} B_k\vee n<k\le m\right)\bigr).$\EOD
\end{itemize}
\end{enumerate}
\label{def:lexord}
\end{definition}

Directly from the above definition, it follows that $\stor{<}$ totally orders the set of possible steps $\E$ and $\lex{<}$ totally orders the set of possible step sequences $\E^*$.

\begin{example}
Assume that $a <_E b <_E c <_E d <_E e$. Then we have $\{a,b,c,e\} \stor{<} \{b,c,d\}$ since $\{a,b,c,e\}\setminus\{b,c,d\}=\{a\}$,
$\{b,c,d\}\setminus\{a,b,c,e\}=\{d\}$, and $a <_E d$. And
$\{a,c\}\{b,c\}\{d\}\{d,c\}$ $ \lex{<} \{a,c\}\{b\}\{c,d,e\}$ since
$|\{b,c\}|>|\{b\}|$.
\EOD
\end{example}

\begin{definition}[g-canonical step sequence]
A step sequence $x \in \E^*$ is {\em g-canonical} if for every step sequence  $y\in \E^*$, we have $\bigl(x\equiv y \wedge x\not=y\bigr) \implies x\lex{<}y$.\EOD
\label{g-can}
\end{definition}

In other words, $x$ is g-canonical if it is the least element in the g-comtrace $[x]$ with respect to the lexicographical ordering $\lex{<}$.

\begin{corollary}\begin{enumerate}
\item Each g-canonical step-sequence is in GMC-form.
\item For every step sequence $x \in \E^*$, there exists a unique g-canonical sequence $u\equiv x$. \END
\end{enumerate}
\label{col:g-can}
\end{corollary}

All of the concepts and results discussed so far in this section hold also for general equational monoids derived from the step sequence monoid (like those considered in \cite{JL}). We will now show that for both comtraces and traces, the  GMC-form, MC-form and g-canonical form correspond to the canonical form discussed in \cite{Foa,DJKL,JK95,JL}.

\subsection{Canonical representations of comtraces}
First note that comtraces are just g-comtraces with an empty relation$inl$, so all definitions for g-comtraces also hold for comtraces.

Let $\theta=(E,sim,ser)$ be a comtrace alphabet (i.e. $inl=\emptyset$)  and $\E$ be
the set of all steps over $\theta$.
In principle, $(a,b)\in ser$ means that the sequence $\{a\}\{b\}$ can be replaced by the set $\{a,b\}$ (and vice versa). We start with the definition of a relation between steps that allows such replacement.

\begin{definition}[forward dependency]
Let $\mathbb{FD} \subseteq \E \times \E$ be a relation comprising all pairs of steps $(A,B)$ such there exists a step $C\in \E$ such that $$ C\subseteq B \;\wedge\; A\times C \subseteq ser \;\wedge\;
C\times (B\setminus C)\subseteq ser.$$
The relation $\mathbb{FD}$ is called {\em forward dependency} on steps. \EOD
\label{def:FD}
\end{definition}

Note that in this definition $C\in \E$ implies $C\neq\emptyset$, but $C=B$ is allowed. The next result explains the name ``forward dependency'' of $\mathbb{FD}$. If $(A,B)\in \mathbb{FD}$, then some elements from $B$ can be moved to $A$ and the outcome will still be equivalent to $AB$.

\begin{lemma}
$(A,B)\in \mathbb{FD} \iff \bigl(\exists C\in \PSB{B}.\; (A\cup C)(B\setminus C) \equiv AB\bigr) \vee\; A\cup B\equiv AB$.
\label{FD}
\end{lemma}
\begin{proof}
($\Rightarrow$) If $C=B$ then $A\cup B\; \approx\; AB$ which implies
$A\cup B \equiv AB$. If $C\subset B$ and $C\neq\emptyset$ then we have
$(A\cup C)(B\setminus C)\approx AC(B\setminus C)\approx AB$, i.e. $(A\cup C)(B\setminus C) \equiv AB$.

($\Leftarrow$) Assume $A\cup B \equiv AB$. This means $A\cup B \in \E$ and consequently $A\cap B=\emptyset$, $A\times B\subseteq ser$. Let $a\in A$, $b\in B$.
By Proposition \ref{prop:p3}(6), $\{a,b\}=\pi_{\{a,b\}}(A\cup B)\equiv \pi_{\{a,b\}}(AB)=\{a\}\{b\}$. But $\{a,b\}\equiv\{a\}\{b\}$ means $(a,b)\in ser$.
Therefore $A\times B \subseteq ser$, i.e. $(A,B)\in \mathbb{FD}$.\\
Assume $C\subset B$, $C\neq\emptyset$ and $(A\cup C)(B\setminus C)\equiv AB$.
This implies $A\cup C \in \E$ and $A\cap C=\emptyset$. Let $a\in A$ and $c\in C$. By Proposition \ref{prop:p3}(6), $\{a,c\}=\pi_{\{a,c\}}(A\cup C)(B\setminus C)\equiv \pi_{\{a,c\}}(AB)=\{a\}\{c\}$. But $\{a,c\}\equiv\{a\}\{c\}$ means $(a,c)\in ser$. Hence $A\times C \subseteq ser$. Let $b\in B\setminus C$ and $c\in C$.
By Proposition \ref{prop:p3}(6), $\{c\}\{b\}=\pi_{\{b,c\}}(A\cup C)(B\setminus C)\equiv \pi_{\{b,c\}}(AB)=\{b,c\}$. Thus
$\{c\}\{b\}\equiv \{b,c\}$, which means $(c,b)\in ser$, i.e. $C\times(B\setminus C)\subseteq ser$. Hence $(A,B)\in \mathbb{FD}$.\END
\end{proof}

We will now recall the definition of a canonical step sequence for comtraces.

\begin{definition}[comtrace canonical step sequence \cite{JK95}]
\begin{sloppypar}
A step sequence $u=A_1\ldots A_k$ is {\em canonical} if we have
$(A_i,A_{i+1})\notin \mathbb{FD}$ for all $i$, $1\le i < k$.\EOD
\end{sloppypar}
\label{Can}
\end{definition}

The next results show that a canonical step sequence for comtraces is in fact ``greedy".

\begin{lemma}
For each non-empty canonical step sequence $u=A_1\ldots A_k$, we have $$A_1= \bset{a
\mid \exists w\in[u].\ w = C_1\ldots C_m \wedge a \in C_1}.$$
\label{l2}
\end{lemma}
\begin{proof}
Let $A= \set{a \mid \exists w\in[u].\ w = C_1\ldots C_m \wedge a \in C_1}$.
Since $u\in [u]$, $A_1 \subseteq A$.
We need to prove that $A\subseteq A_1$. Definitely $A=A_1$ if $k=1$, so assume $k>1$.
Suppose that $a\in A\setminus A_1$, $a\in A_j$, $1<j\leq k$ and $a\notin A_i$ for $i<j$.
Since $a\in A$, there is $v=Bx \in [u]$ such that $a\in B$.
Note that $A_{j-1}A_j$ is also {\em canonical} and $u'=A_{j-1}A_j = (u\RC (A_{j+1}\ldots A_k))\LC (A_1\ldots A_{j-2})$.
Let $v'= (v\RC (A_{j+1}\ldots A_k))\LC (A_1\ldots A_{j-2})$. We have $v'=B'x'$ where $a\in B'$.
By Corollary \ref{c2}, $u'\eqb v'$.
Since $u'=A_{j-1}A_j$ is canonical then $\exists c\in A_{j-1}.\; (c,a)\notin ser$ or
$\exists b\in A_j.\; (a,b)\notin ser$.
\begin{itemize}
 \item For the former case: $\pi_{\{a,c\}}(u')=\{c\}\{a\}$ (if $c\notin A_j$)
or $\pi_{\{a,c\}}(u')=\{c\}\{a,c\}$ (if $c\in A_j$). If $\pi_{\{a,c\}}(u')=\{c\}\{a\}$ then
$\pi_{\{a,c\}}(v')$ equals either $\{a,c\}$ (if $c\in B'$) or $\{a\}\{c\}$ (if $c\notin B'$), i.e.,
 in both cases $\pi_{\{a,c\}}(u')\not\eqb \pi_{\{a,c\}}(v')$, contradicting \propref{p3}(6).
If $\pi_{\{a,c\}}(u')=\{c\}\{a,c\}$ then
$\pi_{\{a,c\}}(v')$ equals either $\{a,c\}\{c\}$ (if $c\in B'$) or $\{a\}\{c\}\{c\}$ (if $c\notin B'$).
However in both cases $\pi_{\{a,c\}}(u')\not\eqb \pi_{\{a,c\}}(v')$, contradicting \propref{p3}(6).
For the latter case, let $d\in A_{j-1}$. Then $\pi_{\{a,b,d\}}(u')=\{d\}\{a,b\}$ (if $d\notin A_j$),
or $\pi_{\{a,b,d\}}(u')=\{d\}\{a,b,d\}$ (if $d\in A_j$). If $\pi_{\{a,b,d\}}(u')=\{d\}\{a,b\}$ then
$\pi_{\{a,b,d\}}(v')$ is one of the following $\{a,b,d\}$, $\{a,b\}\{d\}$, $\{a,d\}\{b\}$,
$\{a\}\{b\}\{d\}$ or $\{a\}\{d\}\{b\}$, and in either case  $\pi_{\{a,b,d\}}(u')\not\eqb \pi_{\{a,b,d\}}(v')$,
again contradicting \propref{p3}(6).
 \item If $\pi_{\{a,b,d\}}(u')=\{d\}\{a,b,d\}$, then we know
$\pi_{\{a,b,d\}}(v')$ is one of the following $\{a,b,d\}\{d\}$, $\{a,b\}\{d\}\{d\}$,
$\{a,d\}\{b,d\}$,
$\{a,d\}\{b\}\{d\}$, $\{a,d\}\{d\}\{b\}$,
$\{a\}\{b\}\{d\}\{d\}$, $\{a\}\{d\}\{b\}\{d\}$, or $\{a\}\{d\}\{d\}\{b\}$. However
in any of these cases we have $\pi_{\{a,b,d\}}(u')\not\eqb \pi_{\{a,b,d\}}(v')$,
contradicting \propref{p3}(6) as well.\END
\end{itemize}
\end{proof}

We will now show that for comtraces the canonical form from Definition \ref{Can} and GMC-form are equivalent, and that each comtrace has a unique canonical representation.

\begin{theorem}
A step sequence $u$ is in GMC-form if and only if it is canonical.
\label{GMC-can}
\end{theorem}
\begin{proof}
($\Leftarrow$) Suppose that $u=A_1\ldots A_k$ is canonical. By Lemma \ref{l2} we have that for each $B_1y_1\equiv A_1\ldots A_k$, $|B_1|\leq |A_1|$.
Since each $A_i\ldots A_k$ is also canonical, $A_2\ldots A_k$ is canonical so by Lemma \ref{l2} again we have that for each $B_2y_2\equiv A_2\ldots A_k$, $|B_2|\leq |A_2|$. And so on, i.e. $u=A_1\ldots A_k$ is in GMC-form.

($\Rightarrow$) Suppose that $u=A_1\ldots A_k$ is not canonical, and $j$ is the smallest number such that $(A_j,A_{j+1})\in \mathbb{FD}$. Hence
$A_1\ldots A_{j-1}$ is canonical, and, by ($\Leftarrow$) of this theorem, in GMC-form.
By Lemma \ref{FD}, either there is a non empty $C\subset A_{j+1}$ such that $(A_j\cup C)(A_{j+1}\setminus B) \equiv A_jA_{j+1}$, or $A_j\cup A_{j+1}\equiv A_jA_{j+1}$. In the first case since $C\neq\emptyset$, $|A_j\cup C|>|A_j|$; in the second case $|A_j\cup A_{j+1}|>|A_j|$, so
$A_j\ldots A_k$ is not in GMC-form, which means
$u=A_1\ldots A_k$ is not in GMC-form either. \END
\end{proof}

\begin{theorem}[implicit in \cite{JK95}]
For each step sequence $v$ there is a unique canonical step sequence $u$ such that $v\equiv u$.
\label{t4}
\end{theorem}
\begin{proof}
The existence follows from Proposition \ref{ExGMC} and Theorem \ref{GMC-can}. We only need to show uniqueness. Suppose that $u=A_1\ldots A_k$ and $v=B_1\ldots B_m$ are both canonical step sequences and $u\equiv v$. By induction on $k=|u|$ we will show that $u=v$.  By Lemma \ref{l2}, we have $B_1 = A_1$. If $k=1$, this ends the proof. Otherwise, let $u' = A_2 \ldots A_k$ and $w'= B_2 \ldots B_m$ and $u',v'$ are both canonical step sequences of $[u']$. Since $|u'|<|u|$, by the induction hypothesis, we obtain $A_i=B_i$ for $i=2,\ldots,k$ and $k=m$.\END
\end{proof}

The result of Theorem \ref{t4} was not stated explicitly in \cite{JK95}, but it can be derived from the results of Propositions 3.1, 4.8 and 4.9 of \cite{JK95}. However Propositions 3.1 and 4.8 of \cite{JK95} involve the concepts of partial orders and stratified order structures, while the proof of Theorem \ref{t4} uses only the algebraic properties of step sequences and comtraces. \\


Immediately from Theorems \ref{GMC-can} and \ref{t4} we get the following result.

\begin{corollary}
A step sequence $u$ is canonical if and only if it is g-canonical. \END
\label{can-gcan}
\end{corollary}

It turns out that for comtraces the canonical representation and MC-representation are also equivalent.

\begin{lemma}
If a step sequence $u$ is canonical and $u\equiv v$, then $\mbox{length}(u)\leq\mbox{length}(v)$.
\label{min}
\end{lemma}
\begin{proof}
By induction on $\len(v)$. Obvious for $\len(v)=1$ as then $u=v$.
Assume it is true for all $v$ such that $\len(v)\leq r-1$, $r\ge 2$. Consider $v=B_1B_2\ldots B_r$ and let $u=A_1A_2\ldots A_k$ be a canonical step sequence such that $v\equiv u$. Let $v_1=v\div_L A_1=C_1\ldots C_s$. By Corollary \ref{c2}(2), $v_1\equiv u\div_L A_1= A_2\ldots A_k$, and $A_2\ldots A_k$ is clearly canonical. Hence by induction assumption $k-1=\len(A_2\ldots A_k)\leq s$.
By Lemma \ref{l2}, $B_1\subseteq A_1$, hence $v_1=v\div_L A_1 = B_2\ldots B_r\div_L A_1=C_1\ldots C_s$, which means $s\leq r-1$.
Therefore $k-1\leq s\leq r-1$, i.e. $k\leq r$, which ends the proof. \END
\end{proof}

\begin{theorem}
A step sequence $u$ is maximally concurrent if and only if it is canonical.
\label{MCcan}
\end{theorem}
\begin{proof}
($\Leftarrow$) Let $u$ be canonical. From Lemma \ref{min} it follows the condition (1) of Definition \ref{MC} is satisfied. By Theorem \ref{GMC-can}, $u$ is in GMC-form, so the condition (2) of Definition \ref{MC} is satisfied as well.

($\Rightarrow$) By induction on $\len(u)$. It is obviously true for $u=A_1$. Suppose it is true for $\len(u)=k$. Let $u=A_1A_2\ldots A_kA_{k+1}$ be maximally concurrent. The step sequence $A_2\ldots A_{k+1}$ is also maximally concurrent and canonical by the induction assumption. If $A_1A_2\ldots A_{k+1}$ is not canonical, then
$(A_1,A_2)\in \mathbb{FD}$. By Lemma \ref{FD}, either there is non-empty $C\subset B$ such that $(A_1\cup C)(A_2\setminus C) \equiv A_1A_2$, or $A_1\cup A_2\equiv A_1B_2$.
Hence either $(A_1\cup C)(A_2\setminus C)A_3\ldots A_{k+1}\equiv A_1\ldots A_{k+1}=u$ or $(A\cup A_2)A_3\ldots A_{k+1}\equiv A_1\ldots A_{k+1}=u$.
The former contradicts the condition (2) of Definition \ref{MC}, the latter one contradicts the condition (1) of Definition \ref{MC}, so
$u$ is not maximally concurrent, which means $(A_1,A_2)\notin \mathbb{FD}$, so $u=A_1\ldots A_{k+1}$ is canonical. \END
\end{proof}

Summing up, as far as canonical representation is concerned, comtraces behave quite nicely. All three forms for g-comtraces, GMC-form, MC-form and g-canonical form, collapse to one comtrace canonical form if $inl=\emptyset$.

\subsection{Canonical representations of traces}
We will show that the canonical representations of traces are conceptually the same as the canonical representations of comtraces. The differences are merely ``syntactical'', as traces are sets of sequences, so ``maximal concurrency" cannot be expressed explicitly, while comtraces are sets of step sequences.\\

Let $(E,ind)$ be a trace alphabet and $( E^*/\!\!\equiv ,\circledast,[\lambda])$ be the corresponding
monoid of traces. A sequence $x=a_1 \ldots a_k\in E^*$ is called {\em fully commutative}
if $(a_i,a_j)\in ind$ for all $i\neq j$ and $i,j\in\set{1,\ldots,k}$.

\begin{corollary}
If $x=a_1 \ldots a_k\in E^*$ is fully commutative and $y=a_{i_1} \ldots a_{i_k}$ is any permutation of $a_1 \ldots a_k$, then $x\equiv y$. \END
\label{c:fc}
\end{corollary}

The above corollary could be interpreted as saying that if $x=a_1 \ldots a_k\in E^*$ is fully commutative than the set of events $\{a_1,\ldots ,a_k\}$ can be executed simultaneously.

\begin{definition}[greedy maximally concurrent form for traces \cite{Foa,DJKL}]
A sequence $x\in E^*$ is in {\em greedy maximally concurrent form} ({\em GMC-form}) if $x=\lambda$ or
$x=x_1 \ldots x_n$ such that
\begin{enumerate}
\item each $x_i$ is fully commutative, for $i=1,\ldots,n$,
\item for each $1\leq i \leq n-1$ and for each element $a$ of $x_{i+1}$ there exists an element $b$
of $x_i$ such that 
$(a,b)\notin ind$. \EOD
\end{enumerate}
\label{Can-Tr}
\end{definition}

Often the form from the above definition is called ``canonical" \cite{DJKL,JL92,JL}.


\begin{theorem}[\cite{Foa,DJKL}]
For every trace $\textbf{t} \in E^*/\!\!\equiv$, there exists $x \in E^*$ such that
$\textbf{t} = [x]$ and x is in the GMC-form. \END
\label{t1}
\end{theorem}

The GMC-form
as defined above is not unique, a trace may have more than one
GMC representation. For instance the trace ${\bf t}_1=[abcbca]$ from \exref{traces1}
has four GMC representations: $abcbca$, $acbbca$, $abccba$, and $acbcba$. The GMC-form is however unique when traces are represented as {\em vector firing sequences}\footnote{Vector firing sequences were introduced by Mike Shields in 1979 \cite{Shi} as an alternative representation of Mazurkiewicz traces.} \cite{DJKL,JL92,Shi}, where each fully commutative sequence is represented by a unique vector of events (so the name ``canonical" used in \cite{DJKL,JL92} is justified).
To get uniqueness for Mazurkiewicz traces, it suffices to order fully commutative sequences. For example,
we may introduce an arbitrary total order on $E$, extend it lexicographically to $E^*$
and add the condition that in the representation $x=x_1\ldots x_n$, each $x_i$ is minimal
w.r.t. the lexicographic ordering. The GMC-form with this additional condition
is called {\em Foata canonical form}.

\begin{theorem}[\cite{Foa}]
Every trace has a unique representation in the Foata canonical form. \END
\label{t2}
\end{theorem}

We will now show the relationship between GMC-form for traces and GMC-form (or canonical form) for comtraces.\\

Define $\E$, the set of steps generated by $(E,ind)$ as the set of all cliques of the graph of the relation $ind$, and for each fully commutative sequence $x=a_1\ldots a_n$, let $\st(x)=\{a_1,\ldots ,a_n\}\in \E$ be the step generated by $x$.

For each sequence $x=x_1\ldots x_k$ in GMC-form in $(E,ind)$, we call the step sequence $x^{\{max\}} = \st(x_1)\ldots \st(x_k) \in \E^*$, the {\em maximally concurrent step sequence representation} of $x$. Note that by Theorem \ref{t2}, the step sequence $x^{\{max\}}$ is unique. The name is formally justified by the following result (which also follows implicitly from \cite{DJKL}).

\begin{proposition}
\begin{enumerate}
\item A sequence $ x=x_1\ldots x_n $ is in GMC-form in $(E,ind)$ if and only if the step sequence $x^{\{ max \}} = \st(x_1)\ldots \st(x_k) $ is in GMC-form (or canonical form) in $(E,sim,ser)$ where
$sim=ser=ind$.
\item $[x]_{\equiv_{ind}} \TCT [x^{\{max \}}]_{\equiv_{ser}}$.
\end{enumerate}
\label{ct-cc}
\end{proposition}
\begin{proof}\begin{enumerate}
\item If $x=x_1\ldots x_n $ is not in GMC-form then by (2) of Definition \ref{Can-Tr}, there are $x_i, x_{i+1}$ and $b\in \st(x_{i+1})$ such that for all $a\in \st(x_i)$, $(a,b)\in ind$. Since $ser=ind$ this means that $(\st(x_1),\st(x_{i+1}))\in \mathbb{FD}$, so $x^{\{max \}}$ is not canonical.
Suppose that $x^{\{max \}}$ is not canonical, i.e. $(\st(x_1),\st(x_{i+1}))\in \mathbb{FD}$ for some $i$. This means there is a non-empty
$C\subseteq \st(x_{i+1})$ such that $\st(x_i)\times C\subseteq ser$ and $C\times (\st(x_{i+1})\setminus C) \subseteq ser$.
Let $a\in \st(x_i)$ and $b\in C\subseteq \st(x_{i+1})$. Since $ind=ser$, then $(a,b)\in ind$, so $x=x_1\ldots x_n $ is not in GMC-form.
\item Clearly $[x]_{\equiv_{ind}} \TCT [x^{\{\}}]_{\equiv_{ser}}$. Let $a_1\ldots a_n$ be a fully commutative sequence. Since $ser=ind$, $\{a_1\}\ldots \{a_n\} \equiv_{ser} \{a_1,\ldots ,a_n\}$. Hence, for
each sequence $x$, $x^{\{\}} \equiv_{ser} x^{\{max \}}$, i.e.
$[x^{\{\}}]_{\equiv_{ser}}=[x^{\{max \}}]_{\equiv_{ser}}$. \END
\end{enumerate}
\end{proof}

Hence we have proved that the GMC-form (or canonical form) for comtraces and GMC-form for traces are semantically identical concepts. They both describe the greedy maximally concurrent semantics, which for both comtraces and traces is also the global maximally concurrent semantics.

\section{Comtraces and Stratified Order Structures \label{sec:com2strat}}
In this section
we will recall the major result of \cite{JK95} that shows how comtraces define appropriate so-structures. 
We will start with the definition of \emph{$\lozenge$-closure} construction that plays a substantial role in most applications of so-structures for modelling concurrent systems (cf. \cite{JK95,KK,JL08,JL08b}).

\begin{definition}[diamond closure of relational structures \cite{JK95}]$\;$\\
Given a relational structure $S=(X,R_1,R_2)$, we define $S^\lozenge$, the {\em $\lozenge$-closure} of $S$, as
$$S^\lozenge \df \bigl(X, \prec_{R_1,R_2} , \sqsubset_{R_1,R_2} \bigr),$$
where $\prec_{R_1,R_2} \df (R_1\cup R_2)^* \circ R_1 \circ (R_1\cup R_2)^*$ and $\sqsubset_{R_1,R_2}\df (R_1\cup R_2)^* \setminus id_X$. \EOD
\label{d:SO-CL}
\end{definition}

The motivation behind the above definition is the following. For `reasonable' $R_1$ and $R_2$, the relational structure $(X,R_1,R_2)^\lozenge$ should satisfy the axioms S1--S4 of the so-structure definition.
Intuitively, $\lozenge$-closure is a generalization of the transitive closure constructions for relations to so-structures.
Note that if $R_1=R_2$ then $(X,R_1,R_2)^\lozenge=(X,R_1^+,R_1^+)$.
The following result shows that the properties of $\lozenge$-closure are   close to the appropriate properties of transitive closure.

\begin{theorem}[closure properties of $\lozenge$-closure \cite{JK95}]
For a relational structure $S=(X,R_1,R_2)$,
\begin{enumerate}
\item If $R_2$ is irreflexive, then $S\subseteq S^\lozenge$.
\item $\bigl(S^\lozenge\bigr)^\lozenge = S^\lozenge$.
\item $S^\lozenge$ is a so-structure if and only if
$\prec_{R_1,R_2} =(R_1\cup R_2)^* \circ R_1 \circ (R_1\cup R_2)^*$
is irreflexive.
\item If $S$ is a so-structure, then $S=S^\lozenge$.
\END
\end{enumerate}
\label{t:SO-CL}
\end{theorem}

Every comtrace is a set of equivalent step sequences and every step sequence represents a stratified order, so a comtrace can be interpreted as a set of equivalent stratified orders. From the theory presented in Section \ref{sec:relsurvey}  and the fact that comtrace satisfies paradigm $\pi_3$, it follows that this set of orders should define a so-structure, which should be called a so-structure defined by a given comtrace. On the other hand, with respect to a comtrace alphabet, every comtrace can be uniquely generated from any step sequence it contains. Thus, we will show that given a step sequence $u$ over a comtrace alphabet, \emph{without analyzing any other elements of the comtrace $[u]$ but $u$ itself}, we will be able to construct the same so-structure as the one defined by the whole comtrace. Formulations and proofs  of such results are  done in \cite{JK95} and depend heavily on the  $\lozenge$-closure construction and its properties. \\

Let $\theta=(E,sim,ser)$ be a comtrace alphabet, and let $u\in \E^*$ be a step sequence and let $\lhd_u\subseteq \Sigma_u\times\Sigma_u$ be the stratified order generated by $u$ as defined in Section~\ref{s:sts-so}.
Note that if $u\equiv w$ then $\Sigma_u=\Sigma_w$. Thus, for every comtrace ${\bf x}=[u]\in \E^*/\equiv$, we can define $\Sigma_{\bf x} = \Sigma_u$.

We will now show how
the $\lozenge$-closure operator is used to define a so-structure induced by a single step sequence
$u$.

\begin{definition}
Let $u\in \E^*$. We define the relations $\prec_u,\sqsubset_u \subseteq \Sigma_{\bf u}\times \Sigma_{\bf u}$
as:
\begin{enumerate}
\item
$\alpha \prec_u \beta \iffdf \alpha \lhd_u \beta \wedge (l(\alpha),l(\beta))\notin ser,$

\item
$ \alpha \sqsubset_u \beta \iffdf \alpha \lhd_u^\frown \beta \wedge (l(\beta),l(\alpha))\notin ser$. \EOD
\end{enumerate}

\label{all2}
\end{definition}

\begin{lemma}[{\cite[Lemma 4.7]{JK95}}]
For all $u, v\in \E^*$, if $u\equiv v$, then $\prec_u=\prec_v$ and $\sqsubset_u=\sqsubset_v$. \END
\label{l:inv}
\end{lemma}

Definition~\ref{all2} together with Lemma~\ref{l:inv} describes two basic \emph{local} invariants of the elements of $\Sigma_{\bf u}$. The relation $\prec_u$ captures the situation when $\alpha$ {\em always precedes}
$\beta$, and the relation $\sqsubset_u$ captures the situation when $\alpha$ {\em never follows}
$\beta$.

\begin{definition}
Given a comtrace ${\bf u}=[u]\in \E^*\!/\!\!\equiv$. We define 
\begin{align*}
S^{\{u\}} &\df \bigl( \Sigma_{\bf u}, \prec_u ,\sqsubset_u \bigr)^\lozenge & S_{\bf u} \df \left(\Sigma_{\bf u}, \bigcap_{x\in {\bf u}}\lhd_x, \bigcap_{x\in {\bf u}}\lhd_x^\frown\right)
\end{align*}
\EOD
\label{def:all1}
\end{definition}

The relational structure $S^{\{u\}}$ is the so-structure induced by the single step sequence $u$ and $S_{\bf u}$ is the so-structure defined by the comtrace ${\bf u}$.  The following theorem justifies the names and summarizes some nontrivial results concerning the so-structures generated by comtraces.

\begin{theorem}[\cite{JK95,JK}]
For all $u,v\in \E^*$, we have
\begin{enumerate}
\item $S^{\{u\}}$ and $S_{[u]}$ are so-structures,
\item $u \equiv v \iff S^{\{u\}}=S^{\{v\}}$,
\item $S^{\{u\}}=S_{[u]}$,
\item $ext\bigl(S_{[u]}\bigr) = \bset{ \lhd_x \mid x \in [u] }$. \END
\end{enumerate}
\label{theo:all1}
\end{theorem}

\theoref{all1} states that the so-structures $S^{\{u\}}$ and $S_{[u]}$ from \defref{all1} are identical and their stratified extensions are exactly the elements of the comtrace $[u]$ with step sequences interpreted as stratified orders. However, from an algorithmic point of view, the definition of $S^{\{u\}}$ is   more interesting, since building the relations $\prec_u$ and $\sqsubset_u$ and getting their $\lozenge$-closure, which in turn can be reduced to computing transitive closure of relations, can be done   efficiently. In contrast,  a direct use of the $S_{[u]}$ definition requires precomputing up to exponentially many elements of the comtrace $[u]$.

Figure \ref{f:ct-sos} shows an example of a comtrace and the so-structure it generates.\\

\begin{figure}[h]

\setlength{\unitlength}{5mm}

\begin{picture}(11,8)(0,1)
\setlength{\unitlength}{5mm}

\put(1,4){\circle*{0.2}}
\put(1,7){\circle*{0.2}}
\put(4,4){\circle*{0.2}}
\put(4,7){\circle*{0.2}}
\dottedline[.]{0.25}(4,4)(1,7)(4,7)(1,4)

\put(0.9,7.2){\makebox(0,0)[br]{$a$}}
\put(0.9,3.8){\makebox(0,0)[tr]{$c$}}
\put(4.1,7.2){\makebox(0,0)[bl]{$b$}}
\put(4.1,4){\makebox(0,0)[tl]{$d$}}
\put(1.8,2){\makebox(0,0)[bl]{$sim$}}

\put(7,4){\circle*{0.2}}
\put(7,7){\circle*{0.2}}
\put(10,4){\circle*{0.2}}
\put(10,7){\circle*{0.2}}

\put(10,7){\vector(-1,-1){2.9}}
\put(7,7){\vector(1,0){2.9}}

\qbezier(7,7),(8.5,8),(10,7)
\put(7.4,7.2){\vector(-2,-1){0.2}}

\put(6.9,7.2){\makebox(0,0)[br]{$a$}}
\put(6.9,3.8){\makebox(0,0)[tr]{$c$}}
\put(10.1,7.2){\makebox(0,0)[bl]{$b$}}
\put(10.1,4){\makebox(0,0)[tl]{$d$}}
\put(7.8,2){\makebox(0,0)[bl]{$ser$}}

\end{picture}

\begin{picture}(0,0)(-14,1)

\put(2,3){\circle*{0.2}}
\put(2,6){\circle*{0.2}}
\put(2,9){\circle*{0.2}}
\put(5,3){\circle*{0.2}}
\put(5,9){\circle*{0.2}}

\put(2,9){\vector(0,-1){2.9}}
\put(2,6){\vector(0,-1){2.9}}
\put(5,9){\vector(0,-1){5.9}}
\put(2,9){\vector(1,-2){2.95}}
\put(2,6){\vector(1,-1){2.9}}

\put(5,9){\vector(-1,-2){2.95}}

\qbezier(2,9),(-0.5,6),(2,3)
\put(1.75,3.4){\vector(1,-2){0.2}}

\put(2.4,9.2){\makebox(0,0)[br]{$a^{(1)}$}}
\put(1.9,5.8){\makebox(0,0)[br]{$c^{(1)}$}}
\put(2.2,2.7){\makebox(0,0)[tr]{$a^{(2)}$}}

\put(5.1,9.2){\makebox(0,0)[bl]{$b^{(1)}$}}
\put(5.1,2.7){\makebox(0,0)[tl]{$d^{(1)}$}}

\put(4,1.5){\makebox(0,0)[tr]{$\prec$}}

\end{picture}

\begin{picture}(0,0)(-21,0.1)

\put(2,3){\circle*{0.2}}
\put(2,6){\circle*{0.2}}
\put(2,9){\circle*{0.2}}
\put(5,3){\circle*{0.2}}
\put(5,9){\circle*{0.2}}

\put(2,9){\vector(0,-1){2.9}}
\put(2,6){\vector(0,-1){2.9}}
\put(5,9){\vector(0,-1){5.9}}
\put(2,9){\vector(1,-2){2.95}}

\put(5,9){\vector(-1,-2){2.95}}
\put(5,9){\vector(-1,-1){2.9}}
\put(5,3){\vector(-1,0){2.9}}
\put(2,6){\vector(1,-1){2.9}}

\qbezier(2,9),(-0.5,6),(2,3)
\put(1.7,3.4){\vector(1,-2){0.2}}

\qbezier(2,3),(3.5,2),(5,3)
\put(4.8,2.8){\vector(2,1){0.2}}


\put(2.4,9.2){\makebox(0,0)[br]{$a^{(1)}$}}
\put(1.9,5.8){\makebox(0,0)[br]{$c^{(1)}$}}
\put(2.3,2.7){\makebox(0,0)[tr]{$a^{(2)}$}}

\put(5.1,9.2){\makebox(0,0)[bl]{$b^{(1)}$}}
\put(5.1,2.7){\makebox(0,0)[tl]{$d^{(1)}$}}

\put(4,1.5){\makebox(0,0)[tr]{$\sqsubset$}}

\end{picture}

\caption{An example of the relations $sim$, $ser$ on $E=\{a,b,c,d\}$, and the so-structure $(X,\prec,\sqsubset)$ defined by the comtrace $[ \{a,b\}\{c\}\{a,d\} ]_{\equiv_{ser}}=\bigl\{ \{a,b\}\{c\}\{a,d\}, \{a\}\{b\}\{c\}\{a,d\},\bigr.$ $\bigl.\{a\}\{b,c\}\{a,d\},\{b\}\{a\}\{c\}\{a,d\}\bigr\}$.
}
\label{f:ct-sos}

\end{figure}
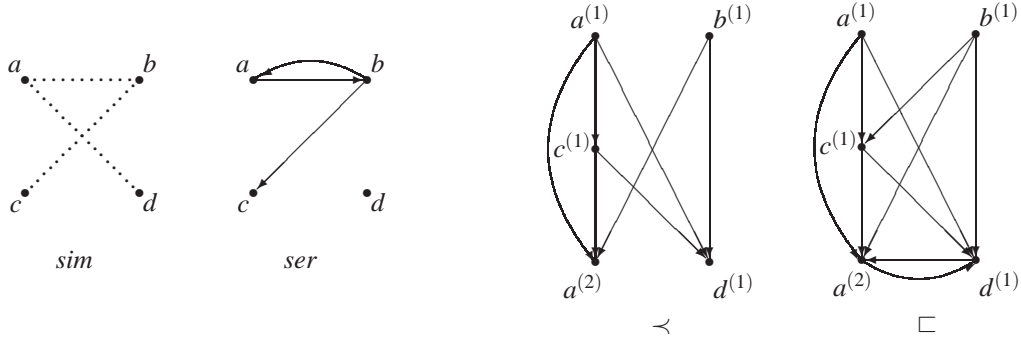

\section{Generalized Stratified Order Structures Generated by Generalized Comtraces\label{sec:gcom2strat}}
The relationship between g-comtraces and gso-structures is in principle the same as the relationship between comtraces and so-structures discussed in the previous section. Each g-comtrace uniquely determines a finite labeled gso-structure. However the formulations and proofs  of these analogue results for g-comtraces are   more complex. The difficulties are mainly due to the following facts:
\begin{itemize}
\item The definition of gso-structure is implicit, it involves using the induced so-structures (see \defref{gsos}), which makes
practically all definitions   more complex (especially the counterpart of $\lozenge$-closure), and the use of \theoref{SzpGStrat} more difficult than the use of \theoref{SzpStrat}.
\item The internal property expressed by Theorem \ref{pi3}, which says that $ext(S)$ conforms to paradigm $\pi_3$ of \cite{J4}, does not hold for gso-structures.
\item Generalized comtraces do not have a `natural' canonical form with a well understood interpretation.
\item The relation $inl$ introduces plenty of irregularities and  increases substantially the number of cases that need to be considered in many proofs.
\end{itemize}

In this section, we will prove the analogue of \theoref{all1} showing that every g-comtrace uniquely determines a finite gso-structure.

\subsection{Commutative closure of relational structures}
We will start with the notion of \emph{commutative closure} of a relational structure. It is an extension of the concept of $\lozenge$-closure (see Definition \ref{d:SO-CL}) which was used in \cite{JK95} and the previous section to construct finite so-structures from single step sequences or stratified orders.

\begin{definition}[commutative closure]\mbox{}\\
Let $G=(X,R_1,R_2)$ be any relational structure, and let $R_{3}= R_1\cap R_2^*$. Using the notation from Definition \ref{d:SO-CL}, the \emph{commutative closure} of the relational structure $G$  is defined as
\begin{center}
\hspace{0.3\textwidth}$G^{\ccl} = \Bigl(X, \sym{(\prec_{R_{3} R_2})} \cup R_1, \sqsubset_{R_{3} R_2}\Bigr).$\EOD               \end{center}
\label{def:ccl}
\end{definition}

The motivation behind the above definition is similar to that for $\lozenge$-closure: for `reasonable' $R_1$ and $R_2$, $(X,R_1,R_2)^{\ccl}$ should be a gso-structure.
Intuitively the $\ccl$-closure is also a generalization of transitive closure for relations.  Note that if $R_1=R_2$ then $(X,R_1,R_2)^{\ccl}=(X,\sym{(R_1^+)},R_1^+)$. Since the definition of gso-structures involves the definition of so-structures (see Definition \ref{def:gsos}), the definition of $\ccl$-closure uses the concept of $\lozenge$-closure.

Note that we do not have an equivalent of Theorem \ref{t:SO-CL} for $\ccl$-closure. The reason is that $\ccl$-closure is tailored to simplify the proofs  in the next section rather than to be a closure operator by itself. Nevertheless,  $\ccl$-closure satisfies some general properties of a closure operator.

The first property is the monotonicity of $\ccl$-closure.

\begin{proposition}
If $G_1=(X,R_1,R_2)$ and $G_2=(X,Q_1,Q_2)$ are two relational structures such that $G_1\subseteq G_2$, then $G_1^{\ccl}\subseteq G_2^{\ccl}$. 
\label{prop:monoccl}
\begin{proof}
Since $R_1\subseteq Q_1$ and $R_2\subseteq Q_2$ then $R_3\subseteq Q_3$, and $(X,R_3,R_2)^\Diamond \subseteq (X,Q_3,Q_2)^\Diamond$, 
i.e. \\$\prec_{R_3 R_2}\;\subseteq\;\prec_{Q_3 Q_2}$ and  
$\sqsubset_{R_3 R_2}\;\subseteq\;\sqsubset_{Q_3 Q_2}$, which immediately implies $G_1^{\ccl}\subseteq G_2^{\ccl}$ \END
\end{proof}

\end{proposition}

Another desirable property of $\ccl$-closure is that gso-structures are fixed points of  $\ccl$.

\begin{proposition}
If $G=(X,\com,\sq)$ is a gso-structure then $G=G^{\ccl}$.
\label{prop:cclgsos}
\end{proposition}
\begin{proof}
Since $G$ is a  gso-structure, by \defref{gsos}, $S_G=(X,\PO_G,\sq)$ is a so-structure. Hence, by Theorem \ref{t:SO-CL}(4), $S_G=S_G^{\Diamond}$, which implies $\sq=(\PO_G\cup\sq)^*\setminus id_X$. But since $S_G$ is a so-structure, $\PO_G\;\subseteq\;\sq$. So $\sq=\sq^*\setminus id_X$. Let $\PO=\com\cap\sq^*$. Then since $\com$ is irreflexive, 
$$\PO=\com\cap\sq^* = \com\cap\left(\sq^*\setminus id_X\right) = \com\cap\sq =\PO_G.$$
Hence, $(X,\PO,\sq)=(X,\PO_G,\sq)$ is a so-structure. By Theorem \ref{t:SO-CL}(4),  we know $(X,\PO,\sq)=\left(X,\PO,\sq\right)^{\Diamond}.$ So from \defref{ccl},  $G^{\ccl}=\left(X,\sym{\PO}\cup\com,\sq\right)$.
Since $\com$ is symmetric and $\PO\;\subseteq\; \com$, we have
$\sym{\PO}\cup\com= \com.$ Thus, $G=G^{\ccl}$.\END
\end{proof}

\subsection{Generalized stratified order structure generated by a step sequence}
We will now introduce a construction that derives a gso-structure from a single step sequence over a given g-comtrace alphabet. The idea of the construction is the same as $S^{\{u\}}$ from the previous section. First we construct some relational invariants and next we will use $\ccl$-closure in the similar manner as $\lozenge$-closure was used for $S^{\{u\}}$. However the construction is more elaborate and requires full use of the notation from Section \ref{s:sts-so} that allows us to define the formal relationship between step sequences and (labeled) stratified orders. We will also need the following two useful operators for relations.

\begin{definition}Let $R$ be a binary relation on $X$. We define the
\begin{itemize}
\item
\emph{symmetric intersection} of $R$ as
$R^{\,\Cap}\df R\cap R^{-1}$, and
\item
the \emph{complement} of $R$ as $\reco{R}\df (X\times X)\setminus R$.\EOD
\end{itemize}
\end{definition}

Let $\Theta=(E,sim,ser,inl)$ be a g-comtrace alphabet. Note that if $u\equiv w$ then $\Sigma_u=\Sigma_w$ so for every g-comtrace
${\bf s}=[s]\in \E^*/\!\equiv$, we can define $\Sigma_{\bf s} = \Sigma_s$.

\begin{definition}
Given a step sequence $s\in \E^*$.
\begin{enumerate}
\item
Let the relations $\com_s,\sq_s,\PO_s\subseteq\Sigma_{\bf s}\times\Sigma_{\bf s}$ be defined as follows:
\begin{align}
\alpha&\com_s\beta \iffdf  (l(\alpha),l(\beta))\in inl \label{eq:ss2gsos.1}\\
\alpha&\sq_s\beta \iffdf  \alpha\lhdf_s\beta \wedge (l(\beta),l(\alpha))\notin ser\cup inl  \label{eq:ss2gsos.2}\\
\alpha&\PO_s\beta \iffdf  \alpha \lhd_s \beta \notag\\
&\wedge \left(\begin{array}{ll}
	&(l(\alpha),l(\beta))\notin ser\cup inl \\
	\vee&(\alpha,\beta)\in\; \com_s\cap\left(\si{\sq_s^*}\circ\reco{\com}_s\circ \si{\sq_s^*}\right)\\
	\vee&\left(\begin{array}{ll}
			& (l(\alpha),l(\beta))\in ser\\
		\wedge	&\exists\delta,\gamma\in\Sigma_s.
		\left(\begin{array}{ll}
		&\delta\lhd_s\gamma\;\wedge\;(l(\delta),l(\gamma))\notin ser\\
		\wedge&\alpha\,\sq_s^*\,\delta\,\sq_s^*\,\beta\wedge\alpha\,\sq_s^*\,\gamma\,\sq_s^*\,\beta
		\end{array}\right)\end{array}\right)
	\end{array}\right) \label{eq:ss2gsos.3}
\end{align}
\item
The triple
$$G^{\{s\}} \df \left(\Sigma_s,\PO_s\cup\com_s,\PO_s\cup\sq_s\right)^{\ccl}$$ is called the {\em relational structure induced by the step sequence $s$}. \EOD
\end{enumerate}
\label{def:ss2gsos}
\end{definition}

The intuition of Definition \ref{def:ss2gsos} is similar to that of Definition \ref{all2}. Given a step sequence $s$ and g-comtrace alphabet $(E,sim,ser,inl)$, without analyzing any other elements of $[s]$ except $s$ itself, we would like to construct the gso-structure that is defined by the whole g-comtrace. So we will define appropriate ``local'' invariants $\com_s$, $\sq_s$ and $\PO_s$ from the sequence $s$.
\begin{enumerate}[(a)]
\item Equation \ref{eq:ss2gsos.1} is used to construct the relationship $\com_s$, where two
event occurrences $\alpha$ and $\beta$ might possibly be commutative because they
are related by the $inl$ relation.

\item Equation \ref{eq:ss2gsos.2} define the not later than relationship and this happens
when $\alpha$ occurs not later than $\beta$ on the step sequence $s$ and $\{\alpha,\beta\}$ cannot be
serialized into $\{\beta\}\{\alpha\}$, and $\alpha$ and $\beta$ are not commutative.

\item Equation \ref{eq:ss2gsos.3} is the most complicated one, since we want to take into
consideration the ``earlier than'' relationships which are not taken care of by
the commutative closure. There are three such cases:
\begin{enumerate}[(i)]
 \item $\alpha$ occurs before $\beta$ on the step sequence $s$, and two event occurrences $\alpha$ and $\beta$ cannot be put together into a single step ($(\alpha,\beta)\notin ser$) and are not commutative ($(\alpha,\beta)\notin inl$).
 \item $\alpha$ and $\beta$ are supposed to be commutative but they cannot be
commuted into $\beta$ and $\alpha$ because $\alpha$ is ``synchronous" with
some $\gamma$ and $\beta$ is ``synchronous'' with some $\delta$, and $(\gamma,\delta)$ is not in $inl$ (``synchronous'' in a sense that they must happen
simultaneously).
 \item $(\alpha, \beta)$ is in $ser$ but they can never be put
together into a single step because there are two distinct event occurrences $\delta$ and $\gamma$
which are ``squeezed''  between $\alpha$ and $\beta$ such that  
$(\delta,\gamma)\notin ser$, and thus  $\delta$ and $\gamma$ can never be put
together into a single step.
\end{enumerate}
\end{enumerate}

After building all of these ``local'' invariants from the step sequence $s$, all other ``global'' invariants which can be inferred from the axioms of the gso-structure definition are fully constructed by the commutative closure. \\

The next lemma will shows that the relations from $G^{\{s\}}$ really correspond to positional invariants of all the step sequences from the g-comtrace $[s]$.

\begin{lemma}
Let $s\in \E^*$,
$G^{\{s\}}=(\Sigma_s,\com,\sq)$, and $\PO=\com\cap\sq$. If $\alpha,\beta\in\Sigma_{{\bf s}}$, then
\begin{enumerate}
\item $\alpha\com\beta \iff  \forall u\in [s].\ pos_{u}(\alpha)\not= pos_{u}(\beta)$
\item $\alpha\sq\beta \iff \alpha\not=\beta\wedge \forall u\in [s].\ pos_{u}(\alpha)\le pos_{u}(\beta)$
\item $\alpha\PO\beta \iff \forall u\in [s].\ pos_{u}(\alpha)< pos_{u}(\beta)$
\item If $l(\alpha)=l(\beta)$ and $pos_s(\alpha)<pos_s(\beta)$, then $\alpha\PO\beta$. \END
\end{enumerate}
\label{lem:gposinv}
\end{lemma}

Eventhough the results of the above lemma are expected and look deceptively simple, the proof is long and highly technical and can be found in Appendix A. \\

Note that  \lemref{gposinv} also implies that we can construct the relational structure induced by the step sequence $G^{\{s\}}$ (we cannot claim that it is a gso-structure right now) if all the step sequences of a g-comtrace  are known. We will first show how to define the gso-structure induced from all the positional invariants of all the step sequences of a g-comtrace.

\begin{definition}
For every ${\bf s}\in \E^*/{\!\equiv}$, we define
$G_{{\bf s}} = \Bigl(\Sigma_{{\bf s}},\bigcap_{u\;\in\; {\bf s}}\sym{\lhd_u},\bigcap_{u\;\in\;{\bf s}}\lhd_u^\frown\Bigr)$. \EOD
\label{def:G-s}
\end{definition}


Note that Theorem \ref{theo:SzpGStrat} {\em does not} immediately imply that $G_{{\bf s}}$ is a gso-structure. It needs to be proved separately.

We will now show that given a step sequence $s$ over a g-comtrace alphabet, the definition of $G^{\{s\}}$ and the definition of $G_{[s]}$ yield exactly the same gso-structure.

\begin{theorem} Let $s \in \E^*$. Then $G^{\{s\}}= G_{[s]}$.
\label{theo:xi}
\end{theorem}
\begin{proof}
Let  $G^{\{s\}}=(\Sigma_s,\com,\sq)$ and $\alpha,\beta\in \Sigma_s$. Then by  \lemref{gposinv}(1, 2), we have
\[\begin{array}{ll}
\alpha\com\beta&\iff\forall u\in[s].\; pos_u(\alpha)\not= pos_u(\beta)\iff (\alpha,\beta)\in\bigcap_{u\;\in\;[s]}\sym{\lhd_u}\\
\alpha\;\sq\;\beta&\iff\bigl(\alpha\not=\beta\;\wedge\;\forall u\in[s].\; pos_u(\alpha)\le pos_u(\beta)\bigr)\iff(\alpha,\beta)\in\bigcap_{u\;\in\;[s]}\sym{(\lhd_u^\frown)}
\end{array}\]
Hence, $G^{\{s\}}=(\Sigma_s,\com,\sq)=\left(\Sigma_s,\bigcap_{u\;\in\;
[s]}\sym{\lhd_u},\bigcap_{u\;\in\;[s]}\lhd_u^\frown\right)=G_{[s]}$. \END
\end{proof}

We will next show that $G^{\{s\}}$ is indeed a gso-structure.

\begin{theorem}
Let $s \in \E^*$.
Then $G^{\{s\}}=\left(\Sigma_s,\com,\sq\right)$ is a gso-structure.
\label{theo:xi1}
\end{theorem}
\begin{proof}
Since  $\com=\bigcap_{u\;\in\; [s]}\sym{\lhd_u}$ and $\sym{\lhd_u}$ is irreflexive and symmetric, $\com$ is irreflexive and symmetric. Since $\sq=\bigcap_{u\;\in\;[s]}\lhd_u^\frown$ and $\lhd_u^\frown$ is irreflexive, $\sq$ is irreflexive.

Let $\PO=\com\cap\sq$, it remains to show that $S=\left(\Sigma,\PO,\sq\right)$  satisfies the conditions S1--S4 of \defref{sos}.  Since $\sq$ is irreflexive, S1 is satisfied. Since $\PO\;\subseteq\;\sq$, S2 is satisfied. Assume $\alpha\sq\beta\sq\gamma$ and $\alpha\not=\gamma$. Then
\begin{align*}
	&\alpha\sq\beta\sq\gamma\wedge  \alpha\not=\gamma &\\
\implies\;&{\textstyle (\alpha,\beta)\in\bigcap_{u\;\in\;[s]}\lhd_u^\frown \wedge\; (\beta,\gamma)\in\bigcap_{u\;\in\;[s]}\lhd_u^\frown} \;\wedge  \alpha\not=\gamma& \ttcomment{\theoref{xi}}\\
\implies\;&\forall u\in[s].\; pos_u(\alpha)\le pos_u(\beta)\le pos_u(\gamma) \;\wedge\;  \alpha\not=\gamma\hspace{2cm}& \ttcomment{Definition of $\lhd_u$}\\
\implies\;&\alpha\sq\gamma&\ttcomment{ \lemref{gposinv}(2)}
\end{align*}
Hence, S3 is satisfied. Next we assume that $\alpha\PO\beta\sq_s\gamma$. Then
\begin{align*}\textstyle
	&\alpha\PO\beta\sq\gamma&\\
\implies\;&{\textstyle(\alpha,\beta)\in\bigcap_{u\;\in\;[s]}(\lhd_u^\frown\cap\sym{\lhd_u}) \wedge (\beta,\gamma)\in\bigcap_{u\;\in\;[s]}(\lhd_u^\frown\cap\sym{\lhd_u})}&\ttcomment{\theoref{xi}}\\
\implies\;&\bigl(\forall u\in[s].\; pos_u(\alpha)\le pos_u(\beta)\wedge pos_u(\alpha)\not= pos_u(\beta)\bigr)&\\
	& \quad\quad\;\wedge\; \bigl(\forall u\in[s].\; pos_u(\beta)\le pos_u(\gamma)\wedge pos_u(\beta)\not= pos_u(\gamma)\bigr)&\ttcomment{Definition of $\lhd_u$}\\
\implies\;&\forall u\in[s].\; pos_u(\alpha)< pos_u(\gamma)&\\
\implies\;& \alpha\PO\gamma &\ttcomment{ \lemref{gposinv}(3)}
\end{align*}
Similarly, we can show $\alpha\sq\beta\PO\gamma\implies\alpha\PO\gamma$. Thus, S4 is satisfied. \END
\end{proof}

Theorem \ref{theo:xi1} justifies the following definition.
\begin{definition}
For every step sequence $s$, $G^{\{s\}} =\bigl(\Sigma_s,\PO_s\cup\com_s,\PO_s\cup\sq_s\bigr)^{\ccl} $ is the {\em gso-structure induced by $s$}. \EOD
\end{definition}

At this point it is worth discussing the roles of the two different definitions of the gso-structures generated from a given g-comtrace. \defref{ss2gsos} allows us to build the gso-structure by looking at a single step sequence of the g-comtrace and its g-comtrace alphabet. On the other hand, to build the gso-structure from a g-comtrace using \defref{G-s}, we need to know either all the positional invariants or all elements of the g-comtrace. By \theoref{xi}, these two definitions are equivalent. However, in our proof, \defref{ss2gsos} is more convenient when we want to deduce the properties of the gso-structure defined from a single step sequence over a given g-comtrace alphabet.  On the other hand, \defref{G-s} will be used to reconstruct the gso-structure when positional invariants of a g-comtrace are known.

\subsection{Generalized stratified order structures generated by generalized comtraces}
\label{Gso-gcom}
In this section, we want to show that the construction from \defref{ss2gsos} indeed yields a gso-structure representation of comtraces. But before doing so, we need some preliminary results.

\begin{proposition}
Let $s\in \E^*$.
Then $\lhd_s\in ext\bigl(G^{\{s\}}\bigr)$.
\label{prop:sinext}
\end{proposition}
\begin{proof}
Let $G^{\{s\}}=(\Sigma,\com,\sq)$. By  \lemref{gposinv}, for all $\alpha,\beta\in \Sigma$,
\begin{align*}
\alpha \com \beta &\implies pos_{s}(\alpha)\not= pos_{s}(\beta)\implies \alpha\lhd_s \beta \;\vee\; \beta\lhd_s \alpha \implies \alpha\sym{\lhd_s} \beta\\
\alpha\sq\beta &\implies pos_{s}(\alpha)\le pos_{s}(\beta)\implies \alpha\lhd_s^\frown \beta
\end{align*}
Hence, by \defref{gsosext}, we get $\lhd_s\in ext\bigl(G^{\{s\}}\bigr)$.\END
\end{proof}

\begin{proposition}
Let $s\in \E^*$.
If $\lhd\in ext(G^{\{s\}})$, then there exists $u\in \E^*$  such that $\lhd=\lhd_u$.
\label{prop:gsextss}
\end{proposition}
\begin{proof}
Let $G^{\{s\}}=(\Sigma_s,\com,\sq)$ and $\Omega_\lhd=B_1\ldots B_k$. We will show that $u=l[B_1]\ldots l[B_k]$ is a step sequence such that $\lhd=\lhd_u$.

Suppose $\alpha,\beta\in B_i$ are two distinct event occurrences such that $(l(\alpha),l(\beta))\notin sim$. Then $pos_s(\alpha)\not=pos_s(\beta)$, which by  \lemref{gposinv} implies that $\alpha\com\beta$. Since $\lhd\in ext(G^{\{s\}})$, by \defref{gsosext}, $\alpha\lhd\beta$ or $\beta\lhd\alpha$ contradicting that $\alpha,\beta\in B_i$. Thus, we have shown for all $B_i$ ($1\le i\le k$),
\begin{equation}
\alpha,\beta\in B_i \wedge \alpha\not=\beta \implies  (l(\alpha),l(\beta))\in sim
\label{eq:gsextss.1}
\end{equation}
By \propref{samelab}(2) (in Appendix A), if $e^{(i)},e^{(j)}\in \Sigma_s$ and $i\not=j$ then $\forall u\in[s].\;pos_u\bigl(e^{(i)}\bigr)\not=pos_u\bigl( e^{(j)}\bigr)$. So it follows from  \lemref{gposinv}(1) that $e^{(i)}\com e^{(j)}$. Since $\lhd\in ext(G^{\{s\}})$, by \defref{gsosext},
\begin{equation}
\text{if } e^{(k_0)}\in B_k \text{ and } e^{(m_0)}\in B_m \text{, then }  k_0\not=m_0\iff k\not=m.
\label{eq:gsextss.2}
\end{equation}
From \eref{gsextss.1} it follows that $u$ is a step sequence over $\theta$. Also by \eref{gsextss.2}, $pos_u^{-1}[\set{i}]=B_i$  and $|l[B_i]|=|B_i|$ for all $i$. Hence, $\Omega_\lhd=\Omega_{\lhd_u}$, which implies $\lhd=\lhd_u$.\END
\end{proof}

We want to show that two step sequences over the same g-comtrace alphabet induce the same gso-structure iff they belong to the same g-comtrace (\theoref{xieqb} below). The proof of an analogous result for comtraces from \cite{JK95} is simpler because every comtrace has a unique natural canonical representation that is both greedy and maximally concurrent and can be easily constructed. Moreover the canonical representation for comtraces correspond to the unique greedy stratified extension of appropriate causality relation $\prec$ (see \cite{JK95}). Nothing similar holds for g-comtraces.
For g-comtraces both natural representations, GMC and MC, are not unique. The g-canonical representation (Definition \ref{g-can}) is unique but its uniqueness is artificial and induced by some  step sequence lexicographical order $\lex{<}$ (Definition
\ref{def:lexord}).  Nevertheless this  lexicographical order $\lex{<}$ will be the basic tool used in the next lemma. The lack of natural unique representation will make our reasoning a bit harder.

\begin{lemma}
Let $s$ be a step sequence over a g-comtrace alphabet $(E,ser,sim,inl)$ and $<_E$ be any total order on $E$. Let $u=A_1\ldots A_n$  be the g-canonical representation of $[s]$ (i.e.,  $u$ is the least element of the g-comtrace $[s]$ w.r.t.  $\lex{<}$). Let $G^{\{s\}}=(\Sigma,\com,\sq)$ and $\PO=\com\cap\sq$. For each $X\subseteq\Sigma$, let $\mins_{\PO}(X)$ denote the set of all minimal elements of $X$ w.r.t. $\PO$ and define
\begin{align*}
Z(X)\df \Bigl\lbrace Y \subseteq \mins_{\PO}(X) \Bigl\lvert \begin{array}{l}
	 \bigl(\forall \alpha,\beta\in Y.\; \neg(\alpha\com\beta)\bigr)
\wedge  \bigl(\forall \alpha\in Y\;\forall \beta\in X\setminus Y.\; \neg (\beta \sq \alpha)\bigr)
\end{array}\Bigr. \Bigr\rbrace
\end{align*}

Let $\h{u}=\h{A_1}\ldots \h{A_n}$ be the enumerated step sequence of $u$. Then $A_i$ is the least element of the set $\bset{l[Y]\mid Y\in Z(\Sigma\setminus\Al\bigl(\h{A_1}\ldots\h{A_{i-1}})\bigr)}$ w.r.t. the ordering $\stor{<}$. \END
\label{lem:semican}
\end{lemma}

Before presenting the proof, we will explain the intuition behind the definition of the set $Z(X)$. Let us consider $Z(\Sigma)$ first. Then $A_1$ in this lemma is the least element of the set $\set{l[Y]\mid Y\in Z(\Sigma)}$ w.r.t. the ordering $\stor{<}$. Our goal is to construct $A_1$ by looking only at the gso-structure $G$ without having to construct up to exponentially many stratified extensions of $G$. The most technical part of this proof is to show that $\h{A_1}$ actually belongs to the set $Z(\Sigma)$. Recall that to show that 
$Y\in Z(\Sigma)$ satisfies, we want to show that $Y$ satisfies the following conditions:
\begin{enumerate}[i.]
\item
 no two elements in $Y$ are commutative,
\item
 for an element $\alpha \in Y$ and $\beta \in \Sigma\setminus Y$, it is not
the case that $\beta$ is not later than $\alpha$.
\end{enumerate}

Note that we actually define $Z(X)$ instead of $Z(\Sigma)$, because we want to apply
it successively to build {\em all} the steps $A_{i}$ of the g-canonical representation $u$ of $G^{\{s\}}$. This lemma can be seen
as an algorithm to build the g-canonical representation of $[s]$ by looking only at $G^{\{s\}}$.

\begin{proof}[of Lemma \ref{lem:semican}]
First notice that by \lemref{gposinv}(3), for every nonempty $X\subseteq\Sigma$, since $\Sigma$ is finite, we know that  $\mins_{\PO}(X)$ is nonempty and finite. Furthermore by  \lemref{gposinv}(4), if $e^{(i)},e^{(j)}\in \Sigma$ and $i<j$, then  $e^{(i)}\PO e^{(j)}$. Hence, for all $\alpha,\beta\in \mins_{\PO}(X)$, where $X\subseteq \Sigma$, we have $l(\alpha)\not=l(\beta)$. This ensures that if $Y\in Z(X)$ and $X\subseteq \Sigma$, then $|Y|=|l[Y]|$.

For all $\alpha\in\h{A_1}$ and $\beta\in\Sigma$, $pos_s(\beta)\ge pos_s(\alpha)$. Hence, by  \lemref{gposinv}(3), $\neg(\beta\PO\alpha)$. Thus,
\begin{equation}
\h{A_1}\subseteq \mins_{\PO}(\Sigma)
\label{eq:semican.1}
\end{equation}
For all $\alpha,\beta\in\h{A_1}$, since $pos_s(\beta)= pos_s(\alpha)$, by  \lemref{gposinv}(1), we have
\begin{equation}
\neg(\alpha\com\beta)
\label{eq:semican.2}
\end{equation}
For any $\alpha\in\h{A_1}$ and $\beta\in \Sigma\setminus\h{A_1}$, since $pos_s(\beta)< pos_s(\alpha)$, by  \lemref{gposinv}(2),
\begin{equation}
\neg (\beta \sq \alpha)
\label{eq:semican.3}
\end{equation}
From \eref{semican.1}, \eref{semican.2} and \eref{semican.3}, we know that $\h{A_1}\in Z(\Sigma)$. Hence, $Z(\Sigma)\not=\emptyset$. This ensures the  least element of $\set{l[Y]\mid Y\in Z(\Sigma)}$ w.r.t. $\stor{<}$ is well defined.

Let $Y_0\in Z(\Sigma)$ such that $B_0=l[Y_0]$ is the least element of $\set{l[Y]\mid Y\in Z(\Sigma)}$ w.r.t.  $\stor{<}$. We want to show that $A_1=B_0$. Since $\stor{<}$ is a total order, we know that $A_1\stor{<}B_0$ or $B_0\stor{<}A_1$ or $A_1=B_0$. But since $\h{A_1}\in Z(\Sigma)$ and $B_0$ be the least element of the set $\set{l[B]\mid B\in Z(\Sigma)}$, $\neg(A_1\stor{<}B_0)$. Hence, to show that $A_1=B_0$, it suffices to show $\neg(B_0\stor{<}A_1)$.

Suppose that $B_0\stor{<}A_1$. We first want to show that for every nonempty $W\subseteq Y_0$ there is an enumerated step sequence $v$ such that
\begin{equation}
\h{v}=W_0\h{v_0}\eqb \h{A_1}\ldots \h{A_n}\text{ and } W\subseteq W_0\subseteq Y_0
\label{eq:semican.g1}
\end{equation}
We will prove this by induction on $|W|$. 

\paragraph{Base case} When $|W|=1$, we let $\set{\alpha_0}=W$. We choose $\h{v_1}=\h{E_0}\ldots\h{E_k}\h{y_1}\eqb \h{A_1}\ldots \h{A_n}$ and $\alpha_0\in\h{E_k}$ $(k\ge 0)$ such that  for all $\h{v'}=\h{E'_0}\ldots\h{E'_{k'}}\;\h{y'_1}\eqb \h{A_1}\ldots \h{A_n}$ and $\alpha_0\in\h{E'_{k'}}$, we have
\begin{enumerate}[(i)]
\item $\wei(\h{E_0}\ldots\h{E_k})\le \wei(\h{E'_0}\ldots\h{E'_{k'}})$, and
\item $\wei(\h{E_{k-1}}\;\h{E_k})\le \wei(\h{E'_{k'-1}}\;\h{E'_{k'}})$.
\end{enumerate}
We then consider only $\h{w}=\h{E_0}\ldots\h{E_k}$. We observe by the way we chose $\h{v_1}$, we have
$\forall \beta\in \Al(\h{w}).\bigl(\beta\not=\alpha_0 \implies \forall t\in[w].\; pos_t(\beta)\le pos_t(\alpha_0)\bigr).$
Hence, since $\h{w}=\h{u}\RC\h{v_0}$, it follows from \propref{invsubs}(1,2) that
\[\forall \beta\in \Al(\h{w}).\Bigl(\beta\not=\alpha_0 \implies\forall t\in[A_1\ldots A_n].\; pos_t(\beta)\le pos_t(\alpha_0)\Bigr)
\]
Then it follows from  \lemref{gposinv}(2) that $\forall \beta\in \Al(\h{w}).\; (\beta\not=\alpha_0 \implies \beta\sq\alpha_0)$. But by the way $Y_0$ was chosen, we know that $\forall \alpha\in Y_0.\;\forall \beta\in \Sigma\setminus Y_0.\; \neg (\beta \sq \alpha)$. Hence,
\begin{equation}
\Al(\h{w})=(\h{E_0}\cup\ldots\cup\h{E_k})\subseteq Y_0
\label{eq:semican.5}
\end{equation}

We next want to show
\begin{equation}
\forall\alpha\in \h{E_i}.\forall\beta\in\h{E_j}.\;\set{\alpha}\set{\beta}\eqb\set{\alpha,\beta} \quad\quad (0\le i<j\le k)
\label{eq:semican.6}
\end{equation}
Suppose not. Then either $[\set{\alpha}\set{\beta}]=\set{\set{\alpha}\set{\beta}}$ or $[\set{\alpha}\set{\beta}]=\set{\set{\alpha}\set{\beta},\set{\beta}\set{\alpha}}$. In either case,  we have $\forall t\in[\set{l(\alpha)}\set{l(\beta)}].\;pos_t(\alpha)\not=pos_t(\beta)$. Since $\set{\alpha}\set{\beta}\eqb\pi_{\set{\alpha,\beta}}(\h{u})$, by \propref{invsubs}(3), $\forall t\in[u].\;pos_t(\alpha)\not=pos_t(\beta)$. So by  \lemref{gposinv}, $\alpha\com\beta$. This contradicts that $Y_0\in Z(\Sigma)$ and $\alpha,\beta\in \Sigma(\h{w})\subseteq Y_0$. Thus, we have shown \eref{semican.6}, which implies  that for all  $\alpha\in \h{E_i}$ and $\beta\in\h{E_j}$ ($0\le i<j\le k$), $(l(\alpha),l(\beta))\in ser$. Then $\h{E_0}\ldots\h{E_k}\eqb\bigcup_{i=0}^{k}\h{E_i}$. Hence, by \eref{semican.5} and \eref{semican.6}, there exists a step sequence $v''_1$ such that $\h{v''_1}=\left(\bigcup_{i=0}^{k}\h{E_i}\right)\h{y_1}\eqb \h{A_1}\ldots \h{A_n}$
and $\set{\alpha_0}\subseteq \bigcup_{i=0}^{k}\h{E_i}\subseteq Y_0$.

\paragraph{Inductive step} When $|W|>1$, we pick an element $\beta_0\in W$. By applying the induction hypothesis on $W\setminus \set{\beta_0}$, we get a step sequence $v_2$ such that $\h{v_2}=\h{F_0}\h{y_2}\eqb \h{A_1}\ldots \h{A_n}$ where $W\setminus \set{\beta_0}\subseteq\h{F_0}\subseteq Y_0$. If $W\subseteq \h{F_0}$, we are done. Otherwise, proceeding like the base case, we construct a step sequence $v_3$ such that $\h{v_3}=\h{F_0}\;\h{F_1}\h{y_3}\eqb \h{A_1}\ldots \h{A_n}$ and $\set{\beta_0}\subseteq\h{F_1}\subseteq Y_0$. Since  $\h{F_0}\subseteq Y_0$, we have $W\subseteq \h{F_0}\cup\h{F_1}\subseteq Y_0$. Then similarly to how we proved \eref{semican.6}, we can show that
$\forall\alpha\in \h{F_0}.\;\forall \beta\in\h{F_1}.\;\set{\alpha}\set{\beta}\eqb\set{\alpha,\beta}$.
This means that for all $\alpha\in \h{F_0}$ and $\beta\in\h{F_1}$, $(l(\alpha),l(\beta))\in ser$. Hence, $\h{F_0}\h{F_1}\eqb \h{F_0}\cup\h{F_1}$. Hence, there is a step sequence $v_4$ such that $\h{v_4}=(\h{F_0}\cup\h{F_1})\;\h{y_4}\eqb \h{A_1}\ldots \h{A_n}$ and $W\subseteq (\h{F_0}\cup\h{F_1}) \subseteq Y_0$.\\

Thus, we have shown \eref{semican.g1}. So by choosing $W=Y_0$, we get a step sequence $v$ such that
$\h{v}=W_0\h{v_0}\eqb \h{A_1}\ldots \h{A_n}$ and $Y_0\subseteq W_0\subseteq Y_0$. Hence, $\h{v}=W_0\h{v_0}\eqb \h{A_1}\ldots \h{A_n}$. Thus, $v=B_0 v_0\eqb A_1\ldots A_n$. But since $B_0\stor{<}A_1$, this contradicts the fact that $A_1\ldots A_n$ is the least element of $[s]$ w.r.t. $\lex{<}$. Hence, $A_1$ is the least element of $\set{l[Y]\mid Y\in Z(\Sigma)}$ w.r.t.  $\stor{<}$.\\

We now prove that $A_i$ is the least element of
the set $\bset{l[Y]\mid Y\in Z(\Sigma\setminus\Al\bigl(\h{A_1}\ldots\h{A_{i-1}})\bigr)}$
w.r.t. $\stor{<}$ by induction on $n$, the number of steps of the g-canonical step sequence $u=A_1\ldots A_n$. If $n=0$, we are done. If $n>0$, then we have just shown that $A_1$ is the least element of $\bigl\{l[Y]\mid Y\in Z(\Sigma)\bigr\}$ w.r.t.  $\stor{<}$. By applying the induction hypothesis on $p=\h{A_2}\ldots\h{A_n}$, $\Sigma_{p}=\Sigma\setminus\h{A_1}$, and its gso-structure $(\Sigma_{p},\com\cap (\Sigma_{p}\times \Sigma_{p}),\sq\cap(\Sigma_{p}\times \Sigma_{p}))$, we get  $A_i$ is the least element of the set
$\bset{l[Y]\mid Y\in Z(\Sigma\setminus\Al\bigl(\h{A_1}\ldots\h{A_{i-1}})\bigr)}$
 w.r.t. $\stor{<}$ for all $i\ge 2$.
\END
\end{proof}

\begin{theorem}
Let $s$ and $t$ be step sequences over a g-comtrace alphabet $(E,sim,ser,inl)$. Then $s\eqb t$ iff $G^{\{s\}}=G^{\{t\}}$.
\label{theo:xieqb}
\end{theorem}
\begin{proof} ($\Rightarrow$) If $s\eqb t$, then $[s]=[t]$. Hence, by \theoref{xi},
$G^{\{s\}}=G^{\{t\}}$.

($\Leftarrow$) By \lemref{semican},  we can use $G^{\{s\}}$ to construct a unique element $w_1$ such that $w_1$ is the least element of $[s]$ w.r.t. $\lex{<}$, and then use $G^{\{t\}}$ to construct a unique element $w_2$ that is the least element of $[t]$ w.r.t. $\lex{<}$. But since $G^{\{s\}}=G^{\{t\}}$, we get $w_1=w_2$. Hence, $s\eqb t$. \END
\end{proof}

Theorem \ref{theo:xieqb} justifies the following definition:
\begin{definition}
For every g-comtrace $[s]$, $G_{[s]}=G^{\{s\}} =\left(\Sigma_s,\PO_s\cup\com_s,\PO_s\cup\sq_s\right)^{\ccl} $ is the {\em gso-structure induced by the g-comtrace $[s]$}. \EOD
\end{definition}

To end this section, we prove two major results. \theoref{t7} says that the stratified extensions of the gso-structure induced by a g-comtrace $[t]$ are exactly those generated by the step sequences in $[t]$. \theoref{t7extra} says that the gso-structure induced by a g-comtrace is uniquely identified by any of its stratified extensions.

\begin{lemma}
Let $s,t \in \E^*$
and $\lhd_s\in ext(G^{\{t\}})$. Then $G^{\{s\}}=G^{\{t\}}$. \END
\label{lem:ext2xi}
\end{lemma}

The proof of the above lemma uses \defref{ss2gsos} heavily and thus requires a separate analysis of many cases and was moved to Appendix B.

\begin{theorem}
Let $s,t \in \E^*$.
Then $ext(G^{\{s\}})=\set{\lhd_u\mid u\in[s]}$.
\label{theo:t7}
\end{theorem}
\begin{proof} ($\subseteq$) Suppose $\lhd\in ext(G^{\{s\}})$. By  \propref{gsextss}, there is a step sequence $u$ such that $\lhd_u=\lhd$. Hence, by \lemref{ext2xi}, we have $G^{\{u\}}=G^{\{s\}}$, which by \theoref{xieqb} implies that $u\eqb s$. Hence, $ext(G^{\{s\}})\subseteq\set{\lhd_u\mid u\in[s]}$.

($\supseteq$) If $u\in[s]$, then it follows from \theoref{xieqb} that $G^{\{u\}}=G^{\{s\}}$. This and \propref{sinext} imply $\lhd_u\in ext(G^{\{s\}})$. Hence, $ext(G^{\{s\}})\supseteq\set{\lhd_u\mid u\in[s]}$.\END
\end{proof}

\begin{theorem}
Let $s,t \in \E^*$
and $ext(G^{\{s\}})\cap ext(G^{\{t\}})\not=\emptyset$. Then $s\eqb t$.
\label{theo:t7extra}
\end{theorem}
\begin{proof}
Let $\lhd\in ext(G^{\{s\}})\cap ext(G^{\{t\}})$. By \propref{gsextss}, there is a step sequence $u$ such that $\lhd_u=\lhd$. By \lemref{ext2xi}, we have $G^{\{s\}}=G^{\{u\}}=G^{\{t\}}$. This and \theoref{xieqb} yields $s\eqb t$.\END
\end{proof}

Summing up, we have proved the analogue of \theoref{all1} for g-comtraces. In fact, \theoref{all1} is a straightforward consequence of this section for $inl=\emptyset$.

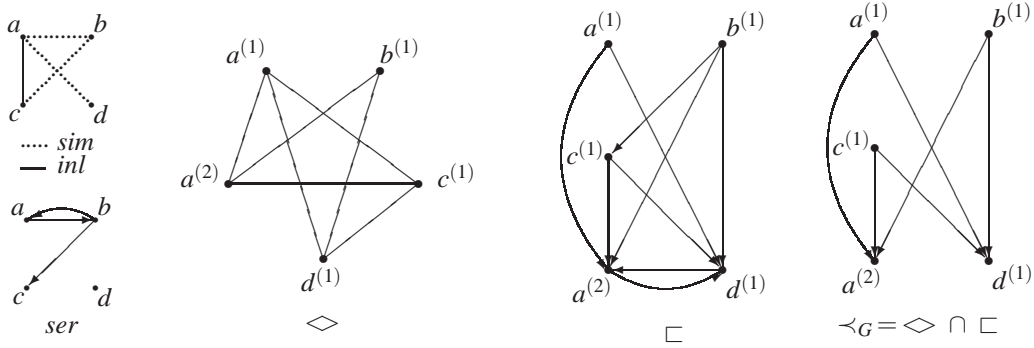
\begin{figure}

\setlength{\unitlength}{5mm}

\begin{picture}(11,7)(0,-1)
\setlength{\unitlength}{3mm}

\put(1,4){\circle*{0.2}}
\put(1,7){\circle*{0.2}}
\put(4,4){\circle*{0.2}}
\put(4,7){\circle*{0.2}}
\dottedline[.]{0.25}(4,4)(1,7)(4,7)(1,4)
\drawline(1,4)(1,7)

\put(0.9,7.2){\makebox(0,0)[br]{$a$}}
\put(0.9,3.8){\makebox(0,0)[tr]{$c$}}
\put(4.1,7.2){\makebox(0,0)[bl]{$b$}}
\put(4.1,4){\makebox(0,0)[tl]{$d$}}

\dottedline[.]{0.25}(1.0,2.2)(2.0,2.2)
\put(2.5,2){\makebox(0,0)[bl]{$sim$}}

\drawline(0.9,1.2)(2.0,1.2)
\put(2.5,1){\makebox(0,0)[bl]{$inl$}}

\end{picture}

\begin{picture}(0,0)(3.5,3)
\setlength{\unitlength}{3mm}

\put(7,4){\circle*{0.2}}
\put(7,7){\circle*{0.2}}
\put(10,4){\circle*{0.2}}
\put(10,7){\circle*{0.2}}

\put(10,7){\vector(-1,-1){2.9}}
\put(7,7){\vector(1,0){2.9}}

\qbezier(7,7),(8.5,8),(10,7)
\put(7.4,7.2){\vector(-2,-1){0.2}}

\put(6.9,7.2){\makebox(0,0)[br]{$a$}}
\put(6.9,3.8){\makebox(0,0)[tr]{$c$}}
\put(10.1,7.2){\makebox(0,0)[bl]{$b$}}
\put(10.1,4){\makebox(0,0)[tl]{$d$}}
\put(7.8,2){\makebox(0,0)[bl]{$ser$}}

\end{picture}

\begin{picture}(0,0)(-5,2)

\put(1,5){\circle*{0.2}}
\put(2,8){\circle*{0.2}}
\put(3.5,3){\circle*{0.2}}
\put(5,8){\circle*{0.2}}
\put(6,5){\circle*{0.2}}
\drawline(2,8)(6,5)
\drawline(5,8)(1,5)(2,8)(3.5,3)
\drawline(1,5)(6,5)(3.5,3)(5,8)

\put(2,8.2){\makebox(0,0)[br]{$a^{(1)}$}}
\put(5,8.2){\makebox(0,0)[bl]{$b^{(1)}$}}
\put(0.8,4.9){\makebox(0,0)[br]{$a^{(2)}$}}

\put(6.5,4.9){\makebox(0,0)[bl]{$c^{(1)}$}}

\put(4,2.8){\makebox(0,0)[tr]{$d^{(1)}$}}

\put(3,1.5){\makebox(0,0)[tl]{$\com$}}

\end{picture}

\begin{picture}(0,0)(-21,1.2)
\setlength{\unitlength}{5mm}

\put(2,3){\circle*{0.2}}
\put(2,6){\circle*{0.2}}
\put(2,9){\circle*{0.2}}
\put(5,3){\circle*{0.2}}
\put(5,9){\circle*{0.2}}

\put(2,6){\vector(0,-1){2.9}}
\put(5,9){\vector(0,-1){5.9}}
\put(2,9){\vector(1,-2){2.95}}
\put(2,6){\vector(1,-1){2.9}}

\put(5,9){\vector(-1,-2){2.95}}

\qbezier(2,9),(-0.5,6),(2,3)
\put(1.75,3.4){\vector(1,-2){0.2}}

\put(2.4,9.2){\makebox(0,0)[br]{$a^{(1)}$}}
\put(1.9,5.8){\makebox(0,0)[br]{$c^{(1)}$}}
\put(2.2,2.8){\makebox(0,0)[tr]{$a^{(2)}$}}

\put(5.1,9.2){\makebox(0,0)[bl]{$b^{(1)}$}}
\put(5.1,2.9){\makebox(0,0)[tl]{$d^{(1)}$}}

\put(1,1.5){\makebox(0,0)[tl]{$\prec_G\:=\:\com\;\cap\;\sqsubset$}}

\end{picture}

\begin{picture}(0,0)(-14,0.6)

\put(2,3){\circle*{0.2}}
\put(2,6){\circle*{0.2}}
\put(2,9){\circle*{0.2}}
\put(5,3){\circle*{0.2}}
\put(5,9){\circle*{0.2}}

\put(2,6){\vector(0,-1){2.9}}
\put(5,9){\vector(0,-1){5.9}}
\put(2,9){\vector(1,-2){2.95}}

\put(5,9){\vector(-1,-2){2.95}}
\put(5,9){\vector(-1,-1){2.9}}
\put(5,3){\vector(-1,0){2.9}}
\put(2,6){\vector(1,-1){2.9}}

\qbezier(2,9),(-0.5,6),(2,3)
\put(1.7,3.4){\vector(1,-2){0.2}}

\qbezier(2,3),(3.5,2),(5,3)
\put(4.8,2.8){\vector(2,1){0.2}}


\put(2.4,9.2){\makebox(0,0)[br]{$a^{(1)}$}}
\put(1.9,5.8){\makebox(0,0)[br]{$c^{(1)}$}}
\put(2.1,2.8){\makebox(0,0)[tr]{$a^{(2)}$}}

\put(5.1,9.2){\makebox(0,0)[bl]{$b^{(1)}$}}
\put(5.1,2.9){\makebox(0,0)[tl]{$d^{(1)}$}}

\put(4,1.5){\makebox(0,0)[tr]{$\sqsubset$}}

\end{picture}

\caption{A g-comtrace alphabet $(E,sim,ser,inl)$, where $E=\{a,b,c,d\}$, the gso-structure $G=(X,\com,\sqsubset)$ and $\prec_G=\com\cap\sqsubset$ defined by the g-comtrace
$[ \{a,b\}\{c\}\{a,d\} ]=\bset{ \{a,b\}\{c\}\{a,d\},$ $ \{a\}\{b\}\{c\}\{a,d\},
\{a\}\{b,c\}\{a,d\}, \{b\}\{a\}\{c\}\{a,d\},
\{b\}\{c\}\{a\}\{a,d\},$ $\{b,c\},\{a\}\{a,d\}}$.}
\label{f:gct-gos}

\end{figure}

Figure \ref{f:gct-gos} shows an example of a g-comtrace and the gso-structure it generates.

\section{Conclusion and Future Work \label{sec:conclusion}}
The comtrace concept is revisited and its extension, the g-comtrace notion, is introduced. Comtraces and g-comtraces are generalizations of Mazurkiewicz traces, where the  concepts of simultaneity, serializability and interleaving are used to define the quotient monoids instead of the usual independency relation in the case of traces. We analyzed some  algebraic properties of comtraces and g-comtraces, where an interesting application  is the proof of the uniqueness of comtrace canonical representation. We study the canonical representations of traces, comtraces and g-comtraces and their mutual relationships in a more unified framework. We observe that comtraces have a natural unique canonical form which corresponds to  their maximal concurrent representation\footnote{This is also true for traces when they are represented as {\em vector firing sequences} \cite{DJKL}}, while the unique canonical representation of g-comtrace can only be obtained by choosing the lexicographically least element. 

The most important contribution of this paper, \theoref{t7}, shows that every g-comtrace uniquely determines a labeled gso-structure. We believe the reason why the proof of \theoref{t7} is   more technical than the similar theorem for comtraces  is that both comtraces and so-structures satisfy paradigm $\pi_3$ while g-comtraces and gso-structures do not. Intuitively, what paradigm $\pi_3$ really says is that the underlying structure consists of partial orders. For comtraces and so-structures, we did augment some more priority relationships into the incomparable elements with respect to the standard causal partial order to produce the not later than relation, and this process might introduce cycles into the graph of the ``not later than'' relation. However, it is important to observe that any two distinct elements lying on a cycle of the ``not later than'' relation must belong to the same synchronous set since the ``not later than'' relation is a \emph{strict preorder}. Thus, if  we  collapse each synchronous set into a single vertex, then the resulting ``quotient'' graph of the ``not later than'' relation is a partial order. The reader is referred to the second author's recent work \cite{Le10} for more detailed discussion on the preorder property of the ``not later than'' relation and how this property manifests itself in the comtrace notion. When paradigm $\pi_3$ is not satisfied, as with g-comtraces or gso-structures, we have more than a partial order structure, and hence the usual techniques that depend too   on the underlying partial order structure of comtraces and so-structures are often not applicable.

Despite some obvious advantages, for instance,  handy composition and no need to use labels, quotient monoids
(perhaps with some exception of  traces) are   less popular for analyzing issues of concurrency than their relational counterparts such as partial orders, so-structures, occurrence graphs, etc. We believe that in many cases, more sophisticated quotient monoids, e.g., comtraces and g-comtraces, can provide simpler and more adequate models of concurrent histories than their relational equivalences.

Much harder future tasks are in the area of comtrace and g-comtrace languages where major problems like recognizability \cite{Och}, acceptability \cite{Zie}, etc. are still open.

\section*{Acknowledgments} We are grateful to the anonymous referees, who pointed out many typos and suggested us a   better way to organize the materials of this paper. We are indebted to many others for their feedbacks and encouragements. These include Steve Cook, Grzegorz Herman, Marek Janicki, Galina Jir\'askov\'a, Gregory H. Moore, Michael Soltys, Yuli Ye, Marek Zaionc, Nadya Zubkova, and Jeffery Zucker. This work is financially  supported by the Ontario Graduate Scholarship and the Natural Sciences and Engineering Research Council of Canada.

\appendix

\section{Proof of \lemref{gposinv}}

\begin{proposition}
Let $u$ be a step sequence over a g-comtrace alphabet $(E,sim,ser,inl)$ and $\alpha,\beta\in\Sigma_u$ such that $l(\alpha)=l(\beta)$ and $\alpha\neq \beta$. Then
\begin{enumerate}
 \item $pos_u(\alpha)\not= pos_u(\beta)$
 \item If $pos_u(\alpha)<pos_u(\beta)$ and $v$ is a step sequence satisfying $v\eqb u$, then $pos_v(\alpha)<pos_v(\beta)$.
\end{enumerate}
\label{prop:samelab}
\end{proposition}
\begin{proof}
\begin{enumerate}
 \item Follows from the fact that $sim$ is irreflexive.
 \item Follows from \propref{gcomequi} and that $ser$ and $inl$ are irreflexive.\END
\end{enumerate}
\end{proof}

From the definition of g-comtrace $\eqa_{\{ser,inl\}}$ (Definition \ref{def:eqa}), we can easily show the following proposition, which aims to describe the intuition that if an event $\alpha$ occurs before (or simultaneously with) $\beta$ in the first step sequence and $\alpha$ occurs later than $\beta$ on the second step sequence congruent with the first one, then there must be two ``immediately congruent'' step sequences, i.e., related by the relation $\eqa_{\{ser,inl\}}$ (writtten as just $\approx$), where this commutation (or serialization) of $\alpha$ and $\beta$ occurs.    
\begin{proposition}
Let $u,w$ be step sequences over  a g-comtrace alphabet $(E,sim,ser,inl)$ such that $u(\eqa\cup\eqa^{-1})w$. Then
\begin{enumerate}
\item If $pos_u(\alpha)<pos_u(\beta)$ and $pos_w(\beta)<pos_w(\alpha)$, then there are $x,y,A,B$ such that $\h{u}=\h{x}\h{A}\;\h{B}\h{y}(\eqa\cup\eqa^{-1})\h{x}\h{B}\;\h{A}\h{y}=\h{w}$ and $\alpha\in \h{A},\beta\in\h{B}$. We also have $(l(\alpha),l(\beta))\in inl$.
\item If $pos_u(\alpha)=pos_u(\beta)$ and $pos_w(\beta)<pos_w(\alpha)$, then there are $x,y,A,B,C$ such that $\h{u}=\h{x}\h{A}\h{y}\eqa\h{x}\h{B}\;\h{C}\h{y}=\h{w}$ and $\beta\in \h{B}$ and $\alpha\in\h{C}$. This also means $(l(\beta),l(\alpha))\in ser$.
\qed
\end{enumerate}
\label{prop:eqapos1}
\end{proposition}

\begin{proposition}
Let $s$ be a step sequence over a g-comtrace alphabet $(E,sim,ser,inl)$. If $\alpha,\beta\in\Sigma_s$, then
\begin{enumerate}
\item $\alpha\com_{s}\beta \implies  \forall u\in [s].\ pos_{u}(\alpha)\not= pos_{u}(\beta)$,
\item $\alpha\sq_{s}\beta \implies\forall u\in [s].\ pos_{u}(\alpha)\le pos_{u}(\beta)$,
\item $\alpha\PO_{s}\beta\implies\forall u\in [s].\ pos_{u}(\alpha)< pos_{u}(\beta)$.
\end{enumerate}
and $\alpha\neq \beta$ in all three cases.
\label{prop:gsmicro}
\end{proposition}
\begin{proof}1. Follows from the fact that $inl\cap sim=\emptyset$.

2. Assume that $\alpha\sq_{s}\beta$. Suppose that $\exists u\in [s].\ pos_{u}(\alpha)> pos_{u}(\beta)$. Then there must be some $u_1,u_1\in [s]$ such that $u_1(\eqa\cup\eqa^{-1}) u_2$ and $pos_{u_1}(\alpha)\le pos_{u_1}(\beta)$ and $pos_{u_2}(\alpha)> pos_{u_2}(\beta)$. There are two cases:
\begin{enumerate}[(i)]
\item If $pos_{u_1}(\alpha)< pos_{u_1}(\beta)$ and $pos_{u_2}(\alpha)> pos_{u_2}(\beta)$, then by \propref{eqapos1}(1), $(l(\alpha),l(\beta))\in inl$, contradicting that $\alpha\sq_{s}\beta$.

\item If $pos_{u_1}(\alpha)= pos_{u_1}(\beta)$ and $pos_{u_2}(\alpha)> pos_{u_2}(\beta)$, then it follows from \propref{eqapos1}(2),  $(l(\beta),l(\alpha))\in ser$, contradicting that $\alpha\sq_{s}\beta$.
\end{enumerate}

3. Assume that $\alpha\PO_{s}\beta$. Suppose that $\exists u\in [s].\ pos_{u}(\alpha)\ge pos_{u}(\beta)$. Then there must be some $u_1,u_1\in [s]$ such that $u_1(\eqa\cup\eqa^{-1}) u_2$ and $pos_{u_1}(\alpha)< pos_{u_1}(\beta)$ and $pos_{u_2}(\alpha)\ge pos_{u_2}(\beta)$. There are two cases:
\begin{enumerate}[(i)]
\item
\begin{sloppypar}
If $pos_{u_1}(\alpha)< pos_{u_1}(\beta)$ and $pos_{u_2}(\alpha)= pos_{u_2}(\beta)$, then it follows from \propref{eqapos1}(2) that  $(l(\alpha),l(\beta))\in ser$ and $\neg(\alpha\com_s\beta)$. Hence, it follows from \eref{ss2gsos.3} that
\[\exists\delta,\gamma\in\Sigma_s.
		\left(\begin{array}{ll}
		&pos_s(\delta)<pos_s(\gamma)\wedge(l(\delta),l(\gamma))\notin ser\\
		\wedge&\alpha\,\sq_s^*\,\delta\,\sq_s^*\,\beta\wedge\alpha\,\sq_s^*\,\gamma\,\sq_s^*\,\beta
		\end{array}\right).\]
By (2) and transitivity of $\le$, we have
		\[\left(\begin{array}{ll}
		&\gamma\not=\delta\;\wedge\;(l(\delta),l(\gamma))\notin ser\\
		\wedge&(\forall u\in [s].\;pos_u(\alpha)\le pos_u(\delta) \le pos_u(\beta))\\
		\wedge&(\forall u\in [s].\;pos_u(\alpha)\le pos_u(\gamma) \le pos_u(\beta)
		\end{array}\right).\]
But since $pos_{u_2}(\alpha)= pos_{u_2}(\beta)$, we get  $pos_{u_2}(\gamma)=pos_{u_2}(\delta)$. Since we assumed $pos_s(\delta)<pos_s(\gamma)$, it follows from \propref{eqapos1}(2) that  $(l(\delta),l(\gamma))\in ser$, a contradiction.
\end{sloppypar}

\item If $pos_{u_1}(\alpha)< pos_{u_1}(\beta)$ and $pos_{u_2}(\alpha)> pos_{u_2}(\beta)$, then by \propref{eqapos1}(1), $(l(\alpha),l(\beta))\in inl$. Since we already assumed $\alpha\PO_{s}\beta$, by \eref{ss2gsos.3}, $(\alpha,\beta)\in \com_s\cap\left(\si{\sq_s^*}\circ\reco{\com}_s\circ \si{\sq_s^*}\right)$. So there are $\gamma,\delta$ such that  $\alpha\; \si{\sq_s^*}\;\gamma\;\reco{\com}_s\;\delta\;\si{\sq_s^*}\;\beta$. Observe that
\begin{align*}
  	&\alpha\; \si{\sq_s^*}\;\gamma&\\
\implies&\alpha\;(\sq_s^*)\;\gamma\;\wedge\; \gamma\;(\sq_s^*)\;\alpha&\\
\implies&\forall u\in [s].\ pos_{u}(\alpha)\le pos_{u}(\gamma)\;\wedge\;\forall u\in [s].\ pos_{u}(\gamma)\le pos_{u}(\alpha)&\ttcomment{by (2)}\\
\implies&\forall u\in [s].\ pos_{u}(\alpha)= pos_{u}(\gamma)&\\
\implies&\set{\alpha,\gamma}\subseteq \h{A} & \ttcomment{since $\alpha\in \h{A}$}
\end{align*}

Similarly, since $\delta\; \si{\sq_s^*}\;\beta$, we can show that $\set{\delta,\beta}\subseteq \h{B}$. Since $\h{x}\h{A}\;\h{B}\h{y}\left(\eqa\cup\eqa^{-1}\right)\h{x}\h{B}\;\h{A}\h{y}$, we get $A\times B\subseteq inl$. So $(l(\gamma),l(\delta))\in inl$. But $\gamma\;\reco{\com}_s\;\delta$ implies that $(l(\gamma),l(\delta))\notin inl$, a contradiction.\END
\end{enumerate}
\end{proof}

Immediately from \propref{gsmicro}, we get the following proposition.
\begin{proposition}
Let $s$ be a step sequence over a g-comtrace alphabet $(E,sim,ser,inl)$ and   $G^{\{s\}}=(\Sigma_s,\com,\sq)$. If $\alpha,\beta\in\Sigma_s$, then
\begin{enumerate}
\item $\alpha\com\beta \implies  \forall u\in [s].\ pos_{u}(\alpha)\not= pos_{u}(\beta)$
\item $\alpha\sq\beta \implies\left(\alpha\not=\beta\;\wedge\;\forall u\in [s].\ pos_{u}(\alpha)\le pos_{u}(\beta)\right)$ \qed
\end{enumerate}
\label{prop:gposinv2}
\end{proposition}

\begin{definition}[serializable and non-serializable steps]
Let $A$ be a step over a g-comtrace alphabet $(E,sim,ser,inl)$ and let $a\in A$ then:
\begin{enumerate}
\item Step $A$ is called \emph{serializable} iff
$$\exists B,C\in \PSB{A}.\;B\cup C = A \;\wedge\; B\times C \subseteq ser.$$
Step $A$ is called \emph{non-serializable} iff $A$ is not serializable.
Every non-serializable step is a synchronous step as defined in \defref{synsteps}.
\item Step $A$ is called \emph{serializable to the left of $a$} iff
$$\exists B,C\in \PSB{A}.\;B\cup C = A\;\wedge\; a\in B\;\wedge\; B\times C \subseteq ser.$$
Step $A$ is called \emph{non-serializable to the left of $a$} iff $A$ is not serializable to the left of $a$, i.e.,
$\forall B,C\in \PSB{A}.\left(B\cup C = A\;\wedge\; a\in B\right)\implies B\times C \not\subseteq ser.$
\item Step $A$ is called \emph{serializable to the right of $a$} iff
$$\exists B,C\in \PSB{A}.\;B\cup C = A\;\wedge\; a\in C\;\wedge\; B\times C \subseteq ser.$$
Step $A$ is called \emph{non-serializable to the right of $a$} iff $A$ is not serializable to the right of $a$, i.e., $\forall B,C\in \PSB{A}.\left(B\cup C = A\;\wedge\; a\in C\right)\implies B\times C \not\subseteq ser.$ \EOD
\end{enumerate}
\end{definition}

\begin{proposition}
Let $A$ be a step over a g-comtrace alphabet $(E,sim,ser,inl)$. Then
\begin{enumerate}
\item If $A$ is non-serializable to the left of $l(\alpha)$ for some $\alpha\in\h{A}$, then $\alpha\sqsubset_A^*\beta$ for all $\beta\in \h{A}$.
\item If $A$ is non-serializable to the right of $l(\beta)$ for some $\beta\in\h{A}$, then $\alpha\sqsubset_A^*\beta$ for all $\alpha\in \h{A}$.
\item If $A$ is non-serializable, then $\forall\alpha,\beta\in \h{A}.\ \alpha\sqsubset_A^*\beta$.
\end{enumerate}
\label{prop:gnonser}
\end{proposition}

Before we proceed with the proof, since for all $\alpha,\beta\in\h{A}$, $(l(\alpha),l(\beta))\notin inl$, observe that
\begin{align*}
\alpha\sq_A\beta \iff  pos_A(\alpha)\le pos_A(\beta) \wedge (l(\beta),l(\alpha))\notin ser.
\end{align*}
\begin{proof}
1. For any $\beta\in\h{A}$, we have to show that $\alpha\sqsubset_A^*\beta$. We define the $\sqsubset_A$-\textit{right closure set} of $\alpha$ inductively as follows:
\begin{align*}
RC^0(\alpha)&\df\set{\alpha}&
RC^n(\alpha)&\df\bset{\delta\in \h{A}\mid\exists\gamma\in RC^{n-1}(\alpha)\;\wedge\;\gamma\sqsubset_A\delta}
\end{align*}

Then by induction on $n$, we can show that $|RC^{n+1}(\alpha)|> |RC^n(\alpha)|$ or $RC^n(\alpha)=\h{A}$. So if $A$ is finite, then for some $n<|A|$, we must have $RC^n(\alpha)=\h{A}$ and $\beta\in RC^n(\alpha)$. It follows that $\alpha\sqsubset_A^*\beta$.

2. Dually to (1).

3. Since $A$ is non-serializable, it follows that $A$ is non-serializable to the left of $l(\alpha)$ for \emph{every} $\alpha\in\h{A}$. Hence, the assertion follows.\END
\end{proof}

The existence of a non-serializable sub-step of a step $A$ to the left/right of an element $a\in A$ can be explained by the following proposition.

\begin{proposition}
Let $A$ be a step over an alphabet $\Theta=(E,sim,ser,inl)$ and $a \in A$. Then
\begin{enumerate}
 \item There exists a unique $B\subseteq A$ such that $a\in B$, $B$ is non-serializable to the left of $a$, and  $A\not=B \implies A\eqb (A\setminus B)B.$
 \item There exists a unique $C\subseteq A$ such that $a\in C$, $C$ is non-serializable to the right of $a$, and $A\not=C \implies A\eqb C(A\setminus C).$
 \item There exists a unique $D\subseteq A$ such that $a\in D$, $D$ is non-serializable, and $A\eqb xDy$, where $x$ and $y$ are step sequences over $\Theta$.
\end{enumerate}
\label{prop:gnonsercon}
\end{proposition}
\begin{proof}1. If $A$ is non-serializable to the left of $a$, then $B=A$. If $A$ is serializable to the left of $a$, then the following set is not empty:
$$\zeta \df \bset{D\in\PSB{A}\mid \exists C\in\PSB{A}.\left(C\cup D = A\;\wedge\;a\in D\;\wedge\;C\times D \subseteq ser\right)}$$

Let $B\in\zeta$ such that $B$ is a minimal element of the poset $(\zeta,\subset)$. 
Let $B\in\zeta$ such that $B$ is a minimal element of the poset $(\zeta,\subset)$. We claim that $B$ is non-serialisable to the left of $a$. Suppose for a contradiction that $B$ is serialisable to the left of $a$, then there are some sets $E,F\in \PSB(B)$ such that $E\cup F = B\;\wedge\;a\in F\;\wedge\;E\times F \subseteq ser.$ Since $B\in\chi$, there is some set $G\in\PSB(A)$ such that $G\cup B=A\;\wedge\;a\in B\;\wedge\;G\times B\subseteq ser$. Because  $G\times B\subseteq ser$ and $F \subset B$, it follows that $G\times F\subseteq ser$. But since $E\times F \subseteq ser$, we have $(G\cup E)\times F \subseteq ser$. Hence, $(G\cup E)\cup F = A\;\wedge\;a\in F\;\wedge\;(G\cup E)\times F \subseteq ser$. So $E\in\zeta$ and $E\subset B$. This contradicts that $B$ is minimal. Hence, $B$ is non-serialisable to the left of $a$.

By the way the set $\zeta$ is defined, $A\eqb(A\setminus B)B$. 
It remains to prove the uniqueness of $B$. Let $B'\in\zeta$ such that $B'$ is a minimal element of the poset $(\zeta,\subset)$. We want to show that $B=B'$.

We first show that $B\subseteq B'$. Suppose that there is some $b\in B$ such that  $b\not= a$ and $b\notin B'$. Let $\alpha$ and $\beta$ denote the event occurrences $a^{(1)}$ and $b^{(1)}$ in $\Sigma_{A}$ respectively. Since $a\in B$ and $b$ is non-serializable to the left of $a$ and $a\not=b$, it follows from \propref{gnonser}(1) that  $\alpha\sq_{[A]}\beta$. Hence, by \propref{gsmicro}(2), we have
\begin{equation}
\forall u\in [A].\ pos_{u}(\alpha)\le pos_{u}(\beta) \label{eq:gnonsercon.1}
\end{equation}
By the way $B'$ is chosen, we know $A\eqb (A\setminus B')B'$ and $b\notin B'$. So it follows that $b\in (A\setminus B')$. Hence, we have $(A\setminus B')B'\in[A]$ and $pos_{(A\setminus B')B'}(\beta)< pos_{(A\setminus B')B'}(\alpha)$, which contradicts \eref{gnonsercon.1}. Thus, $B\subseteq B'$.
By reversing the roles of $B$ and $B'$, we can prove that $B\supseteq B'$. Hence, $B=B'$.

2. Dually to (1).

3. By (1) and (2), we choose $D$ to be non-serializable to the left and to the right of $a$.
\END
\end{proof}

\begin{lemma}
Let $s$ be a step sequence over a g-comtrace alphabet $(E,sim,ser,inl)$ and   $G^{\{s\}}=(\Sigma_s,\com,\sq)$. Let $\PO=\sq\cup\com$. If $\alpha,\beta\in\Sigma_s$, then
\begin{enumerate}
\item $\left(\begin{array}{ll}
		&(\forall u\in [s].\ pos_{u}(\alpha)\not= pos_{u}(\beta))\\
		\wedge&(\exists u\in [s].\ pos_{u}(\alpha)< pos_{u}(\beta))\\
		\wedge&(\exists u\in [s].\ pos_{u}(\alpha)> pos_{u}(\beta))
		\end{array}\right)\implies  \alpha\com\beta$
\item $(\forall u\in [s].\ pos_{u}(\alpha) < pos_{u}(\beta))\implies  \alpha\PO\beta$
\item $(\alpha\not=\beta \;\wedge\; \forall u\in [s].\ pos_{u}(\alpha)\le pos_{u}(\beta))\implies \alpha\sq\beta $
\end{enumerate}
\label{lem:pos2gs}
\end{lemma}
\begin{proof}1. Assume the left-hand side of the implication
Then by \propref{eqapos1}(1),  $(l(\alpha),l(\beta))\in inl$, which by \eref{ss2gsos.1} implies that $\alpha\com_s\beta$. By Definitions \ref{def:ccl} and \ref{def:ss2gsos}, it follows that $\alpha\com\beta$.\\

2, 3. Since statements (2) and (3) are mutually related due to the fact that $\PO\;\subseteq\; \sq$, we cannot prove each statement seperately. The main technical insight is that, to have a stronger induction hypothesis, we need prove both  statements simultaneously.

Assume $\forall u\in [s].\ pos_{u}(\alpha)\le pos_{u}(\beta)$ and $\alpha\not=\beta$. Hence, we can choose $u_0\in [s]$ where $\h{u_0}=\h{x_0}\;\h{E_1}\ldots\h{E_k}\;\h{y_0}$ ($k\ge 1$), $E_1,E_k$ are non-serializable, $\alpha\in\h{E_1}$, $\beta\in\h{E_k}$, and
\begin{align}
\forall u'_0\in [s].\left(\begin{array}{lc}
		&\left(\h{u'_0}=\h{x'_0}\;\h{E'_1}\ldots\h{E'_{k'}}\;\h{y'_0}\;\wedge\; \alpha\in\h{E'_1}\;\wedge\; \beta\in\h{E'_{k'}}\right)\\
		\implies&\wei(\h{E_1}\ldots\h{E_k})\le \wei(\h{E'_1}\ldots\h{E'_{k'}})
		\end{array}\right)
\label{eq:pos2gs.0}
\end{align}
We will prove by induction on $\wei(\h{E_1}\ldots\h{E_k})$ that
\begin{align}
(\forall u\in [s].\ pos_{u}(\alpha) < pos_{u}(\beta))\implies  \alpha\PO\beta \label{eq:pos2gs.1}\\
(\alpha\not=\beta \;\wedge\; \forall u\in [s].\ pos_{u}(\alpha)\le pos_{u}(\beta))\implies \alpha\sq\beta \label{eq:pos2gs.2}
\end{align}

\paragraph{Base case} When $\wei(\h{E_1}\ldots\h{E_k})=2$, then we consider two cases:
\begin{itemize}
\item If $\alpha\not=\beta$, $\forall u\in [s].\ pos_{u}(\alpha) \le pos_{u}(\beta)$ and $\exists u\in [s].\ pos_{u}(\alpha) = pos_{u}(\beta)$, then 
	\begin{itemize}
	 \item $\h{u_0}=\h{x_0}\set{\alpha,\beta}\h{y_0}$, or
	 \item $\h{u_0}=\h{x_0}\set{\alpha}\set{\beta}\h{y_0}\eqb \h{x_0}\set{\alpha,\beta}\h{y_0}$
	\end{itemize}
But since $\forall u\in [s].\ pos_{u}(\alpha) \le pos_{u}(\beta)$, in either case, we must have $\set{l(\alpha),l(\beta)}$ is not serializable to the right of $l(\beta)$. Hence, by \propref{gnonser}(2), $\alpha\;(\sq_s)^*\beta$. This by Definitions \ref{def:ccl} and \ref{def:ss2gsos} implies that $\alpha\sq\beta$.

\item If $\forall u\in [s].\ pos_{u}(\alpha) < pos_{u}(\beta)$, then it follows $\h{u_0}=\h{x_0}\set{\alpha}\set{\beta}\h{y_0}$ and $(l(\alpha),l(\beta))\notin ser\cup inl$. This, by \eref{ss2gsos.3}, implies that $\alpha\PO_s\beta$. Hence, from Definitions \ref{def:ccl} and \ref{def:ss2gsos}, we get $\alpha\PO\beta$.
\end{itemize}
Since $\PO\;\subseteq\;\sq$, it follows from these two cases that \eref{pos2gs.1} and \eref{pos2gs.2} hold.

\paragraph{Inductive step} When $\wei(\h{E_1}\ldots\h{E_k})>2$, then $\h{u_0}=\h{x_0}\;\h{E_1}\ldots\h{E_k}\;\h{y_0}$ where $k\ge 1$. We need to consider two cases:

\emph{Case (i):} If $\alpha\not=\beta$ and $\forall u\in [s].\ pos_{u}(\alpha) \le pos_{u}(\beta)$ and $\exists u\in [s].\ pos_{u}(\alpha) = pos_{u}(\beta)$, then there is some $v_0$  $\h{v_0}=\h{w_0}\;\h{E}\;\h{z_0}$ and $\alpha,\beta\in \h{E}$. Either $E$ is non-serializable to the right of $l(\beta)$, or by \propref{gnonsercon}(2) $\h{v_0}=\h{w_0}\;\h{E}\;\h{z_0}\eqb \h{w'_0}\;\h{E'}\;\h{z'_0}$ where $E'$ is non-serializable to the right of $l(\beta)$. In either case, by \propref{gnonser}(2), we have $\alpha\sq_s^*\beta$. So by Definitions \ref{def:ccl} and \ref{def:ss2gsos},  $\alpha\sq\beta$.

\emph{Case (ii):}  If $\forall u\in [s].\ pos_{u}(\alpha) < pos_{u}(\beta)$, then it follows $\h{u_0}=\h{x_0}\;\h{E_1}\ldots\h{E_k}\;\h{y_0}$ where $k\ge 2$ and $\alpha\in\h{E_1},\beta\in\h{E_k}$. If $(l(\alpha),l(\beta))\notin ser\cup inl$, then by \eref{ss2gsos.3}, $\alpha\PO_s\beta$. Hence, from Definitions \ref{def:ccl} and \ref{def:ss2gsos}, we get  $\alpha\PO\beta$.  So we need to consider only when $(l(\alpha),l(\beta))\in ser$ or $(l(\alpha),l(\beta))\in inl$. There are three cases to consider:
\begin{enumerate}[(a)]
\item If $\h{u_0}=\h{x_0}\;\h{E_1}\;\h{E_2}\;\h{y_0}$ where $E_1$ and $E_2$ are non-serializable, then since we assume $\forall u\in [s].\ pos_{u}(\alpha) < pos_{u}(\beta)$, it follows that $E_1\times E_2\not\subseteq ser$ and $E_1\times E_2\not\subseteq inl$. Hence, there are $\alpha_1,\alpha_2\in \h{E_1}$ and $\beta_1,\beta_2\in \h{E_2}$ such that $(l(\alpha_1),l(\beta_1))\notin inl$ and $(l(\alpha_2),l(\beta_2))\notin ser$. Since $E_1$ and $E_2$ are non-serializable, by \propref{gnonser}(3), $\alpha_1\sq_s^*\alpha_2$ and $\beta_2\sq_s^*\beta_1$. Also by \defref{ss2gsos}, we know that $\alpha_1\com_s\beta_2$ and $\alpha_2\reco{\com}_s\beta_1$. Thus, by  \defref{ss2gsos}, we have  $\alpha_1\PO_s\beta_2$. Since $E_1$ and $E_2$ are non-serializable, by \propref{gnonser}(3), $\alpha\sq_s^*\alpha_1\PO_s\beta_2\sq_s^*\beta$. Hence, by Definitions \ref{def:ccl} and \ref{def:ss2gsos}, $\alpha\PO\beta$.

\item If $\h{u_0}=\h{x_0}\;\h{E_1}\ldots\h{E_k}\;\h{y_0}$ where $k\ge 3$ and $(l(\alpha),l(\beta))\in inl$, then let $\gamma\in \h{E_2}$. Observe that we must have  \[\h{u_0}=\h{x_0}\;\h{E_1}\ldots\h{E_k}\;\h{y_0}\eqb \h{x_1}\;\h{E_1}\,\h{w_1}\,\h{F}\,\h{z_1}\,\h{E_k}\;\h{y_1}\eqb \h{x_2}\;\h{E_1}\,\h{w_2}\,\h{F}\,\h{z_2}\,\h{E_k}\;\h{y_2}\]
such that $\gamma\in \h{F}$, $F$ is a non-serializable, and $\wei(\h{E_1}\,\h{w_1}\,\h{F}),\wei(\h{F}\,\h{z_2}\,\h{E_k})$ satisfy the minimal condition similarly to \eref{pos2gs.0}. Since from the way $u_0$ is chosen, we know that $\forall u\in [s].\ pos_{u}(\alpha) \le pos_{u}(\gamma)$ and $\forall u\in [s].\ pos_{u}(\gamma) \le pos_{u}(\beta)$, by applying the induction hypothesis, we get
\begin{align}
\alpha\sq \gamma\sq \beta \label{eq:pos2gs.3}
\end{align}
So by transitivity of $\sq$, we get $\alpha\sq\beta$. But since  we assume $(l(\alpha),l(\beta))\in inl$, it follows that $\alpha\com\beta$. Hence, $(\alpha,\beta)\in\; \sq\cap\com \;=\; \PO$.

\item If $\h{u_0}=\h{x_0}\;\h{E_1}\ldots\h{E_k}\;\h{y_0}$ where $k\ge 3$ and $(l(\alpha),l(\beta))\in ser$, then we observe from how $u_0$ is chosen that
\begin{align*}
\forall \gamma\in\Al(\h{E_1}\ldots\h{E_k}).\left(\forall u\in [s].\ pos_{u_0}(\alpha) \le pos_{u_0}(\gamma)\le pos_{u_0}(\beta)\right)
\end{align*}
Similarly to how we show \eref{pos2gs.3}, we can prove that
\begin{align}
\forall \gamma\in\Al(\h{E_1}\ldots\h{E_k})\setminus\set{\alpha,\beta}.\;\alpha\sq\gamma\sq\beta \label{eq:pos2gs.4}
\end{align}

We next want to show that
\begin{align}
\exists \delta,\gamma\in\Al(\h{E_1}\ldots\h{E_k}).\bigl(pos_{u_0}(\delta)<pos_{u_0}(\gamma) \wedge (l(\delta),l(\gamma))\notin ser\bigr) \label{eq:pos2gs.5}
\end{align}
Suppose that \eref{pos2gs.5} does not hold, then
\begin{align*}
\forall \delta,\gamma\in\Al(\h{E_1}\ldots\h{E_k}).\bigl(pos_{u_0}(\delta)<pos_{u_0}(\gamma) \implies (l(\delta),l(\gamma))\in ser\bigr)
\end{align*}
It follows that $\h{u_0}=\h{x_0}\;\h{E_1}\ldots\h{E_k}\;\h{y_0}\eqb\h{x_0}\;\h{E}\;\h{y_0}$, which contradicts that $\forall u\in [s].\ pos_{u}(\alpha) < pos_{u}(\beta).$
Hence, we have shown \eref{pos2gs.5}.

Let $\delta,\gamma\in\Al(\h{E_1}\ldots\h{E_k})$ be event occurrences such that  $pos_{u_0}(\delta)<pos_{u_0}(\gamma)$ and $(l(\delta),l(\gamma))\notin ser$. By \eref{pos2gs.4},
$\alpha(\sq\cup\; id_{\Sigma_s})\delta(\sq\cup\; id_{\Sigma_s})\beta$ and $\alpha(\sq\cup\; id_{\Sigma_s})\gamma(\sq\cup \;id_{\Sigma_s})\beta$. If $\alpha\PO\delta$ or $\delta\PO\beta$ or $\alpha\PO\gamma$ or $\gamma\PO\beta$, then by (S4) of \defref{sos}, $\alpha\PO\beta$. Otherwise, by Definitions \ref{def:ccl} and \ref{def:ss2gsos}, we have $\alpha\sq_s^*\delta\sq_s^*\beta$ and $\alpha\sq_s^*\gamma\sq_s^*\beta$. But since  $pos_{u_0}(\delta)<pos_{u_0}(\gamma)$ and $(l(\delta),l(\gamma))\notin ser$, by \defref{ss2gsos}, $\alpha\PO_s\beta$. So by  Definitions \ref{def:ccl} and \ref{def:ss2gsos}, we have $\alpha\PO\beta$.
\end{enumerate}
Thus, we have shown \eref{pos2gs.1} and \eref{pos2gs.2} as desired.\END
\end{proof}

\noindent{\bf \lemref{gposinv}.}
 {\em Let $s$ be a step sequence over a g-comtrace alphabet $(E,sim,ser,inl)$. Let $G^{\{s\}}=(\Sigma_s,\com,\sq)$, and let $\PO=\com\cap\sq$. Then for every $\alpha,\beta\in\Sigma_s$, we have
\begin{enumerate}
\item $\alpha\com\beta \iff  \forall u\in [s].\ pos_{u}(\alpha)\not= pos_{u}(\beta)$
\item $\alpha\sq\beta \iff\alpha\not=\beta \wedge \forall u\in [s].\ pos_{u}(\alpha)\le pos_{u}(\beta) $
\item $\alpha\PO\beta \iff \forall u\in [s].\ pos_{u}(\alpha)< pos_{u}(\beta)$
\item If $l(\alpha)=l(\beta)$ and $pos_s(\alpha)<pos_s(\beta)$, then $\alpha\PO\beta$.
\end{enumerate}}
\begin{proof}
\begin{enumerate}
\item Follows from \propref{gposinv2}(1) and \lemref{pos2gs}(1, 2).
\item Follows from \propref{gposinv2}(2) and \lemref{pos2gs}(3).
\item Follows from (1) and (2).
\item Follows from \propref{samelab}(2). \END
\end{enumerate}
\end{proof}

\section{Proof of Lemma \ref{lem:ext2xi}.}
\noindent{\bf Lemma \ref{lem:ext2xi}.}
\emph{Let $s,t \in \E^*$
and $\lhd_s\in ext(G^{\{t\}})$. Then $G^{\{s\}}=G^{\{t\}}$.}

\begin{proof} To show $G^{\{s\}}=G^{\{t\}}$, it suffices to show that $\com_t\;=\;\com_s$, $\PO_t\;=\;\PO_s$ and $\sq_t\;=\;\sq_s$ since this will imply  that 
\[G^{\{t\}}=(\Sigma,\com_t\cup\PO_t,\sq_t\cup\PO_t)^{\ccl}= (\Sigma,\com_s\cup\PO_s,\sq_s\cup\PO_s)^{\ccl}=G^{\{s\}}.\]

($\com_t\;=\;\com_s$) Trivially follows from \defref{ss2gsos}.\\

($\sq_t\;=\;\sq_s$) If $\alpha\sq_t\beta$, then by Definitions \ref{def:ccl} and \ref{def:ss2gsos}, $\alpha\sq\beta$. But since $\lhd_s\in ext(G^{\{t\}})$, we have $\alpha\lhd_s^\frown\beta$, which implies $pos_s(\alpha)\le pos_s(\beta)$.
But since $\alpha\sq_t\beta$, it follows  by \defref{ss2gsos} that  $(l(\beta),l(\alpha))\notin ser\cup inl$. Hence, by \defref{ss2gsos}, $\alpha\sq_s\beta$. Thus,
\begin{align}
\sq_t\;\subseteq\;\sq_s \label{eq:ext2xi.sg2}
\end{align}

It remains to show that $\sq_s\;\subseteq\;\sq_t$. Let $\alpha\sq_s\beta$, and we suppose that $\neg(\alpha\sq_t\beta)$. Since $\alpha\sq_s\beta$, by \defref{ss2gsos}, $pos_s(\alpha)\le pos_s(\beta)$ and $(l(\beta),l(\alpha))\notin ser\cup inl$. Since we assume $\neg(\alpha\sq_t\beta)$, by \defref{ss2gsos}, we must have $pos_t(\beta)<pos_t(\alpha)$. Hence, by Definitions \ref{def:ccl} and \ref{def:ss2gsos},  $\beta\PO_t\alpha$ and  $\beta\PO\alpha$. But since $\lhd_s\in ext(G^{\{t\}})$, we have $\beta\lhd_s\alpha$. So $pos_s(\beta)<pos_s(\alpha)$, a contradiction. Thus, $\sq_s\subseteq\sq_t$. Together with \eref{ext2xi.sg2}, we get $\sq_t\;=\;\sq_s$\\

($\PO_t\;=\;\PO_s$) If $\alpha\PO_t\beta$, then by Definitions \ref{def:ccl} and \ref{def:ss2gsos}, $\alpha\PO\beta$ (of $G^{\{t\}}$). But since $\lhd_s\in ext(G^{\{t\}})$, we have $\alpha\lhd_s\beta$, which implies
\begin{align}
pos_s(\alpha)< pos_s(\beta) \label{eq:ext2xi.2}
\end{align}
Since $\alpha\PO_t\beta$, by \defref{ss2gsos}, we have
\begin{align*}
\begin{array}{ll}
	&(l(\alpha),l(\beta))\notin ser\cup inl \\
	\vee& (\alpha,\beta)\in\; \com_t\cap\left(\si{\sq_t^*}\circ\reco{\com}_t\circ \si{\sq_t^*}\right)\\
	\vee&\left(\begin{array}{ll}
			& (l(\alpha),l(\beta))\in ser\\
		\wedge	&\exists\delta,\gamma\in\Sigma_t.
		\left(\begin{array}{ll}
		&pos_t(\delta)<pos_t(\gamma) \wedge(l(\delta),l(\gamma))\notin ser\\
		\wedge&\alpha\,\sq_t^*\,\delta\,\sq_t^*\,\beta\wedge\alpha\,\sq_t^*\,\gamma\,\sq_t^*\,\beta
		\end{array}\right)\end{array}\right).
	\end{array}
\end{align*}
We want to show that $\alpha\PO_s\beta$. There are three cases to consider:
	\begin{enumerate}[(a)]
	\item When $(l(\alpha),l(\beta))\notin ser\cup inl$, it follows from \eref{ext2xi.2} and \defref{ss2gsos} that $\alpha\PO_s\beta$.
	\item When $(\alpha,\beta)\in\; \com_t\cap\left(\si{\sq_t^*}\circ\reco{\com}_t\circ \si{\sq_t^*}\right)$, then $\alpha\com_t\beta$ and there are $\delta,\gamma \in \Sigma$ such that  $\alpha\; \si{\sq_t^*}\;\delta\;\reco{\com}_t\;\gamma\; \si{\sq_t^*}\;\beta$. Since $\sq_t\;=\;\sq_s$ and $\com_t\;=\;\com_s$, we have $\alpha\com_s\beta$ and $\alpha\; \si{\sq_s^*}\;\delta\;\reco{\com}_s\;\gamma\; \si{\sq_s^*}\;\beta$. Thus, it follows from \eref{ext2xi.2} and \defref{ss2gsos} that $\alpha\PO_s\beta$.
	\item There remains only the case when $(l(\alpha),l(\beta))\in ser$ and there are $\delta,\gamma\in\Sigma_t$ such that
		\[\left(\begin{array}{ll}
		&pos_t(\delta)<pos_t(\gamma) \wedge(l(\delta),l(\gamma))\notin ser\\
		\wedge&\alpha\,\sq_t^*\,\delta\,\sq_t^*\,\beta\wedge\alpha\,\sq_t^*\,\gamma\,\sq_t^*\,\beta
		\end{array}\right).\] Since $\sq_t\;=\;\sq_s$, we also have $\alpha\,\sq_s^*\,\delta\,\sq_s^*\,\beta\wedge\alpha\,\sq_s^*\,\gamma\,\sq_s^*\,\beta$. Since $(l(\delta),l(\gamma))\notin ser$, we either have $(l(\delta),l(\gamma))\in inl$ or $(l(\delta),l(\gamma))\notin ser\cup inl$.
		\begin{itemize}
		\item If $(l(\delta),l(\gamma))\in inl$, then $pos_s(\delta)\not= pos_s(\gamma)$. Thus, $(pos_s(\delta)<pos_s(\gamma) \wedge(l(\delta),l(\gamma))\notin ser)$ or $(pos_s(\gamma)<pos_s(\delta) \wedge(l(\gamma),l(\delta))\notin ser)$. So it follows from \eref{ext2xi.2} and \defref{ss2gsos} that $\alpha\PO_s\beta$.
		\item If $(l(\delta),l(\gamma))\notin inl$, then $(l(\delta),l(\gamma))\notin ser\cup inl$. Hence, by \defref{ss2gsos}, $\delta\PO_t\gamma$, which by Definitions \ref{def:ccl} and \ref{def:ss2gsos}, $\delta\PO\gamma$. But since $\lhd_s\in ext(G^{\{t\}})$, we have $\delta\lhd_s\gamma$, which implies $pos_s(\delta)< pos_s(\gamma)$. Since $pos_s(\delta)< pos_s(\gamma)$ and $(l(\delta),l(\gamma))\notin ser$, it follows from \eref{ext2xi.2} and \defref{ss2gsos} that $\alpha\PO_s\beta$.
		\end{itemize}
	\end{enumerate}

Thus, we have shown that $\alpha\PO_s\beta$. Hence,
\begin{align}
\PO_t\;\subseteq\;\PO_s \label{eq:ext2xi.sg3}
\end{align}

It remains to show that $\PO_s\;\subseteq\;\PO_t$. Let $\alpha\PO_s\beta$. Suppose that $\neg(\alpha\PO_t\beta)$. Since $\alpha\PO_s\beta$, by \defref{ss2gsos}, we need to consider three cases:
	\begin{enumerate}[(a)]
	\item When $(l(\alpha),l(\beta))\notin ser\cup inl$, we suppose that $\neg(\alpha\PO_t\beta)$. This by \defref{ss2gsos} implies that $pos_t(\beta)\le pos_t(\alpha)$. By Definitions \ref{def:ccl} and \ref{def:ss2gsos}, it follows that  $\beta\sq_t\alpha$ and $\beta\sq\alpha$. But since $\lhd_s\in ext(G^{\{t\}})$, we have $\beta\lhd_s^\frown\alpha$, which implies $pos_s(\beta)\le pos_s(\alpha)$, a contradiction.
	\item If $(\alpha,\beta)\in\; \com_s\cap\left(\si{\sq_s^*}\circ\reco{\com}_s\circ \si{\sq_s^*}\right)$, then since $\com_s=\com_t$ and $\sq_s=\sq_t$, we have $(\alpha,\beta)$ $\in\; \com_t\cap\left(\si{\sq_t^*}\circ\reco{\com}_t\circ \si{\sq_t^*}\right)$. Since $\alpha\com_t\beta$, we have $pos_t(\alpha)<pos_t(\beta)$ or $pos_t(\beta)<pos_t(\alpha)$. We claim that $pos_t(\alpha)<pos_t(\beta)$. Suppose for a contradict that $pos_t(\beta)<pos_t(\alpha)$. 
Since $(\alpha,\beta)\in\; \com_t\cap\left(\si{\sq_t^*}\circ\reco{\com}_t\circ \si{\sq_t^*}\right)$ and $\com_t$ is symmetric, we have $(\beta,\alpha)\in\; \com_t\cap\left(\si{\sq_t^*}\circ\reco{\com}_t\circ \si{\sq_t^*}\right)$. Hence, it follows from Definitions \ref{def:ccl} and \ref{def:ss2gsos} that $\beta\PO_t\alpha$ and $\beta\PO\alpha$. But since $\lhd_s\in ext(G^{\{t\}})$, we have $\beta\lhd_s\alpha$, which implies $pos_s(\beta)< pos_s(\alpha)$, a contradiction. Thus, $pos_t(\alpha)<pos_t(\beta)$. \\ Since $(\alpha,\beta)\in\; \com_t\cap\left(\si{\sq_t^*}\circ\reco{\com}_t\circ \si{\sq_t^*}\right)$, we get $\alpha\PO_t\beta$.
	\item There remains only the case when $(l(\alpha),l(\beta))\in ser$ and there are $\delta,\gamma\in\Sigma_s$ such that 
	\[\left(\begin{array}{ll}
		&pos_s(\delta)<pos_s(\gamma) \wedge(l(\delta),l(\gamma))\notin ser\\
		\wedge&\alpha\,\sq_s^*\,\delta\,\sq_s^*\,\beta\wedge\alpha\,\sq_s^*\,\gamma\,\sq_s^*\,\beta
		\end{array}\right).\]
	Since $\sq_s=\sq_t$, we have $\alpha\,\sq_t^*\,\delta\,\sq_t^*\,\beta$ and $\alpha\,\sq_t^*\,\gamma\,\sq_t^*\,\beta$, which by \defref{ss2gsos} and transitivity of $\le$ implies that $pos_t(\alpha)\le pos_t(\delta)\le pos_t(\beta)$ and $pos_t(\alpha)\le pos_t(\gamma)\le pos_t(\beta)$. Since $(l(\delta),l(\gamma))\notin ser$, we either have $(l(\delta),l(\gamma))\in inl$ or $(l(\delta),l(\gamma))\notin ser\cup inl$.
		\begin{enumerate}[(i)]
		\item If $(l(\delta),l(\gamma))\in inl$, then $pos_t(\delta)\not= pos_t(\gamma)$. This implies that $(pos_t(\delta)<pos_t(\gamma) \wedge(l(\delta),l(\gamma))\notin ser)$ or $(pos_t(\gamma)<pos_t(\delta) \wedge(l(\gamma),l(\delta))\notin ser)$. Since $pos_t(\delta)\not= pos_t(\gamma)$ and $pos_t(\alpha)\le pos_t(\delta)\le pos_t(\beta)$ and $pos_t(\alpha)\le pos_t(\gamma)\le pos_t(\beta)$, we also have $pos_t(\alpha)< pos_t(\beta)$.  So it follows from \defref{ss2gsos} that $\alpha\PO_t\beta$.
		\item If $(l(\delta),l(\gamma))\notin inl$, then $(l(\delta),l(\gamma))\notin ser\cup inl$. We want to show that $pos_t(\delta)<pos_t(\gamma)$. Suppose that $pos_s(\delta)\ge pos_s(\gamma)$. Since $(l(\delta),l(\gamma))\notin ser\cup inl$, by Definitions \ref{def:ccl} and \ref{def:ss2gsos}, we have $\gamma\sq_t\delta$ and $\gamma\sq\delta$. But since $\lhd_s\in ext(G{\{t\}})$, we have $\gamma\lhd_s^\frown\delta$, which implies $pos_s(\gamma)\le pos_s(\delta)$, a contradiction. Since $pos_t(\delta)<pos_t(\gamma)$ and  $pos_t(\alpha)\le pos_t(\delta)\le pos_t(\beta)$ and $pos_t(\alpha)\le pos_t(\gamma)\le pos_t(\beta)$, we have $pos_t(\alpha)<pos_t(\beta)$. Hence, we have $pos_t(\alpha)<pos_t(\beta)$ and 
		\[\left(\begin{array}{ll}
		&pos_t(\delta)<pos_t(\gamma) \wedge(l(\delta),l(\gamma))\notin ser\cup inl\\
		\wedge&\alpha\,\sq_t^*\,\delta\,\sq_t^*\,\beta\wedge\alpha\,\sq_t^*\,\gamma\,\sq_t^*\,\beta
		\end{array}\right).\] 
		So it follows that $\alpha\PO_t\beta$ by \defref{ss2gsos}.
		\end{enumerate}
	\end{enumerate}

Thus, we have shown $\PO_s\;\subseteq\;\PO_t$. This and \eref{ext2xi.sg3} imply $\PO_t\;=\;\PO_s$.\qed
\end{proof}

\label{body end}
\end{document}